\documentclass[11pt]{article}

\usepackage[utf8]{inputenc}

\usepackage{amsmath,amsfonts,amsthm,amssymb,color}
\usepackage[dvipsnames,svgnames,table]{xcolor}
\definecolor{darkgreen}{rgb}{0.0,0,0.9}
\usepackage{mathtools}
\usepackage{authblk}
\usepackage{fullpage}
\usepackage{multirow}
\usepackage{parskip}
\usepackage{comment}
\usepackage{tikz}
\usepackage{bbm}
\usepackage{todonotes}
\presetkeys{todonotes}{inline}{}
\usepackage{dsfont}
\usepackage[sc]{mathpazo}
\usepackage[basic]{complexity}
\usepackage{algorithm2e}
\usepackage[colorlinks=true,citecolor=OliveGreen,linkcolor=BrickRed,urlcolor=BrickRed,pdfstartview=FitH]{hyperref}
\usepackage[capitalize,nameinlink]{cleveref}
\usepackage{tcolorbox}
\usepackage{thm-restate}
\usepackage{natbib}
\usepackage{caption}
\usepackage{booktabs}
\newtcolorbox{wbox}
{
	colback  = white,
}

\newtheorem{theorem}{Theorem}
\newtheorem{lemma}{Lemma}
\newtheorem{corollary}{Corollary}
\newtheorem{definition}{Definition}

\newtheorem{proposition}[theorem]{Proposition}
\newtheorem{claim}[theorem]{Claim}

\newtheorem{question}[theorem]{Question}

\theoremstyle{definition}
\newtheorem{remark}[theorem]{Remark}

\DeclareMathOperator{\rot}{rot}

\DeclareMathOperator{\pre}{pre}

\title{Robust Stable Matchings: \\
Dealing with Changes in Preferences}

\author[1]{Rohith Reddy Gangam}
\author[2]{Tung Mai}
\author[3]{Nitya Raju\footnote{This work was done while the author was a student at the University of California, Irvine.}}
\author[1]{Vijay V.~Vazirani}

\affil[1]{University of California, Irvine}
\affil[2]{Adobe Research}
\affil[3]{University of Maryland, College Park}

\date{}

\begin{document}
\maketitle

\begin{abstract}

We study stable matchings that are robust to preference changes in the two-sided stable matching setting of Gale and Shapley~\cite{GaleS}. Given two instances $A$ and $B$ on the same set of agents, a matching is said to be robust if it is stable under both instances. This notion captures desirable robustness properties in matching markets where preferences may evolve, be misreported, or be subject to uncertainty. While the classical theory of stable matchings reveals rich lattice, algorithmic, and polyhedral structure for a single instance, it is unclear which of these properties persist when stability is required across multiple instances. Our work initiates a systematic study of the structural and computational behavior of robust stable matchings under increasingly general models of preference changes.

We analyze robustness under a hierarchy of perturbation models:
\begin{enumerate}
    \item a single upward shift in one agent’s preference list,
    \item an arbitrary permutation change by a single agent, and
    \item arbitrary preference changes by multiple agents on both sides.
\end{enumerate}

For each regime, we characterize when:
\begin{enumerate}
    \item the set of robust stable matchings forms a sublattice,
    \item the lattice of robust stable matchings admits a succinct Birkhoff partial order enabling efficient enumeration,
    \item worker-optimal and firm-optimal robust stable matchings can be computed efficiently, and
    \item the robust stable matching polytope is integral (by studying its LP formulation).
\end{enumerate}

We provide explicit counterexamples demonstrating where these structural and geometric properties break down, and complement these results with XP-time algorithms running in $O(n^{k})$ time, parameterized by $k$, the number of agents whose preferences change. Our results precisely delineate the boundary between tractable and intractable cases for robust stable matchings.

\end{abstract}

\section{Introduction}
\label{sec:introduction}

Matchings under preferences is a central topic in market design, originating with the seminal 1962 paper by Gale and Shapley~\cite{GaleS}, which introduced the stable matching problem and launched the broader study of matching-based market design. The classical stable matching model assumes that agents report fixed preference rankings, yet real markets rarely satisfy such invariance: preferences may shift as new information arrives, external conditions change, or strategic behavior—including misreports or collusion—arises. These factors naturally lead to the question of how matchings behave under perturbations of the input. In this work, we investigate matchings that remain stable across multiple preference profiles—specifically, matchings that are stable in both a base instance~$A$ and a perturbed instance~$B$, which we refer to as \emph{robust stable matchings}. Our study provides a systematic analysis of their structural, algorithmic, and geometric properties under increasingly general models of perturbations.

A central structural feature of the stable matching problem is that the set of stable matchings of an instance forms a finite distributive lattice, ordered by the preferences of one side~\cite{Knuth-book, GusfieldI}. This lattice admits a compact representation via rotations, which correspond to minimal transformations between stable matchings and underpin many fundamental algorithmic and structural results. We investigate how this lattice structure behaves under preference changes by comparing the lattices associated with two instances defined on the same set of agents. In particular, we study whether the intersection of their stable matchings retains any of the lattice structure that characterizes the classical setting.

We begin with the simplest nontrivial perturbation, namely a single shift in the preference list of one agent, and establish a complete structural characterization of the resulting space of stable matchings. We then consider arbitrary permutations by a single agent and subsequently extend the analysis to instances in which multiple workers and firms simultaneously modify their preferences. These models progressively weaken the lattice structure of the robust stable matchings and expose the precise conditions under which it is preserved or lost.

Alongside these structural results, we investigate the algorithmic and geometric properties of robust stable matchings. We address the complexity of deciding whether the intersection of the stable matchings of two instances is nonempty, the efficient computation of worker-optimal and firm-optimal matchings within this intersection, and the existence of succinct Birkhoff representations supporting efficient enumeration of all robust stable matchings. From a geometric perspective, we analyze the fractional robust stable matching polytope and determine exactly when it is integral. We show that the breakdown of integrality coincides with the loss of lattice structure, yielding a sharp and unifying threshold for both structural and geometric properties.

Preliminary versions of the results in this paper appeared in three earlier works~\cite{MV.robust, GMRV-nearby-instances, GMRV-general-instance}, each addressing robustness under progressively more general models of preference perturbations. The work in~\cite{MV.robust} studied robustness under a single upward shift in one agent’s preference list, establishing lattice structure and polynomial-time algorithms. This was extended in~\cite{GMRV-nearby-instances} to arbitrary permutations by a single agent, together with efficient computation of the associated Birkhoff partial order. The most general setting, in which multiple agents on both sides may arbitrarily change their preferences, was considered in~\cite{GMRV-general-instance}, which identified sharp thresholds for lattice structure and polytope integrality and developed parameterized algorithms for the general case. The present paper unifies, refines, and extends these results into a single coherent framework.

\subsection{Our contributions}
\label{sec:contributions}

Throughout, we consider stable matching instances on $n$ workers and $n$ firms, each with strict and complete preference lists over the opposite side. For an instance~$A$, let $\mathcal{M}_A$ denote the set of its stable matchings and $\mathcal{L}_A$ the finite distributive lattice of these matchings under the usual dominance order (see \Cref{sec:latticeOfSM} for details). We begin with the simplest change, in which a new instance~$B$ is obtained from~$A$ by a single upward shift in the preference list of one agent. In this setting, we give a complete description of how the shift operation splits the lattice of stable matchings: both $\mathcal{M}_A \cap \mathcal{M}_B$ and $\mathcal{M}_A \setminus \mathcal{M}_B$ form sublattices of $\mathcal{L}_A$.

We next extend our analysis to a more general model in which a single agent, whether a worker or a firm, may arbitrarily permute its preference list between instances $A$ and $B$. In this setting, while $\mathcal{M}_A \cap \mathcal{M}_B$ remains a sublattice, the set $\mathcal{M}_A \setminus \mathcal{M}_B$ is no longer guaranteed to form a sublattice of $\mathcal{L}_A$. Nevertheless, some structure persists: it forms a semi-sublattice of $\mathcal{L}_A$. This key property allows us to identify and succinctly characterize all matchings that remain stable across both instances by efficiently characterizing the Birkhoff partial order on rotations (see \Cref{sec:birkhoff-rotations} for definitions).

A central theme in these results is the structural relationship between the lattices of two instances. This leads naturally to the following question, whose importance is underscored by the fact that stable matching lattices are \emph{universal} among all finite distributive lattices: for every finite distributive lattice~$\mathcal{L}$, there exists a stable matching instance~$A$ whose lattice~$\mathcal{L}_A$ is isomorphic to~$\mathcal{L}$. Consequently, an affirmative answer to the question below would yield a fundamental new property for the entire class.

\begin{question}
\label{ques.one}
Is $\mathcal{M}_A \cap \mathcal{M}_B$ always a sublattice of $\mathcal{L}_A$ and $\mathcal{L}_B$ under arbitrary preference changes?
\end{question}

We show that, in general, the answer is ``no.'' Consider a robust stable matching instance $(A,B)$ in which $p$ of the $n$ workers and $q$ of the $n$ firms change their preferences from $A$ to $B$. In this regime, we identify a sharp structural frontier: $\mathcal{M}_A \cap \mathcal{M}_B$ forms a sublattice of both $\mathcal{L}_A$ and $\mathcal{L}_B$ precisely when $p \leq 1$ or $q \leq 1$, and we provide an example showing that this property fails beginning at $(p,q)=(2,2)$.

In addition to these structural results, we address several fundamental computational and geometric questions concerning robust stable matchings:
\begin{enumerate}
    \item Can we efficiently decide whether $\mathcal{M}_A \cap \mathcal{M}_B \neq \emptyset$?
    \item If $\mathcal{M}_A \cap \mathcal{M}_B$ forms a lattice, can the worker-optimal and firm-optimal robust stable matchings (see \Cref{sec:latticeOfSM} for definitions) be efficiently computed?
    \item Does $\mathcal{M}_A \cap \mathcal{M}_B$ admit a succinct Birkhoff partial order and support efficient enumeration?
    \item Is the fractional robust stable matching polytope integral, enabling linear programming (LP) based algorithms?
\end{enumerate}

For all settings in which the sublattice property holds, we also provide a Deferred Acceptance–based algorithm to compute the worker-optimal and firm-optimal matchings in the robust lattice. We show that the fractional robust stable matching polytope remains integral when $p \leq 1$ or $q \leq 1$. For $(p,q)=(2,2)$, we provide an example in which the polytope contains fractional vertices. This yields a sharp threshold that mirrors exactly the boundary at which lattice structure breaks down, and rules out LP-based approaches in the fully general case. To complement this, we give an XP algorithm running in $O(n^{p+q+2})$ time that decides whether a robust stable matching exists, constructs one when it does, and enumerates all such matchings with the same delay bound. A summary of our results is provided in \Cref{fig:summary_table}.

\begin{table}[ht]
\centering
% \captionsetup{justification=centering}
\renewcommand{\arraystretch}{1.5}
\begin{tabular}{|c|c|c|c|}
\hline
\textbf{$(p, q)$} & \textbf{Computation} & \textbf{Structure} & \textbf{Geometry} \\ 
& \textbf{(P?)} & \textbf{(Lattice?)} & \textbf{(Integral?)} \\ \midrule
% \hline
(0,0) & P~\cite{GaleS} & Yes~\cite{irving2} & Yes~\cite{Teo-Sethuraman-Lp} \\
\hline
(0,1) & P~[Thm.~\ref{thm:n_1_algo_works}] & Yes~[Thm.~\ref{thm:sublattice}] & Yes[Thm.~\ref{thm:n_1_integral_polytope}] \\
(0,n) & P~[Thm.~\ref{thm:n_1_algo_works}] & Yes~~[Thm.~\ref{thm:sublattice}] & Yes~[Thm.~\ref{thm:n_1_integral_polytope}] \\
(1,1) & P~[Thm.~\ref{thm:n_1_algo_works}] & \cellcolor{yellow!25}Yes~[Thm.~\ref{thm:sublattice_two_side}] & Yes~[Thm.~\ref{thm:n_1_integral_polytope}] \\
(1,n) & P~[Thm.~\ref{thm:n_1_algo_works}] & \cellcolor{yellow!25}Yes~[Thm.~\ref{thm:sublattice_two_side}] & Yes~[Thm.~\ref{thm:n_1_integral_polytope}] \\
(2,2) & P~[Thm.~\ref{thm:robust_xp}] & No~[Thm.~\ref{thm:both_sides_not_sublattice}] & No~[Thm.~\ref{thm:2_2_not_integral_polytope}] \\
\hline
($p$,$q$) & XP: $O(n^{p+q+2})$~[Thm.~\ref{thm:robust_xp}] & - & - \\
\hline
($n$,$n$) & NP-Complete~\cite{NP-Two-stable} & No~[Thm.~\ref{thm:both_sides_not_sublattice}] & No~[Thm.~\ref{thm:2_2_not_integral_polytope}] \\
\hline
\end{tabular}
\vspace{0.5em}
\caption{
Summary of our results on robust stable matchings. 
$p$ workers and $q$ firms permute their preference lists from instance $A$ to $B$.\\
The ``Computation'' column deals with the decision problem $\mathcal{M}_A \cap \mathcal{M}_B \neq \varnothing $. \\
The ``Structure'' column answers if $\mathcal{M}_A \cap \mathcal{M}_B$ is a sublattice of $\mathcal{L}_A$ and $\mathcal{L}_B$. \\
The ``Geometry'' column answers if the robust fractional stable matching polytope is integral. \\
Finding the Birkhoff partial order differs between the $(0,1)$ and $(0,n)$ cases, since $\mathcal{M}_A \setminus \mathcal{M}_B$ forms a semi-sublattice in the former but not in the latter. The problem of explicitly constructing the Birkhoff partial order for the $(1,1)$ and $(1,n)$ cases remains open; however, we show that polynomial-time solvability of one case implies polynomial-time solvability of the other.
}
\label{fig:summary_table}
\end{table}

\section{Related Work}
\label{sec:related_work}

The stable matching problem was introduced by Gale and Shapley~\cite{GaleS}, who also presented the Deferred Acceptance (DA) algorithm, which produces a stable matching that is optimal for one side of the market and pessimal for the other. Over the years, a rich structural, algorithmic, and game-theoretic theory has developed around this problem, leading to efficient algorithms and influential applications such as residency assignment and school choice systems~\cite{GusfieldI, Roth, NYC-school, Boston}. Knuth~\cite{knuth1976marriages} showed that the family of all stable matchings forms a finite distributive lattice, a foundational result that underlies many subsequent developments. Additional game-theoretic properties, including the Rural Hospitals Theorem and incentive compatibility of the DA algorithm, were established in~\cite{Roth-IC-1982economics, DubinsF}. On the geometry of stable matchings, Teo and Sethuraman~\cite{Teo-Sethuraman-Lp} proved integrality of the stable matching polytope, with alternative formulations given in~\cite{VANDEVATE, Roth85}. The foundational importance of this body of work was recognized through the 2012 Nobel Prize in Economics awarded to Roth and Shapley~\cite{RS}, and today the stable matching problem forms a cornerstone of algorithmic game theory; see \cite{Knuth, GusfieldI, Manlove-book, MM.Book-Online} for comprehensive treatments.

Several works have also examined robustness and stability under more general or adversarial preference changes. Chen et al.~\cite{Chen-Matchings-Under-preferences} studied a notion of robustness based on swap distance between preference profiles, defining $d$-robust matchings that remain stable under bounded perturbations and analyzing the trade-off between stability and social welfare. In a more general setting where agents may arbitrarily alter their preferences, Miyazaki and Okamoto~\cite{NP-Two-stable} considered the problem of finding a matching that is stable across two instances, which they term a jointly stable matching, and showed that deciding the existence of such a matching is NP-hard. These results highlight the computational challenges that arise when robustness is required under unrestricted preference changes.

% Preliminary versions of the results presented in this paper appeared in~\cite{MV.robust, GMRV-nearby-instances, GMRV-general-instance}. The notion of robustness that we study was introduced in~\cite{MV.robust}, where a matching is defined to be robust if it is stable in both an original instance and a changed instance. That work provided polynomial-time algorithms for scenarios in which a single agent changes preferences through simple shifts. In~\cite{GMRV-nearby-instances}, this analysis was extended to arbitrary changes by a single agent. Chen et al.~\cite{Chen-Matchings-Under-preferences} investigated a related but distinct robustness measure based on swap distance, defining $d$-robust matchings and exploring the trade-off between stability and social welfare. When all agents may arbitrarily alter their preferences, Miyazaki and Okamoto~\cite{NP-Two-stable} showed that determining whether there exists a matching that is stable across both instances---which they term a jointly stable matching---is NP-hard.

% The work in~\cite{GMRV-general-instance} further advanced this direction by allowing multiple agents on both sides of the market to permute their preference lists arbitrarily. Whereas prior approaches primarily leveraged the lattice of stable matchings, rotation posets, or variants of DA, this line of research additionally incorporates linear programming methods, which become essential when lattice structure alone does not suffice.

When only one side of the market changes preferences, the setting becomes closely related to stable matching with weak or partially ordered preferences. Several refinements of stability---including weak, strong, and super-strong stability---have been studied in this context. Irving~\cite{IRVING-Indifferences} proved the existence of weakly stable matchings and provided polynomial-time algorithms to compute them. Spieker~\cite{SPIEKER-indifference-lattice} showed that the set of super-stable matchings forms a distributive lattice, and Manlove~\cite{MANLOVE-indifference-lattice} extended this to the strongly stable case. Kunysz et al.~\cite{Kunysz-ssm} gave a succinct partial-order representation of all strongly stable matchings.

Another related direction concerns uncertainty or partial information in preferences. Aziz et al.~\cite{aziz1, aziz2} proposed several uncertainty models and studied matchings that remain stable across all realizations. Genc et al.~\cite{genc2, genc1} introduced $(x,y)$-supermatches, which remain easy to repair after disruptions. Although not robust in the sense considered here, these notions capture complementary resilience properties.

Beyond robustness, numerous generalizations of the stable matching framework have been explored. Chen et al.~\cite{NP-Two-stable-Chen} studied multi-criterion (multi-modal) preferences, while Menon and Larson~\cite{Menon_Larson} examined minimizing the maximum number of blocking pairs across all completions of weak orders. Incremental preference changes have also been investigated: Bredereck et al.~\cite{BredereckCKLN20} and Boehmer et al.~\cite{boehmer_incremental_mfcs, boehmer_incremental_aaai} analyzed algorithms and structural properties for maintaining stability as preferences evolve over time.

Robustness has also been considered in the context of popular matchings. Introduced by G\"{a}rdenfors ~\cite{Gard75a}, popularity compares matchings via majority voting. Robust versions have recently been studied by Bullinger et al.~\cite{robust_pop_matchings_BGS}, who obtained polynomial-time and hardness results based on the extent of preference changes. Cs\'{a}ji~\cite{Csaj24a} further extended these ideas to settings with multi-modal and uncertain preferences.

Finally, robustness to input changes has a long history in computational social choice. A significant body of work studies how small perturbations in votes affect election outcomes~\cite{FaRo15a, SYE13a, BFK+21a}, often through swap-based distance measures~\cite{EFS09a}. Similar ideas have been applied to stable matchings: Boehmer et al.~\cite{BBHN21a} analyzed how strategic changes can enforce or prevent particular matchings, and B\'{e}rczi et al.~\cite{BERCZI} investigated how preference modifications can be used to guarantee the existence of stable matchings with desired structural properties.

\section{Model and Preliminaries}
\label{sec:prelim}

\subsection{The stable matching problem and the robust stable matching problem}
\label{sec:sm-and-rsm}

A stable matching problem instance consists of a set of $n$ workers,
$\mathcal{W} = \{w_1, w_2, \ldots, w_n\}$, and a set of $n$ firms,
$\mathcal{F} = \{f_1, f_2, \ldots, f_n\}$, collectively referred to as \emph{agents}.
Each agent $a \in \mathcal{W} \cup \mathcal{F}$ has a strict total order (preference list), denoted $>_a$,
over the agents of the opposite type.
For instance, $w_i <_f w_j$ indicates that firm $f$ strictly prefers $w_j$ to $w_i$,
and worker preferences are expressed analogously.

A matching $M$ is a one-to-one correspondence between $\mathcal{W}$ and $\mathcal{F}$.
For each pair $(w,f)\in M$, we write $M(w)=f$ and $M(f)=w$ and call them partners under $M$.
A pair $(w,f)\notin M$ is a \emph{blocking pair} for $M$ if $f >_w M(w)$ and $w >_f M(f)$.
A matching is \emph{stable} if it has no blocking pair.

A \emph{robust stable matching instance} consists of $k \ge 2$ stable matching problem instances
$A_1, A_2, \ldots, A_k$ defined on the same agent sets $\mathcal{W} \text{ and } \mathcal{F}$.
A matching $M$ is \emph{robust stable} under $(A_1,\ldots,A_k)$ if $M$ is stable under every $A_i$.
We focus primarily on $k=2$ and write $(A,B)$.
For any instance $I$, let $\mathcal{M}_I$ denote the set of stable matchings and $\mathcal{L}_I$ the lattice
of stable matchings under $I$ (defined formally in \Cref{sec:latticeOfSM}).
Our main object of study is the set of robust stable matchings, $\mathcal{M}_A \cap \mathcal{M}_B$.

We will also use the following notion of ``distance'' between instances.
We say that a pair $(A,B)$ is of type $(p,q)$ if the preferences in $B$ differ from those in $A$ for at most $p$ workers and $q$ firms, and are
identical for the remaining $2n-p-q$ agents.
(Results are symmetric in $p$ and $q$, so we look at the cases where $0\leq p\leq q \leq n$.)

\subsection{The lattice of stable matchings}
\label{sec:latticeOfSM}

Let $M$ and $M'$ be two stable matchings under the same instance.
We say that $M$ \emph{dominates} $M'$, denoted $M \preceq M'$, if every worker weakly prefers
its partner in $M$ to its partner in $M'$.
Equivalently, $M$ is a \emph{predecessor} of $M'$ in the dominance order.

A stable matching $M$ is a \emph{common predecessor} of $M_1$ and $M_2$ if $M \preceq M_1$ and
$M \preceq M_2$. It is a \emph{lowest common predecessor} if it is a common predecessor and no
other common predecessor strictly dominates it.
Analogously, one defines \emph{successor} and \emph{highest common successor}.

The dominance partial order on stable matchings forms a lattice~\cite{GusfieldI}.
For any two stable matchings $M_1, M_2$, their meet $M_1 \wedge M_2$ (the lowest common predecessor)
and join $M_1 \vee M_2$ (the highest common successor) are unique.
Moreover, $M_1 \wedge M_2$ is obtained by assigning each worker its more preferred partner
from $M_1$ and $M_2$, and $M_1 \vee M_2$ is obtained by assigning each worker its less preferred
partner from $M_1$ and $M_2$; both are stable.
These operations satisfy the distributive laws: for any stable $M, M', M''$,
$$
    M \vee (M' \wedge M'') = (M \vee M') \wedge (M \vee M'') 
$$
$$   M \wedge (M' \vee M'') = (M \wedge M') \vee (M \wedge M'').
$$

The lattice contains a unique matching $M_0$ that dominates all others and a unique matching $M_z$
that is dominated by all others.
$M_0$ is the \emph{worker-optimal} (firm-pessimal) stable matching and $M_z$ is the
\emph{firm-optimal} (worker-pessimal) stable matching.

\subsection{Birkhoff's Theorem and rotations}
\label{sec:birkhoff-rotations}

\begin{definition}
\label{def:closedset}
A \emph{closed set} (or \emph{lower set}) of a poset is a subset $S$ such that if $x \in S$ and $y \prec x$,
then $y \in S$.
\end{definition}

For a poset $\Pi$, the family of closed sets of $\Pi$ is closed under union and intersection and forms a
distributive lattice, denoted $L(\Pi)$, with join and meet given by union and intersection.
Birkhoff's theorem~\cite{Birkhoff,Stanley} states that for any finite distributive lattice $\mathcal{L}$, there exists a
poset $\Pi$ such that $L(\Pi) \cong \mathcal{L}$; we say that $\Pi$ \emph{generates} $\mathcal{L}$.

\begin{theorem}[Birkhoff~\cite{Birkhoff}]
\label{thm:birkhoff}
Every finite distributive lattice $\mathcal{L}$ is isomorphic to $L(\Pi)$ for some finite poset $\Pi$.
\end{theorem}

For stable matching lattices, the generating poset can be taken to be the poset of \emph{rotations}
introduced by Irving~\cite{irving,irving2}; see also~\cite{GusfieldI}.
A rotation is a minimal cyclic transformation that maps one stable matching to another.

Let $M$ be a stable matching.
For a worker $w$, let $s_M(w)$ denote the first firm $f$ on $w$'s preference list such that
$f$ strictly prefers $w$ to its partner in $M$.
Let $next_M(w)$ denote the $M$-partner of firm $s_M(w)$.
A \emph{rotation} $\rho$ \emph{exposed} in $M$ is an ordered list of pairs
$\rho = \{w_0f_0, w_1f_1, \ldots, w_{r-1}f_{r-1}\}$
such that for each $i$, $w_{i+1} = next_M(w_i)$, where indices are taken modulo $r$.
Define $M/\rho$ to be the matching obtained by keeping all pairs not in $\rho$ unchanged and
matching each $w_i$ to $f_{i+1} = s_M(w_i)$.
It is known that $M/\rho$ is also stable, and the transformation from $M$ to $M/\rho$ is called
the \emph{elimination} of $\rho$.

\begin{lemma}[\cite{GusfieldI}, Theorem~2.5.4]
\label{lem:seqElimination}
Every rotation appears exactly once in any sequence of eliminations from $M_0$ to $M_z$.
\end{lemma}

Let $\rho = \{w_0f_0, \ldots, w_{r-1}f_{r-1}\}$ be a rotation.
For $0 \le i \le r-1$, we say that $\rho$ \emph{moves} $w_i$ from $f_i$ to $f_{i+1}$ and
moves $f_i$ from $w_i$ to $w_{i-1}$.
If a firm $f$ is either $f_i$ or lies strictly between $f_i$ and $f_{i+1}$ in $w_i$'s list, we say that
$\rho$ \emph{moves $w_i$ below $f$}.
Similarly, $\rho$ \emph{moves $f_i$ above $w$} if $w$ is either $w_i$ or lies between $w_i$ and $w_{i-1}$
in $f_i$'s list.

\subsection{The rotation poset}
\label{sec:rotation-poset}

A rotation $\rho'$ is said to \emph{precede} another rotation $\rho$, denoted $\rho' \prec \rho$,
if $\rho'$ is eliminated in every sequence of eliminations from $M_0$ to any stable matching in which
$\rho$ is exposed.
This precedence relation defines a partial order on the set of rotations, called the \emph{rotation poset},
denoted $\Pi$.

\begin{lemma}[\cite{GusfieldI}, Lemma~3.2.1]
\label{lem:uniqueMoveRotation}
For any worker $w$ and firm $f$, there is at most one rotation that moves $w$ to $f$, moves $w$ below $f$,
or moves $f$ above $w$. Moreover, if $\rho_1$ moves $w$ to $f$ and $\rho_2$ moves $w$ from $f$ then
$\rho_1 \prec \rho_2$.
\end{lemma}

\begin{lemma}[\cite{GusfieldI}, Lemma~3.3.2]
\label{lem:computePoset}
The rotation poset $\Pi$ contains at most $O(n^2)$ rotations and can be computed in polynomial time.
\end{lemma}

There is a one-to-one correspondence between stable matchings and closed subsets of $\Pi$:
given a closed set $S$, eliminating the rotations in $S$ starting from $M_0$ in any topological order
consistent with the order $\Pi$ yields a stable matching $M(S)$, and every stable matching arises in this way.
We say that $S$ \emph{generates} $M(S)$ and that $\Pi$ \emph{generates} the stable matching lattice.

The \emph{Hasse diagram} of a poset is a directed graph with a vertex for each element and an edge
$x \to y$ if $x \prec y$ and there is no $z$ with $x \prec z \prec y$.

\subsection{Sublattice and semi-sublattice}
\label{sec:sublattice_defn}

Let $\mathcal{L}$ be a distributive lattice with join and meet operations $\vee$ and $\wedge$.

A \emph{sublattice} $\mathcal{L}' \subseteq \mathcal{L}$ is a subset such that for any $x,y \in \mathcal{L}'$, both
$x \vee y \in \mathcal{L}'$ and $x \wedge y \in \mathcal{L}'$.
A \emph{join semi-sublattice} $\mathcal{L}_j \subseteq \mathcal{L}$ is a subset such that for any $x,y \in \mathcal{L}_j$,
$x \vee y \in \mathcal{L}_j$.
Similarly, a \emph{meet semi-sublattice} $\mathcal{L}_m \subseteq \mathcal{L}$ is a subset such that for any
$x,y \in \mathcal{L}_m$, $x \wedge y \in \mathcal{L}_m$.
Note that $\mathcal{L}'$ is a sublattice if and only if it is both a join and a meet semi-sublattice.

\subsection{Linear programming formulation}
\label{sec:lp_intro}

\begin{figure}[htp]
    \centering
    \includegraphics[width=8cm]{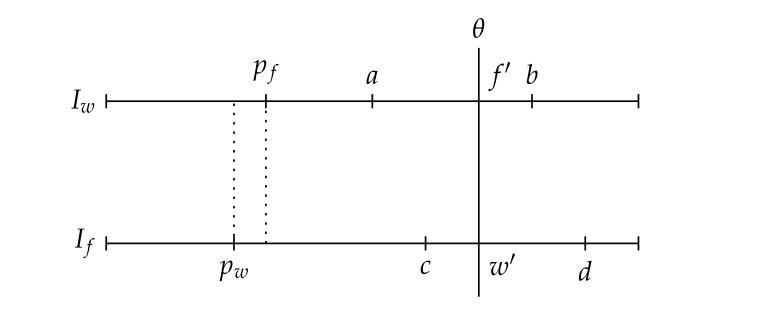}
    \caption{Finding an integral stable matching from a fractional solution (Image from~\cite{MM.Book-Online})}
    \label{fig:1}
\end{figure} 

A well-known result about the stable matching problem is that it admits a linear programming formulation in which all vertices are integral. This yields an efficient method to compute stable matchings. The LP formulation is given below as LP~\eqref{eq:lp_single}. 

The third constraint disallows blocking pairs: for any worker--firm pair $(w,f)$, it ensures that if $w$ is matched (fractionally) to firms less preferred than $f$, then $f$ must be matched to workers more preferred than $w$. The remaining constraints ensure that $x$ defines a fractional perfect matching.

\begin{equation}
\label{eq:lp_single}
    \begin{aligned}[b]
        \text{maximize} \quad & 0 \\
        \text{subject to} \quad 
        & \sum_w x_{wf} = 1 \quad && \forall f \in \mathcal{F}, \\
        & \sum_f x_{wf} = 1 \quad && \forall w \in \mathcal{W}, \\
        & \sum_{f >^A_w f'} x_{wf'} - \sum_{w' >^A_f w} x_{w'f} \leq 0 \quad && \forall w \in \mathcal{W},\ \forall f \in \mathcal{F}, \\
        & x_{wf} \geq 0 \quad && \forall w \in \mathcal{W},\ \forall f \in \mathcal{F}.    
    \end{aligned}
\end{equation}

A solution $x$ to LP~\eqref{eq:lp_single} can be converted into an integral perfect matching as follows. Construct $2n$ unit intervals, each corresponding to an agent. For each worker $w$, divide its interval $I_w$ into $n$ subintervals of lengths $x_{wf}$ (note that $\sum_f x_{wf} = 1$), ordered from $w$'s most preferred firm to its least. Apply the same process to each firm $f$'s interval $I_f$, but order the subintervals from $f$'s least to most preferred worker.

Choose a value $\theta$ uniformly at random from $[0,1]$, and identify the subinterval in each agent’s interval that contains $\theta$. The probability that $\theta$ lies exactly on a boundary is zero, so we may ignore this case.

Let $\mu_{\theta} : \mathcal{W} \rightarrow \mathcal{F}$ denote the firm corresponding to the subinterval containing $\theta$ in each worker's interval, and let $\mu'_{\theta} : \mathcal{F} \rightarrow \mathcal{W}$ denote the analogous mapping for firms. These functions define a perfect matching and are inverses of each other. Moreover, they are stable for any $\theta \in [0,1]$.

\begin{lemma}[\cite{Teo-Sethuraman-Lp}]
\label{lem:lp_perfect} 
If $\mu_{\theta}(w) = f$ then $\mu'_{\theta}(f) = w$.
\end{lemma}

\begin{lemma}[\cite{Teo-Sethuraman-Lp}]
\label{lem:lp_stable} 
For each $\theta \in [0,1]$, the matching $\mu_{\theta}$ is stable.
\end{lemma}

\begin{theorem}[\cite{VANDEVATE}]
\label{thm:lp_polytope}
The stable matching polytope defined by LP~\eqref{eq:lp_single} has integral vertices; that is, it is the convex hull of stable matchings.
\end{theorem}

\section{Results}
\label{sec:results}

In this section, we present our results. We use $\mathcal{M}_X$ to denote the set of stable matchings of instance $X$ and $\mathcal{L}_X$ to denote the corresponding lattice. We primarily study robust stable matchings under two stable matching instances, although the results extend to any $k \geq 2$ instances.

Consider two stable matching instances $A$ and $B$ defined on a set of $n$ workers, $\mathcal{W} = \{w_1, w_2, \ldots, w_n\}$, and a set of $n$ firms, $\mathcal{F} = \{f_1, f_2, \ldots, f_n\}$. We first consider the setting of \Cref{thm:sublattice}, which explains the structure of robust stable matchings when $A$ and $B$ \emph{differ only at agents on one side}.

\begin{theorem}
\label{thm:sublattice}
Let $A$ be an instance of stable matching and let $B$ be another instance obtained from $A$ by changing the preference lists of agents on only one side, either workers or firms, but not both. Then the matchings in $\mathcal{M}_A \cap \mathcal{M}_B$ form a sublattice of each of the two lattices.
\end{theorem}

\begin{proof}
It suffices to show that $\mathcal{M}_A \cap \mathcal{M}_B$ is a sublattice of $\mathcal{L}_A$. Assume $|\mathcal{M}_A \cap \mathcal{M}_B| > 1$ and let $M_1$ and $M_2$ be two distinct matchings in $\mathcal{M}_A \cap \mathcal{M}_B$.

Let $\vee_A$ and $\vee_B$ denote the join operations under instances $A$ and $B$, respectively. Likewise, let $\wedge_A$ and $\wedge_B$ denote the meet operations under $A$ and $B$.

By definition of the join operation (see \Cref{sec:latticeOfSM}), $M_1 \vee_A M_2$ is the matching obtained by assigning each worker its less preferred partner from $M_1$ and $M_2$ according to instance $A$.
Without loss of generality, assume that $B$ is obtained from $A$ by changing only the firms’ preference lists. Since the preference list of each worker is identical in $A$ and $B$, its less preferred partner from $M_1$ and $M_2$ is the same in both instances.
Therefore, $M_1 \vee_A M_2 = M_1 \vee_B M_2$. A similar argument shows that $M_1 \wedge_A M_2 = M_1 \wedge_B M_2$.

Hence, both $M_1 \vee_A M_2$ and $M_1 \wedge_A M_2$ belong to $\mathcal{M}_A \cap \mathcal{M}_B$, as required.
\end{proof}

\begin{corollary}
\label{cor.sublatticeIntersection}
Let $A$ be an instance of stable matching and let $B_1, \ldots, B_k$ be instances obtained from $A$, each by changing the preference lists of agents on only one side, either workers or firms, but not both. Then the matchings in $\mathcal{M}_A \cap \mathcal{M}_{B_1} \cap \cdots \cap \mathcal{M}_{B_k}$ form a sublattice of $\mathcal{L}_A$.
\end{corollary}

\begin{proof}
Assume $|\mathcal{M}_A \cap \mathcal{M}_{B_1} \cap \cdots \cap \mathcal{M}_{B_k}| > 1$ and let $M_1$ and $M_2$ be two distinct matchings in this intersection. Then $M_1$ and $M_2$ belong to $\mathcal{M}_A \cap \mathcal{M}_{B_i}$ for each $1 \le i \le k$. By \Cref{thm:sublattice}, each $\mathcal{M}_A \cap \mathcal{M}_{B_i}$ is a sublattice of $\mathcal{L}_A$. Hence, both $M_1 \vee_A M_2$ and $M_1 \wedge_A M_2$ belong to $\mathcal{M}_A \cap \mathcal{M}_{B_i}$ for every $i$, and therefore to their intersection.
\end{proof}

These results show that, in certain cases, the set of robust stable matchings forms a sublattice of the lattice of stable matchings of the original instance. This motivates us to further characterize sublattices of the stable matching lattice. In \Cref{subsection:semisublatticeNecessarySufficient}, we show that for any instance $B$ obtained by permuting the preference list of a single worker or a single firm, the set $\mathcal{M}_A \setminus \mathcal{M}_B$ forms a semi-sublattice of $\mathcal{L}_A$ (see \Cref{lem:MABsemi}). In particular, if a worker’s preference list is permuted, $\mathcal{M}_A \setminus \mathcal{M}_B$ forms a join semi-sublattice of $\mathcal{L}_A$, and if a firm’s list is permuted, it forms a meet semi-sublattice of $\mathcal{L}_A$. In both cases, $\mathcal{M}_A \cap \mathcal{M}_B$ is a sublattice of both $\mathcal{L}_A$ and $\mathcal{L}_B$, as shown in \Cref{thm:sublattice}.

With this motivation, we first study Birkhoff’s Theorem on sublattices in \Cref{sec:generalization}. We establish properties of distributive lattices, their underlying partial orders, and how sublattices can be generated via \emph{compressions} of partial orders. When the lattice can be decomposed into two sublattices, or into a sublattice and a semi-sublattice, we design algorithms to efficiently describe these structures. We then analyze the structural properties of robust stable matchings under increasingly complex models in \Cref{sec:structure}. \Cref{sec:optimal_matchings} presents algorithms for computing worker-optimal and firm-optimal stable matchings when the robust stable matchings form a sublattice. \Cref{sec:polytope} studies the geometric properties of the robust stable matching polytope. Finally, \Cref{sec:robust_general_algo} provides an XP-time algorithm to compute and enumerate robust stable matchings in the most general setting, where arbitrarily many agents may change their preferences via arbitrary permutations between $A$ and $B$.

\subsection{Birkhoff's Theorem on Sublattices}
\label{sec:generalization}

Let $\Pi$ be a finite poset. For simplicity of notation, in this paper we will assume that $\Pi$ must
have {\em two dummy elements} $s$ and $t$; the remaining elements will be called {\em proper
elements} and the term {\em element} will refer to proper as well as dummy elements. The element
$s$ precedes all other elements and $t$ succeeds all other elements in $\Pi$. A {\em proper
closed set} of $\Pi$ is any closed set that contains $s$ and does not contain $t$.
It is easy to see that the set of all proper closed sets of $\Pi$ form a distributive 
lattice under the operations of set intersection and union. We will denote this lattice by 
$L(\Pi)$. The following has also been called {\em the fundamental theorem for finite distributive lattices}.

\begin{theorem}
(Birkhoff \cite{Birkhoff})	Every finite distributive lattice $\mathcal{L}$ is isomorphic to $L(\Pi)$, 
for some finite poset $\Pi$.
\end{theorem}

Our application of Birkhoff's Theorem deals with the sublattices of a finite distributive 
lattice. First, in \Cref{def:compression} we state the critical operation of 
\emph{compression of a poset}.  

\begin{definition}
	\label{def:compression}
Given a finite poset $\Pi$, first partition its elements; each subset will be called
a \emph{meta-element}. Define the following precedence relations among the meta-elements: 
if $x,y$ are elements of $\Pi$ such that $x$ is in meta-element $X$, $y$ is in meta-element $Y$ 
and $x$ precedes $y$, then $X$ precedes $Y$. Assume that these precedence relations yield a
partial order, say $Q$, on the meta-elements (if not, this particular partition is not useful 
for our purpose). Let $\Pi'$ be any partial order on the meta-elements such that the precedence
relations of $Q$ are a subset of the precedence relations of $\Pi'$. Then $\Pi'$ will be called
a {\em compression} of $\Pi$. Let $A_s$ and $A_t$ denote the meta-elements of $\Pi'$ containing $s$ 
and $t$, respectively.
\end{definition}

\begin{figure}[ht]
	\begin{wbox}
		\begin{minipage}[c]{0.49\textwidth}
			\centering
			\def\svgscale{0.4}
			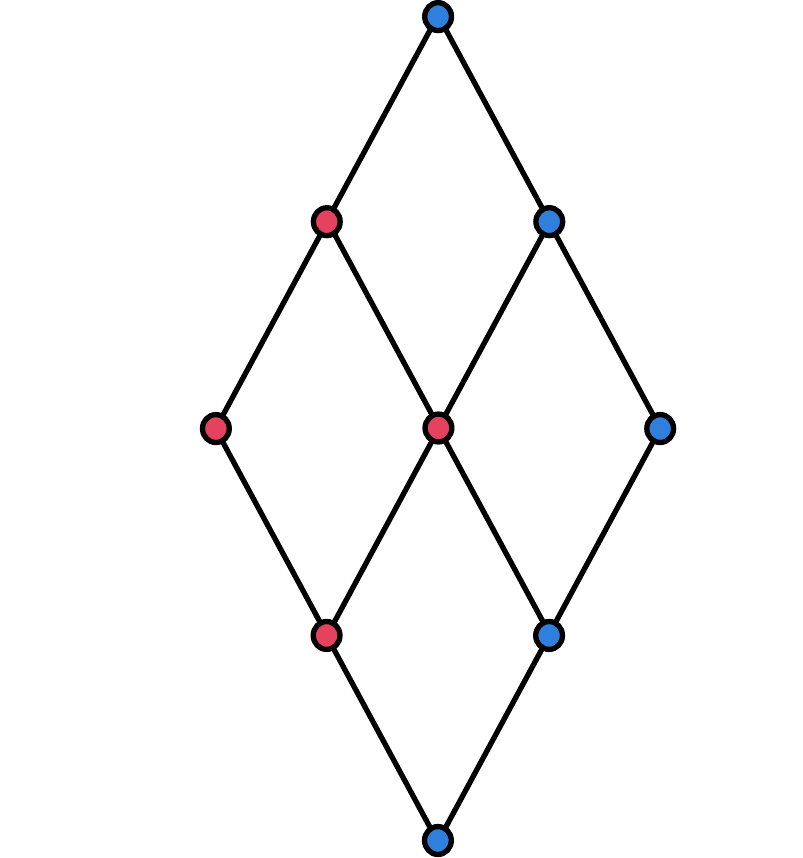
			
			~~~~$\mathcal{L}$
		\end{minipage}
		\begin{minipage}[c]{0.49\textwidth}
			\centering
			\def\svgscale{0.4}
			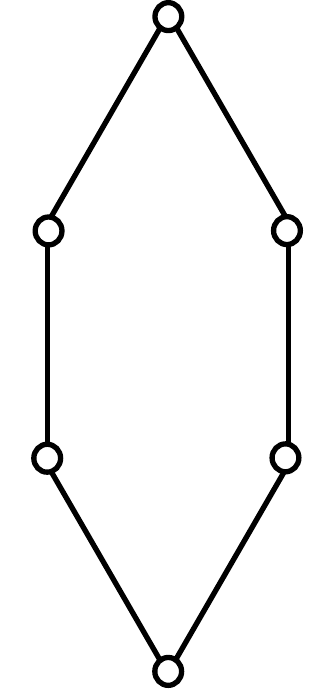
			
			$P$	
		\end{minipage}
		\par  
		\begin{minipage}[c]{0.54\textwidth}
			\centering
			\def\svgscale{0.4}
			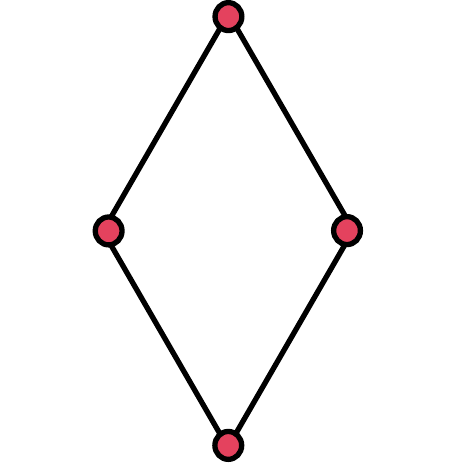
			
			$P_1$	
		\end{minipage}
		\begin{minipage}[c]{0.44\textwidth}
			\centering
			\def\svgscale{0.4}
			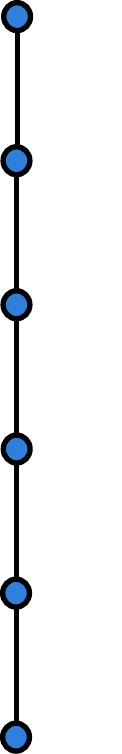
			$P_2$	
		\end{minipage}
	\end{wbox}
	\caption{Two examples of compressions. Lattice $\mathcal{L} = L(P)$.
$P_1$ and $P_2$ are compressions of $P$, and they generate the sublattices in $\mathcal{L}$, of red 
and blue elements, respectively.
	}
	\label{ex:compression} 
\end{figure} 

For examples of compressions see \Cref{ex:compression}.
Clearly, $A_s$ precedes all other meta-elements in $\Pi'$ and $A_t$ succeeds all other meta-elements 
in $\Pi'$. Once again, by a {\em proper closed set of $\Pi'$} we mean a closed set of $\Pi'$ that
contains $A_s$ and does not contain $A_t$. Then the lattice formed by the set of all proper closed 
sets of $\Pi'$ will be denoted by $L(\Pi')$. In the context of sublattices, compressions are very useful due to the \Cref{thm:generalization}. A proof of this theorem, in the context of stable matchings, is provided in Appendix~\ref{app:gen_proof}.

\begin{theorem}
	\label{thm:generalization} 
	There is a one-to-one correspondence between the compressions of $\Pi$ and the sublattices of 
	$L(\Pi)$ such that if sublattice $\mathcal{L}'$ of $L(\Pi)$ corresponds to compression $\Pi'$, 
	then $\mathcal{L}'$ is generated by $\Pi'$.
\end{theorem}

\subsubsection{An alternative view of compression}
\label{sec:alternative}

In this section we give an alternative definition of compression of a poset; this will be used
in the rest of the paper. The advantage of this definition is that it is much 
easier to work with for the applications presented later. Its drawback
is that several different sets of edges may yield the same compression. Therefore, this definition is not suitable for stating a one-to-one correspondence between sublattices of $\mathcal{L}$ and compressions of $\Pi$. Finally we show that any compression $\Pi'$ obtained using the first definition can also be obtained via the second definition and vice versa (\Cref{prop.eq}), hence showing that the two definitions are equivalent for our purposes.

We are given a poset $\Pi$ for a stable matching instance; let
$\mathcal{L}$ be the lattice it generates. Let $H(\Pi)$ denote the Hasse diagram of $\Pi$. 
Consider the following operations to derive a new poset $\Pi'$: Choose a set $E$ of directed edges
to add to $H(\Pi)$ and let $H_E$ be the resulting graph. Let $H'$ be the graph obtained by 
shrinking the strongly connected components of $H_E$; each strongly connected component will be
a meta-rotation of $\Pi'$. The edges which are not shrunk will
define a DAG, $H'$, on the strongly connected components. These edges give precedence relations
among meta-rotation for poset $\Pi'$.

Let $\mathcal{L}'$ be the sublattice of $\mathcal{L}$ generated by $\Pi'$. We will say that the
set of edges $E$ \emph{defines} $\mathcal{L}'$. 
It can be seen that each set $E$ uniquely defines a sublattice $L(\Pi')$; however,
there may be multiple sets that define the same sublattice. 
Observe that given a compression $\Pi'$ of $\Pi$, a set $E$ of edges defining $L(\Pi')$ can easily be obtained.
See~\Cref{ex:edgeSets} for examples of sets of edges which define sublattices.

\begin{proposition}
\label{prop.eq}
	The two definitions of compression of a poset are equivalent.
\end{proposition}
\begin{proof}
	Let $\Pi'$ be a compression of $\Pi$ obtained using the first definition. Clearly, for each
meta-rotation in $\Pi'$, we can add edges to $\Pi$ so the strongly connected component created
is precisely this meta-rotation. Any additional precedence relations introduced among incomparable
meta-rotations can also be introduced by adding appropriate edges.

The other direction is even simpler, since each strongly connected component can be defined to 
be a meta-rotation and extra edges added can also be simulated by introducing new precedence
constraints. 
\end{proof}

\begin{figure}
	\begin{wbox}
		\begin{minipage}[c]{0.33\textwidth}
			\centering
			\def\svgscale{0.4}
			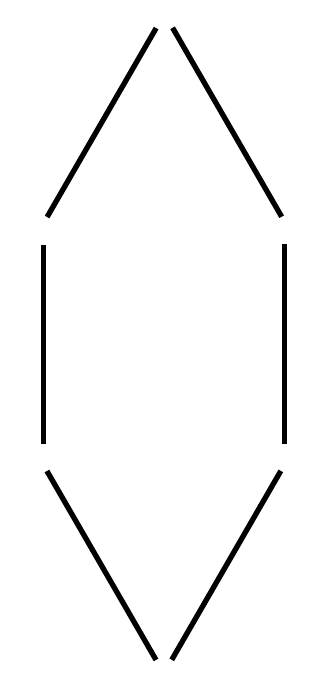
			
			$E_1$
		\end{minipage}
		\begin{minipage}[c]{0.32\textwidth}
			\centering
			\def\svgscale{0.4}
			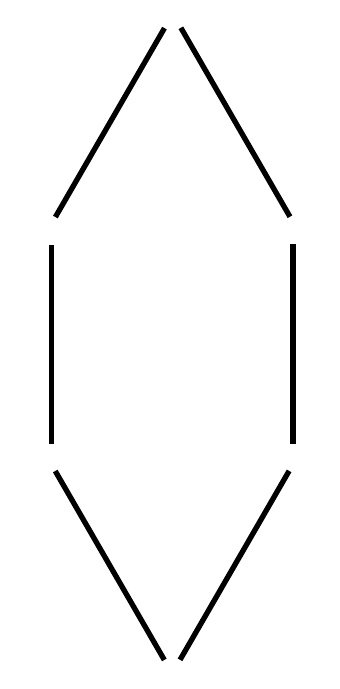
			
			$E_2$
		\end{minipage}
		\begin{minipage}[c]{0.3\textwidth}
			\centering
			\def\svgscale{0.4}
			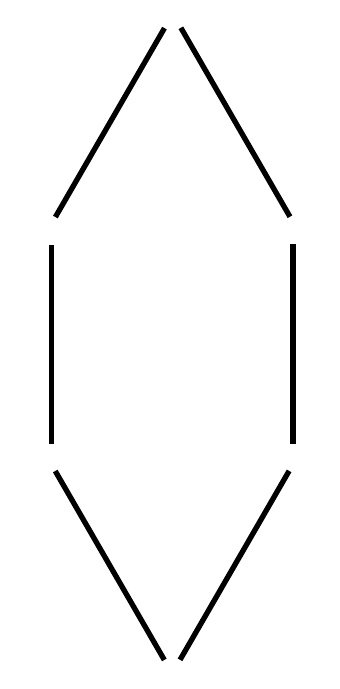
			
			$E_3$
		\end{minipage}
	\end{wbox}
	\caption{$E_1$ (red edges) and $E_2$ (blue edges) define the sublattices in 
~\Cref{ex:compression}, of red and blue elements, respectively. $E_2$ and $E_3$ define the same compression and represent the same sublattice. All black edges in $E_1, E_2$ and $E_3$ are directed from top to bottom (not shown in the figure). 	}
	\label{ex:edgeSets} 
\end{figure} 

For a (directed) edge $e = uv \in E$, $u$ is called the \emph{tail} and $v$ is 
called the \emph{head} of $e$.
Let $I$ be a closed set of $\Pi$. Then we say that:
\begin{itemize}
	\item $I$ \emph{separates} an edge $uv \in E$ if $v \in I$ and $u \not \in I$. 
	\item $I$ \emph{crosses} an edge $uv \in E$ if $u \in I$ and $v \not \in I$. 
\end{itemize}
If $I$ does not separate or cross any edge $uv \in E$, $I$ is called a \emph{splitting} set 
w.r.t. $E$.

\begin{lemma}
	\label{lem:separating}
	Let $\mathcal{L}'$ be a sublattice of $\mathcal{L}$ and $E$ be a set of edges defining $\mathcal{L}'$. 
	A matching $M$ is in $\mathcal{L}'$ iff the closed subset $I$ generating $M$ does not separate any edge $uv \in E$.
\end{lemma}  

\begin{proof}
	Let $\Pi'$ be a compression corresponding to $\mathcal{L}'$.
	By~\Cref{thm:generalization}, the matchings in $\mathcal{L}'$ are generated by eliminating rotations in closed subsets of $\Pi'$. 
	
	First, assume $I$ separates $uv \in E$. Moreover, assume $M \in \mathcal{L}'$ for the sake of contradiction, and let $I'$ be the closed subset of $\Pi'$ corresponding to $M$. Let $U$ and $V$ be the meta-rotations containing $u$ and $v$ respectively.
	Notice that the sets of rotations in $I$ and $I'$ are identical. Therefore, $V \in I'$ and $U \not \in I'$. 
	Since $uv \in E$, there is an edge from $U$ to $V$ in $H'$. Hence, $I'$ is not a closed subset of $\Pi'$.
	
	Next, assume that $I$ does not separate any $uv \in E$. We show that the rotations in $I$ can be partitioned into meta-rotations in a closed subset $I'$ of $\Pi'$. If $I$ cannot be partitioned into meta-rotations, there must exist a meta-rotation $A$ such that $A \cap I$ is a non-empty proper subset of $A$. Since $A$ consists of rotations in a strongly connected component of $H_E$, there must be an edge $uv$ from $A \setminus I$ to $A \cap I$ in $H_E$. Hence, $I$ separates $uv$. Since $I$ is a closed subset, $uv$ can not be an edge in $H$. Therefore, $uv \in E$, which is a contradiction. It remains to show that the set of meta-rotations partitioning $I$ is a closed subset of $\Pi'$. Assume otherwise, there exist meta-rotation $U \in I'$ and $V \not \in I'$ such that there exists an edge from $U$ to $V$ in $H'$. Therefore, there exists $u \in U$, $v \in V$ and $uv \in E$, which is a contradiction.
\end{proof}

\begin{remark}
	\label{remark:addedEdge}
We may assume w.l.o.g. that the set $E$ defining $\mathcal{L}'$ is {\em minimal} in the following sense:
There is no edge $uv \in E$ such that $uv$ is not separated by any closed set of $\Pi$. Observe that
if there is such an edge, then $E \setminus \{uv\}$ defines the same sublattice $\mathcal{L}'$. 
Similarly, there is no edge $uv \in E$ such that each closed set separating $uv$ also separates another edge in $E$.
\end{remark}

\begin{definition}
\label{def.poset}
W.r.t. an element $v$ in a poset $\Pi$, we define four useful subsets of $\Pi$:
\begin{align*}
I_v &= \{r \in \Pi: r \prec v \} \\
J_v & = \{r \in \Pi: r \preceq v\} = I_v \cup \{v\} \\
I'_v &= \{r \in \Pi: r \succ v \} \\
J'_v & = \{r \in \Pi: r \succeq v\} = I'_v \cup \{v\}
\end{align*} 
Notice that $I_v, J_v, \Pi \setminus I'_v, \Pi \setminus J'_v$ are all closed sets. 
\end{definition}

\begin{lemma}
	\label{lem:prime}
	Both $J_v$ and $\Pi \setminus J'_u$ separate $uv$ for each $uv \in E$.
\end{lemma}
\begin{proof}
Since $uv$ is in $E$, $u$ cannot be in $J_v$; otherwise, there is no closed subset separating 
$uv$, contradicting~\Cref{remark:addedEdge}. Hence, $J_v$ separates $uv$ for all $uv$ in $E$.
	
	Similarly, since $uv$ is in $E$, $v$ cannot be in $J'_u$. Therefore, $\Pi \setminus J'_u$ contains $v$ but not $u$, and thus separates $uv$.
\end{proof}

With these definitions and properties, we can now look into the properties of edge sets defining a sublattice $L_1$ of $L$, depending on how $L_1$ partitions $L$. In \Cref{sec:sublattice}, we consider at the case when $L$ is partitioned into two sublattices and in \Cref{sec:semi}, we consider the case when $L$ is partitioned into a sublattice and a semi-sublattice.

\subsubsection{Setting I} 
\label{sec:sublattice}

Under Setting I, the given lattice $\mathcal{L}$ has sublattices $\mathcal{L}_1$ and $\mathcal{L}_2$ such that 
$\mathcal{L}_1$ and $\mathcal{L}_2$ partition $\mathcal{L}$. The main structural fact for this setting is:

\begin{theorem}
	\label{thm:alternating}
	Let $\mathcal{L}_1$ and $\mathcal{L}_2$ be sublattices of $\mathcal{L}$ such that $\mathcal{L}_1$ and $\mathcal{L}_2$ partition $\mathcal{L}$. 
	Then there 
	exist sets of edges $E_1$ and $E_2$ defining $\mathcal{L}_1$ and $\mathcal{L}_2$ such that they
	form an alternating path from $t$ to $s$. 
	%To be precise, there exists a sequence $r_0 = t, r_1, \ldots, r_k = s$ such that $r_{0}r_{1}, \ldots, r_{k-2}r_{k-1} , r_{k-1}r_k$ are all edges in $E_1$ and $E_2$ and they alternate between $E_1$ and $E_2$.
\end{theorem}

We will prove this theorem in the context of stable matchings. 
Let $E_1$ and $E_2$ be any two sets of edges defining $\mathcal{L}_1$ and $\mathcal{L}_2$,
respectively. 
We will show that $E_1$ and $E_2$ can be adjusted so that they form an alternating path from $t$ to $s$,
without changing the corresponding compressions. 

\begin{lemma}
	\label{lem:path}
	There must exist a path from $t$ to $s$ composed of edges in $E_1$ and $E_2$.
\end{lemma}

\begin{proof}
	Let $R$ denote the set of vertices reachable from $t$ by a path of edges in $E_1$ and $E_2$. Assume by contradiction that $R$ does not contain $s$. Consider the matching $M$ generated by rotations in $\Pi \setminus R$. Without loss of generality, assume that $M \in \mathcal{L}_1$. 
	By \Cref{lem:separating}, $\Pi \setminus R$ separates an edge $uv \in E_2$. Therefore, $u \in R$ and $v \in \Pi \setminus R$. Since $uv \in E_2$, $v$ is also reachable from $t$ by a path of edges in $E_1$ and $E_2$.
\end{proof}

Let $Q$ be a path from $t$ to $s$ according to \Cref{lem:path}. 
Partition $Q$ into subpaths $Q_1, \ldots, Q_k$ such that each $Q_i$ consists of edges in either $E_1$ or $E_2$ and 
$E(Q_i) \cap E(Q_{i+1}) = \emptyset$ for all $1 \leq i \leq k-1$.
Let $r_i$ be the rotation at the end of $Q_i$ except for $i = 0$ where $r_0 = t$. Specifically, $t = r_0 \rightarrow r_1 \rightarrow \ldots \rightarrow r_k = s$ in $Q$.
We will show that each $Q_i$ can be replaced by a direct edge from $r_{i-1}$ to $r_i$, and furthermore,
all edges not in $Q$ can be removed.

\begin{lemma}
	\label{lem:replace}
	Let $Q_i$ consist of edges in $E_\alpha$ ($\alpha$ = 1 or 2).
	$Q_i$ can be replaced by an edge from $r_{i-1}$ to $r_i$ where $r_{i-1}r_i \in E_\alpha$.
\end{lemma}
\begin{proof}
	A closed subset separating $r_{i-1}r_i$ must separate an edge in $Q_i$.
	Moreover, any closed subset must separate exactly one of $r_{0}r_{1}, \ldots, r_{k-2}r_{k-1} , r_{k-1}r_k$.
	Therefore, the set of closed subsets separating an edge in $E_1$ (or $E_2$) remains unchanged.
\end{proof}

\begin{lemma}
	\label{lem:remove}
	Edges in $E_1\cup E_2$ but not in $Q$ can be removed.
\end{lemma}
\begin{proof}
	Let $e$ be an edge in $E_1\cup E_2$ but not in $Q$. 
	Suppose that $e \in E_1$.
	Let $I$ be a closed subset separating $e$.
	By \Cref{lem:separating}, the matching generated by $I$ belongs to $\mathcal{L}_2$.
	Since $e$ is not in $Q$ and $Q$ is a path from $t$ to $s$, $I$ must separate another edge $e'$ in $Q$. 
	By \Cref{lem:separating}, $I$ can not separate edges in both $E_1$ and $E_2$. 
	Therefore, $e'$ must also be in $E_1$. 
	Hence, the matching generated by $I$ will still be in $\mathcal{L}_2$ after removing $e$ from $E_1$.
	The argument applies to all closed subsets separating $e$.
\end{proof}

By \Cref{lem:replace} and \Cref{lem:remove}, $r_{0}r_{1}, \ldots, r_{k-2}r_{k-1} , r_{k-1}r_k$ are all edges in $E_1$ and $E_2$ and they alternate between $E_1$ and $E_2$. Therefore, we have \Cref{thm:alternating}. An illustration of such a path is given in \Cref{ex:canonicalPathAndBouquet}(a).

%The path from $t$ to $s$ in \Cref{thm:alternating} is called a \emph{canonical path}.

\begin{proposition}
\label{prop.sub}
	There exists a sequence of rotations $r_0, r_1, \ldots , r_{2k}, r_{2k+1}$ such that 
	a closed subset generates a matching in  $\mathcal{L}_1$ iff it
	contains $r_{2i}$ but not $r_{2i+1}$ for some $0 \leq i \leq k$.
\end{proposition}

\begin{figure}
	\begin{wbox}
		\begin{minipage}[c]{0.49\textwidth}
			\centering
			\def\svgscale{0.4}
			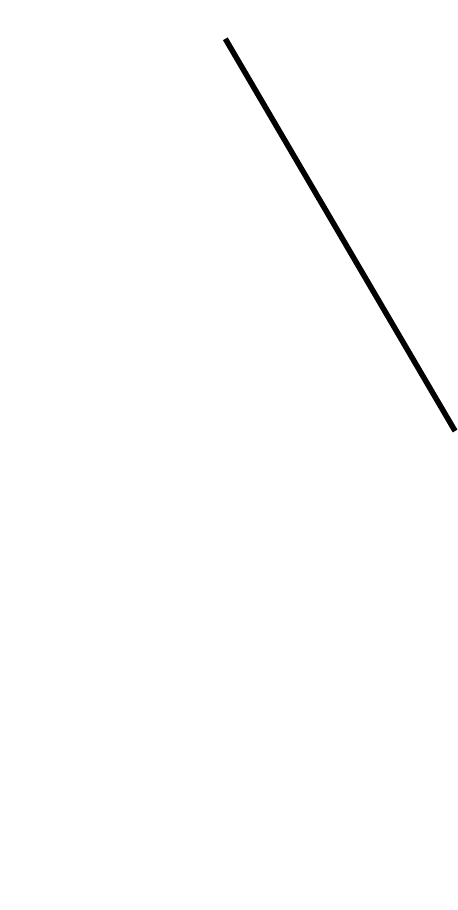
			
			(a)
		\end{minipage}
		\begin{minipage}[c]{0.49\textwidth}
			\centering
			\def\svgscale{0.4}
			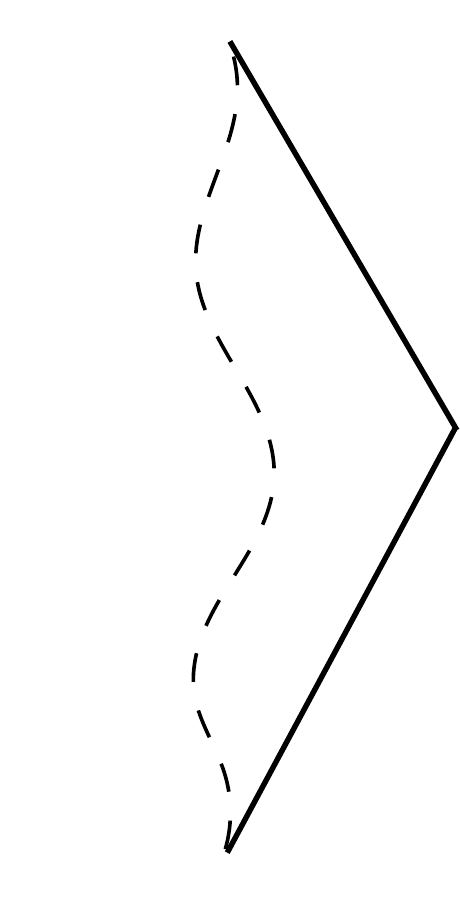
			
			(b)	
		\end{minipage}
	\end{wbox}
	\caption{Examples of: (a) canonical path, and (b) bouquet. 
	}
	\label{ex:canonicalPathAndBouquet} 
\end{figure}

\subsubsection{Setting II}
\label{sec:semi}

Under Setting II, the given lattice $\mathcal{L}$ can be partitioned into a sublattice $\mathcal{L}_1$ and a 
semi-sublattice $\mathcal{L}_2$. 
We assume that $\mathcal{L}_2$ is a join semi-sublattice. 
Clearly by reversing the order of $\mathcal{L}$, the case of meet semi-sublattice is also covered.
The next theorem, which generalizes \Cref{thm:alternating},
gives a sufficient characterization of a set of edges $E$ defining $\mathcal{L}_1$.

\begin{theorem}
\label{thm:semi}
There exists a set of edges $E$ defining sublattice $\mathcal{L}_1$ such that:
\begin{enumerate}
	\item The set of tails $T_E$ of edges in $E$ forms a chain in $\Pi$.
	\item There is no path of length two consisting of edges in $E$.
	\item For each $r \in T_E$, let
	\[F_r = \{ v \in \Pi: rv \in E \}.\]
	Then any two rotations in $F_r$ are incomparable.
	\item For any $r_i,r_j \in T_E$ where $r_i \prec r_j$, there exists a splitting set containing all rotations in $F_{r_i} \cup \{r_i\}$ and no rotations in $F_{r_j} \cup \{r_j\}$.  
\end{enumerate}
\end{theorem}

A set $E$ satisfying \Cref{thm:semi} will be called a \emph{bouquet}. 
For each $r \in T_E$, let $L_r = \{rv \ | \ v \in F_r \}$. Then $L_r$ will be called a 
\emph{flower}. Observe that the bouquet $E$ is partitioned into flowers. 
These notions are illustrated in \Cref{ex:canonicalPathAndBouquet}(b). The black path,
directed from $s$ to $t$, is the chain mentioned in \Cref{thm:semi} and
the red edges constitute $E$. Observe that the tails of edges $E$ lie on the chain. For each such
tail, the edges of $E$ outgoing from it constitute a flower.

Let $E$ be an arbitrary set of edges defining $\mathcal{L}_1$. We will show that $E$ can be modified so 
that the conditions in \Cref{thm:semi} are satisfied.
Let $S$ be a splitting set of $\Pi$. In other words, $S$ is a closed subset such that for all $uv \in E$, either $u,v$ are both in $S$ or $u,v$ are both in $\Pi \setminus S$.

\begin{lemma}
	\label{lem:uniqueMaximal}
	There is a unique maximal rotation in $T_E \cap S$.
\end{lemma}
\begin{proof}
	Suppose there are at least two maximal rotations $u_1,u_2, \ldots u_k$ ($k \geq 2$) in $T_E \cap S$.
	Let $v_1, \ldots v_k$ be the heads of edges containing  $u_1,u_2, \ldots u_k$.
	For each $1 \leq i \leq k$, let $S_i = J_{u_i} \cup J_{v_j}$ where $j$ is any index such that $j \not = i$.
	Since $u_i$ and $u_j$ are incomparable, $u_j \not \in J_{u_i}$. 
	Moreover, $u_j \not \in J_{v_j}$ by 
    \Cref{lem:prime}.  
	Therefore, $u_j \not \in S_i$. 
	It follows that $S_i$ contains $u_i$ and separates $u_j v_j$.
	Since $S_i$ separates $u_jv_j \in E$, the matching generated by $S_i$ is in $\mathcal{L}_2$ according to \Cref{lem:separating}.
	
	Since $\bigcup_{i=1}^k S_i$ contains all maximal rotations in $T_E \cap S$ and $S$ does not separate any edge in $E$,
	$\bigcup_{i=1}^k S_i$ does not separate any edge in $E$ either. Therefore, the matching generated by $\bigcup_{i=1}^k S_i$ 
	is in $\mathcal{L}_1$, and hence not in $\mathcal{L}_2$. This contradicts the fact that $\mathcal{L}_2$ is a join semi-sublattice. 
\end{proof}

%Let $S = P$ in \Cref{lem:uniqueMaximal}, and denote $r_1$ by the unique maximal rotation in $T_E$.
Denote by $r$ the unique maximal rotation in $T_E \cap S$. Let
\begin{align*}
R_r &= \{ v \in \Pi: \text{ there is a path from $r$ to $v$ using edges in $E$}\}, \\
E_r &= \{ uv \in E: u,v \in R_r \}, \\
G_r &= \{R_r, E_r \}.
\end{align*}
%and $G_r$ be the graph whose vertices are in $R_r$ and edges are in $E_r$. 
Note that $r \in R_r$.
For each $v \in R_r$ there exists a path from $r$ to $v$ and $r \in S$. 
Since $S$ does not cross any edge in the path,  $v$ must also be in $S$.
Therefore, $R_r \subseteq S$. 

\begin{lemma}
	\label{lem:semi-replace}
	Let $u \in (T_E \cap S) \setminus R_{r}$ such that $u \succ x$ for $x \in R_{r}$. Then we can replace each $uv \in E$ with $rv$.
\end{lemma}
\begin{proof}
	We will show that the set of closed subsets separating an edge in $E$ remains unchanged. 
	
	Let $I$ be a closed subset separating $uv$. Then $I$ must also separate $rv$ since $r \succ v$. 
	
	Now suppose $I$ is a closed subset separating $rv$. We consider two cases:
	\begin{itemize}
		\item If $u \in I$, $I$ must contain $x$ since $u \succ x$. Hence, $I$ separates an edge in the path from $r$ to $x$.
		\item If $u \not \in I$, $I$ separates $uv$. 
	\end{itemize}
\end{proof}

Keep replacing edges according to \Cref{lem:semi-replace} until there is no $u \in (T_E \cap S) \setminus R_{r}$ such that $u \succ x$ for some $x \in R_{r}$.

\begin{lemma}
	\label{lem:flower-separating}
	Let
	\[ X = \{v \in S: v \succeq x \text{ for some } x \in R_{r}\}. \]
	\begin{enumerate}
		\item $S \setminus X$ is a closed subset. 
		\item $S \setminus X$ contains $u$ for each $u \in (T_E \cap S) \setminus R_{r}$.
		\item $(S \setminus X) \cap R_{r} = \emptyset$.
		\item $S \setminus X$ is a splitting set.
	\end{enumerate}
\end{lemma}

\begin{proof}
The lemma follows from the claims given below: 

	\paragraph{Claim 1.}  $S \setminus X$ is a closed subset.
	\begin{proof}
		Let $v$ be a rotation in $S \setminus X$ and $u$ be a predecessor of $v$.
		Since $S$ is a closed subset, $u \in S$.
		Notice that if a rotation is in $X$, all of its successor must be included.
		Hence, since $v \notin X$, $u \notin X$.
		Therefore, $u \in S \setminus X$.
	\end{proof}
	
	\paragraph{Claim 2.} $S \setminus X$ contains $u$ for each $u \in (T_E \cap S) \setminus R_{r}$.
	\begin{proof}
		After replacing edges according to \Cref{lem:semi-replace}, for each $u \in (T_E \cap S) \setminus R_{r}$ we must have that $u$ does not succeed any $x \in R_r$. Therefore, $u \notin X$ by the definition of $X$.		
	\end{proof}

	\paragraph{Claim 3.} $(S \setminus X) \cap R_{r} = \emptyset$.
	\begin{proof}
		Since $R_{r} \subseteq X$, $(S \setminus X) \cap R_{r} = \emptyset$.
	\end{proof}
	
	\paragraph{Claim 4.} $S \setminus X$ does not separate any edge in $E$.
	\begin{proof}
		Suppose $S \setminus X$ separates $uv \in E$. Then $u \in X$ and $v \in  S \setminus X$.
		By Claim 2, $u$ can not be a tail vertex, which is a contradiction.  
	\end{proof}
	
	\paragraph{Claim 5.} $S \setminus X$ does not cross any edge in $E$.
	\begin{proof}
		Suppose $S \setminus X$ crosses $uv \in E$. Then $u \in S \setminus X$ and $v \in X$.
		Let $J$ be a closed subset separating $uv$. Then $v \in J$ and $u \notin J$.
		
		Since $uv \in E$ and $u \in S$, $u \in T_E \cap S$. Therefore, $r \succ u$ by \Cref{lem:uniqueMaximal}.
		Since $J$ is a closed subset, $r \notin J$.
		
		Since $v \in X$, $v \succeq x$ for $x \in R_r$. Again, as $J$ is a closed subset, $x \in J$.  
		
		Therefore, $J$ separates an edge in the path from $r$ to $x$ in $G_r$. 
		Hence, all closed subsets separating $uv$ must also separate another edge in $E_r$.
		This contradicts the assumption made in \Cref{remark:addedEdge}.
	\end{proof}
\end{proof}

\begin{lemma}
	\label{lem:semi-replace2}
	$E_r$ can be replaced by the following set of edges:
	\[E'_r = \{rv: v \in R_r \}.\]
\end{lemma}
\begin{proof}
	\iffalse
	First consider the set of vertices 
	\[ X = \{u \in R_r: \exists uv \in E\}. \]
	Then for $u \in X$, $u$ is also in $T_E \cap S$.
	By \Cref{lem:uniqueMaximal}, $u \prec r$. 
	Recall that $H_E$ is the graph obtained by adding $E$ to the Hasse diagram of $P'$.
	Since $u \prec r$ in $P$ and $u \in R_r$, $u$ and $r$ are in the same strongly connected component of $H_E$. 
	Therefore, we can remove all edges among vertices in $X$ and connect $r$ to each $u \in X$.
	After such modification, the strongly connected component containing $r$ is preserved, and no edges are added outside the component. 
	
	Now let $u$ be a rotation in $R_r \setminus X$. Then after the above replacements, $u$ must be the end of a length-two path from $r$. Let the path be composed of $rv$ and $vu$ in $E$. 
	Since $vu \in E$, $v \in T_E \cap S$. By \Cref{lem:uniqueMaximal}, $r \succ v$. 
	Thus, any closed subset separating $vu$ must also separate $ru$. 
	Moreover, any closed subset separating $ru$ must either separate $rv$ or $vu$.
	Therefore, $vu$ can be replaced with $ru$.
	\fi
	
	We will show that the set of closed subsets separating an edge in $E_r$ and 
	the set of closed subset separating an edge in $E'_r$ are identical. 
	
	Consider a closed subset $I$ separating an edge in $rv \in E'_r$. 
	Since $v \in R_r$, $I$ must separate an edge in $E$ in a path from $r$ to $v$.
	By definition, that edge is in $E_r$.
	
	Now let $I$ be a closed subset separating an edge in $uv \in E_r$.
	Since $uv \in E$, $u \in T_E \cap S$. By \Cref{lem:uniqueMaximal}, $r \succ u$. 
	Thus, $I$ must also separate $rv \in E'_r$. 
\end{proof}

We can now prove \Cref{thm:semi}.
\begin{proof} [Proof of \Cref{thm:semi}]
	To begin, let $S_1 = \Pi$ and let $r_1$ be the unique maximal rotation according to \Cref{lem:uniqueMaximal}.
	Then we can replace edges according to \Cref{lem:semi-replace} and \Cref{lem:semi-replace2}.
	After replacing, $r_1$ is the only tail vertex in $G_{r_1}$. 
	By \Cref{lem:flower-separating}, there exists a set $X$ such that $S_1 \setminus X$ does not contain any vertex in $R_{r_1}$ and contains all other tail vertices in $T_E$ except $r_1$. 
	Moreover, $S_1 \setminus X$ is a splitting set.
	Hence, we can set $S_2 = S_1 \setminus X$ and repeat. 
	
	Let $r_1, \ldots , r_k$ be the rotations found in the above process. Since $r_i$ is the unique maximal rotation in $T_E \cap S_i$ for all $1 \leq i \leq k$ and $S_1 \supset S_2 \supset \ldots \supset S_k$, we have $r_1 \succ r_2 \succ \ldots \succ r_k$. 
	By \Cref{lem:semi-replace2}, for each $1 \leq i \leq k$, $E_{r_i}$ consists of edges $r_iv$ for $v \in R_{r_i}$.
	Therefore, there is no path of length two composed of edges in $E$ and condition 2 is satisfied. Moreover, $r_1, \ldots , r_k$ are exactly the tail vertices in $T_E$, which gives condition 1.
	
	Let $r$ be a rotation in $T_E$ and consider $u,v \in F_r$. Moreover, assume that $u \prec v$. A closed subset $I$ separating $rv$ contains $v$ but not $r$. Since $I$ is a closed subset and $u \prec v$, $I$ contains $u$. Therefore, $I$ also separates $ru$, contradicting the assumption in \Cref{remark:addedEdge}. The same argument applies when $v \prec u$. Therefore, $u$ and $v$ are incomparable as stated in condition 3. 
	
	Finally, let $r_i, r_j \in T_E$ where $r_i \prec r_j$. By the construction given above, 
	$S_j \supset S_{j-1} \supset \ldots \supset S_i$, $R_{r_j} \subseteq S_j \setminus S_{j-1}$ and $R_{r_i} \subseteq S_i$.
	Therefore, $S_i$ contains all rotations in $R_{r_i}$ but none of the rotations in $R_{r_j}$, giving condition 4. 
\end{proof}

\Cref{thm:semi} naturally leads us to \Cref{prop:gen} explaining the structure if the unique edge set defining a sublattice in Setting II.

\begin{proposition}
\label{prop:gen}
	There exists a sequence of rotations $r_1 \prec \ldots \prec r_{k}$ 
	and a set $F_{r_i}$ for each $1 \leq i \leq k$ such that 
	a closed subset generates a matching in $\mathcal{L}_1$ if and only if
	whenever it contains a rotation in $F_{r_i}$, it must also contain $r_i$.
\end{proposition}

\subsubsection{Algorithm for Finding a Bouquet}
\label{sec:alg}
In this section, we give an algorithm for finding a bouquet. Let $\mathcal{L}$
be a distributive lattice that can be partitioned into a sublattice $\mathcal{L}_1$ and a 
semi-sublattice $\mathcal{L}_2$.
Then given a poset $\Pi$ of $\mathcal{L}$ and a membership oracle, which determines if a matching
of $\mathcal{L}$ is in $\mathcal{L}_1$ or not, the algorithm returns a bouquet defining $\mathcal{L}_1$.

By \Cref{thm:semi}, the set of tails $T_E$ forms a chain $C$ in $\Pi$.
The idea of our algorihm, given in \Cref{alg:flowerSet}, is to find 
the flowers according to their order in $C$.
Specifically, a splitting set $S$ is maintained such that at any point, all flowers outside of $S$ are found.
At the beginning, $S$ is set to $\Pi$ and becomes smaller as the algorithm proceeds. 
Step 2 checks if $M_z$ is a matching in $\mathcal{L}_1$ or not. 
If $M_z \not\in \mathcal{L}_1$, the closed subset $\Pi \setminus \{t\}$ separates an edge in 
$E$ according to \Cref{lem:separating}. Hence, the first tail on $C$ must be $t$.
Otherwise, the algorithm jumps to Step 3 to find the first tail. 
Each time a tail $r$ is found, Step 5 immediately finds the flower $L_r$ corresponding to $r$. 
The splitting set $S$ is then updated so that $S$ no longer contains $L_r$ 
but still contains the flowers that have not been found yet. 
Next, our algorithm continues to look for the next tail inside the updated $S$.
If no tail is found, it terminates. 

\begin{algorithm}[ht]
	\begin{wbox}
		\textsc{FindBouquet}$(\Pi)$: \\
		\textbf{Input:} A poset $\Pi$. \\
		\textbf{Output:} A set $E$ of edges defining $\mathcal{L}_1$. 
		\begin{enumerate}
			\item Initialize: Let $S = \Pi, E = \emptyset$.
			\item If $M_z$ is in $\mathcal{L}_1$: go to Step 3. Else: $r = t$, go to Step 5.
			\item $r$ = \textsc{FindNextTail}$(\Pi,S)$. 
			\item If $r$ is not \textsc{Null}: Go to Step 5. Else: Go to Step 7. 
			\item $F_r$ = \textsc{FindFlower}$(\Pi,S,r)$.
			\item Update:
			\begin{enumerate}
				\item For each $u \in F_r$: $E \leftarrow E \cup \{ru\}$.
				\item $ S \leftarrow S \setminus  \bigcup_{u \in F_r \cup \{ r \}} J'_u $.
				\item Go to Step 3.
			\end{enumerate}
			\item Return $E$. 
		\end{enumerate}
	\end{wbox}
	\caption{Algorithm for finding a bouquet.}
	\label{alg:flowerSet} 
\end{algorithm}

First we prove a simple observation.

\begin{lemma}
	\label{lem:headTailCondition}
	Let $v$ be a rotation in $\Pi$. Let $S \subseteq \Pi$ such that both $S$ and $S \cup \{v\}$ are closed subsets.
	If $S$ generates a matching in $\mathcal{L}_1$ and $S \cup \{v\} $ generates a matching in $\mathcal{L}_2$, 
		$v$ is the head of an edge in $E$.
	If $S$ generates a matching in $\mathcal{L}_2$ and $S \cup \{v\} $ generates a matching in $\mathcal{L}_1$, 
		$v$ is the tail of an edge in $E$.

\end{lemma}
\begin{proof}
	Suppose that $S$ generates a matching in $\mathcal{L}_1$ and $S \cup \{v\} $ generates a matching in $\mathcal{L}_2$.
	By \Cref{lem:separating}, $S$ does not separate any edge in $E$, and 
	$S \cup \{v\}$ separates an edge $e \in E$. 
	This can only happen if $u$ is the head of $e$.
	
	A similar argument can be given for the second case.
\end{proof}

\begin{algorithm}[ht]
	\begin{wbox}
		\textsc{FindNextTail}$(\Pi,S)$: \\
		\textbf{Input:} A poset $\Pi$, a splitting set $S$. \\
		\textbf{Output:} The maximal tail vertex in $S$, or \textsc{Null} if there is no tail vertex in $S$.
		\begin{enumerate}
			\item Compute the set $V$ of rotations $v$ in $S$ such that:  
			\begin{itemize}
				\item $\Pi \setminus I'_v$ generates a matching in $\mathcal{L}_1$.
				\item $\Pi \setminus J'_v$ generates a matching in $\mathcal{L}_2$.
			\end{itemize}
			\item If $V \not = \emptyset$ and  there is a unique maximal element $v$ in $V$: Return $v$. \\
			Else: Return \textsc{Null}.
		\end{enumerate}
	\end{wbox}
	\caption{Subroutine for finding the next tail.}
	\label{alg:findTail} 
\end{algorithm} 

\begin{lemma}
	\label{lem:correctnessFindNextTail}
	Given a splitting set $S$,
	\textsc{FindNextTail}$(\Pi,S)$ (\Cref{alg:findTail}) returns the maximal tail vertex in $S$, 
	or \textsc{Null} if there is no tail vertex in $S$.
\end{lemma}
\begin{proof}
	Let $r$ be the maximal tail vertex in $S$.

	First we show that $r \in V$.
	By \Cref{thm:semi}, the set of tails of edges in $E$ forms a chain in $\Pi$.
	Therefore $\Pi \setminus I'_r$ contains all tails in $S$. 
	Hence, $\Pi \setminus I'_r$ does not separate any edge whose tails are in $S$. 
	Since $S$ is a splitting set, $\Pi \setminus I'_r$ does not separate any edge whose tails are in $\Pi \setminus S$.   
	Therefore, by \Cref{lem:separating}, $\Pi \setminus I'_r$ generates a matching in $\mathcal{L}_1$.
	By \Cref{lem:prime}, $\Pi \setminus J'_r$ must separate an edge in $E$, 
	and hence generates a matching in $\mathcal{L}_2$ according to \Cref{lem:separating}. 
	
	By \Cref{lem:headTailCondition}, any rotation in $V$ must be the tail of an edge in $E$.
	Hence, they are all predecessors of $r$ according to \Cref{thm:semi}.
\end{proof}

\begin{algorithm}[ht]
	\begin{wbox}
		\textsc{FindFlower}$(\Pi,S,r)$: \\
		\textbf{Input:} A poset $\Pi$, a tail vertex $r$ and a splitting set $S$ containing $r$. \\
		\textbf{Output:} The set $F_r = \{ v \in \Pi: rv \in E \}$.
		\begin{enumerate}
			\item Compute $X = \{ v \in I_r: J_v \text{ generates a matching in } \mathcal{L}_1 \}$.
			\item Let $Y = \bigcup_{v \in X} J_v$.
			\item If $Y = \emptyset$ and $M_0 \in \mathcal{L}_2$: Return $\{s\}$. 
			\item Compute the set $V$ of rotations $v$ in $S$ such that:
			\begin{itemize}
				\item $Y \cup I_v$ generates a matching in $\mathcal{L}_1$.
				\item $Y \cup J_v$ generates a matching in $\mathcal{L}_2$.
			\end{itemize} 
			\item Return $V$.
		\end{enumerate}
	\end{wbox}
	\caption{Subroutine for finding a flower.}
	\label{alg:findflower} 
\end{algorithm}

\begin{lemma}
	\label{lem:correctnessFindFlower}
	Given a tail vertex $r$ and a splitting set $S$ containing $r$, \textsc{FindFlower}$(\Pi,S,r)$ (\Cref{alg:findflower})
	correctly returns $F_r$.
\end{lemma}

\begin{proof}
	First we give two crucial properties of the set $Y$.
	By \Cref{thm:semi}, the set of tails of edges in $E$ forms a chain $C$ in $\Pi$.
	
	\paragraph{Claim 1.} $Y$ contains all predecessors of $r$ in $C$.
	\begin{proof}
		Assume that there is at least one predecessor of $r$ in $C$, and denote by $r'$ the direct predecessor. 
		It suffices to show that $r' \in Y$. 
		By \Cref{thm:semi}, there exists a splitting set $I$ such that $R_{r'} \subseteq I$ and $R_r \cap I = \emptyset$.
		Let $v$ be the maximal element in $C \cap I$.
		Then $v$ is a successor of all tail vertices in $I$.
		It follows that $J_v$ does not separate any edges in $E$ inside $I$.
		Therefore, $v \in X$.
		Since $J_v \subseteq Y$, $Y$ contains all predecessors of $r$ in $C$.	
	\end{proof} 
	
	\paragraph{Claim 2.}  $Y$ does not contain any rotation in $F_r$.
	\begin{proof}
		Since $Y$ is the union of closed subset generating matching in $\mathcal{L}_1$, $Y$ also generates a matching in $\mathcal{L}_1$.
		By \Cref{lem:separating}, $Y$ does not separate any edge in $E$. 
		Since $r \not \in Y$, $Y$ must not contain any rotation in $F_r$.
	\end{proof}
	
	By Claim 1, if $Y = \emptyset$, $r$ is the last tail found in $C$. 
	Hence, if $M_0 \in \mathcal{L}_2$, $s$ must be in $F_r$. 
	By \Cref{thm:semi}, the heads in $F_r$ are incomparable. 
	Therefore, $s$ is the only rotation in $C$.
	\textsc{FindFlower} correctly returns $\{s\}$ in Step 3.
	Suppose such a situation does not happen, we will show that the returned set is $F_r$.

	\paragraph{Claim 3.} $V = F_r$.
	\begin{proof}
		Let $v$ be a rotation in $V$. 
		By \Cref{lem:headTailCondition}, $v$ is a head of some edge $e$ in $E$.
		Since $Y$ contains all predecessors of $r$ in $C$, the tail of $e$ must be $r$. 
		Hence, $v \in  F_r$.
		
		Let $v$ be a rotation in $F_r$. 
		Since $Y$ contains all predecessors of $r$ in $C$, $Y \cup I_v$ can not separate any edge 
		whose tails are predecessors of $r$. 
		Moreover, by \Cref{thm:semi}, the heads in $F_r$ are incomparable. 
		Therefore, $I_v$ does not contain any rotation in $F_r$. 
		Since $Y$ does not contain any rotation in $F_r$ by the above claim, 
		$Y \cup I_v$ does not separate any edge in $E$.
		It follows that $Y \cup I_v$ generates a matching in $\mathcal{L}_1$.
		Finally, $Y \cup J_v$ separates $rv$ clearly, and hence generates a matching in $\mathcal{L}_2$.
		Therefore, $v \in V$ as desired. 
	\end{proof}
\end{proof}

\begin{theorem}
	\label{thm:algFindFlower}
	\textsc{FindBouquet}$(\Pi)$, given in \Cref{alg:flowerSet}, 
	returns a set of edges defining $\mathcal{L}_1$.
\end{theorem}
\begin{proof}
	From \Cref{lem:correctnessFindNextTail} and \Cref{lem:correctnessFindFlower}, it suffices to show that
	$S$ is udpated correctly in Step 6(b). To be precised, we need that
	\[ S \setminus  \bigcup_{u \in F_r \cup \{ r \}} J'_u \]
	must still be a splitting set, and contains all flowers that have not been found. 
	This follows from \Cref{lem:flower-separating} by noticing that 
	\[ \bigcup_{u \in F_r \cup \{ r \}} J'_u  = \{v \in \Pi: v \succeq u \text{ for some } u \in R_r\}.\]
\end{proof}

Clearly, a sublattice of $\mathcal{L}$ must also be a semi-sublattice. Therefore, \textsc{FindBouquet} can be used to find a canonical path described in \Cref{sec:sublattice}. The same algorithm can be used to check if $M_A \cap M_B = \emptyset$. Let $E$ be the edge set given by the \textsc{FindBouquet} algorithm and $H_E$ be the corresponding graph obtained by adding $E$ to the Hasse diagram of the original rotation poset $\Pi$ of $\mathcal{L}_A$. If $H_E$ has a single strongly connected component, the compression $\Pi'$ has a single meta-element and represents the empty lattice. 

\begin{theorem}
\label{thm:bouquet}
If a lattice $\mathcal{L}$ can be partitioned into a sublattice $\mathcal{L}_1$ and a semi-sublattice $\mathcal{L}_2$, and there is a polynomial-time oracle that determines whether any $x \in \mathcal{L}$ belongs to $\mathcal{L}_1$ or $\mathcal{L}_2$, then an edge set defining $\mathcal{L}_1$ can be found in polynomial time.
\end{theorem}

\subsection{Structure of Robust Stable Matchings}
\label{sec:structure}

In this section, we study the structure of robust stable matchings, $\mathcal{M}_A \cap \mathcal{M}_B$, for two stable matching instances $A$ and $B$. We show how this set decomposes the lattice $\mathcal{L}_A$, and then extend the analysis to the case of multiple instances. We begin with the most basic setting, in which only a single agent changes preferences from instance $A$ to instance $B$. We first restrict this change to a specific operation known as an \emph{upward shift}. We then generalize to the case in which an agent performs an arbitrary permutation of its preference list, and finally to settings in which multiple agents change their preferences in an arbitrary manner between $A$ and $B$.

\subsubsection{Single agent: upward shift}

We begin with the simplest nontrivial model of preference change, in which the preferences of a single agent are modified by an \emph{upward shift}. Informally, an upward shift corresponds to promoting a single option higher in an agent’s preference list while preserving the relative order of all other options.

Formally, let $A$ be a stable matching instance and let $a$ be an agent. An instance $B$ is obtained from $A$ by an \emph{upward shift} on agent $a$'s preference list if there exist two agents $x$ and $y$ such that $a$ strictly prefers $x$ to $y$ in $A$ while $y$ is immediately above $x$ in $a$’s preference list in $A$, and the relative order of all other agents in $a$’s list remains unchanged. All other agents’ preference lists are identical in $A$ and $B$.

In this setting, \Cref{thm:sublattice} already implies that the set of robust stable matchings, $\mathcal{M}_A \cap \mathcal{M}_B$, forms a sublattice of $\mathcal{L}_A$. We now strengthen this structural result by showing that the complementary set of matchings that are stable in $A$ but not in $B$, namely $\mathcal{M}_A \setminus \mathcal{M}_B$, also forms a sublattice of $\mathcal{L}_A$. Together, these results show that an upward shift induces a clean partition of the stable matching lattice into two sublattices.

We first prove a simple observation.

\begin{lemma}
\label{lem:blockingPair}
Let $M \in \mathcal{M}_A \setminus \mathcal{M}_B$. The only blocking pair of $M$ under instance $B$ is $(w,f)$.
\end{lemma}

\begin{proof}
Since $M \notin \mathcal{M}_B$, there exists a blocking pair $(x,y) \notin M$ under $B$. Suppose $(x,y) \neq (w,f)$; we show that $(x,y)$ must also be a blocking pair under $A$.

Let $y'$ be the partner of $x$ in $M$, and let $x'$ be the partner of $y$ in $M$. Since $(x,y)$ is a blocking pair under $B$, we have
$x \succ^{B}_y x'$ and $y \succ^{B}_x y'$.

The preference list of $x$ is unchanged from $A$ to $B$, and hence $y \succ^{A}_x y'$.
We now consider two cases.

\begin{itemize}
    \item If $y \neq f$, then the preference list of $y$ is unchanged between $A$ and $B$, implying
    $x \succ^{A}_y x'$. Thus, $(x,y)$ is a blocking pair under $A$.

    \item If $y = f$, then since $(x,y) \neq (w,f)$, we have $x \neq w$. For all workers $x \neq w$,
    $x \succ^{B}_f x'$ implies $x \succ^{A}_f x'$, because only $w$ is moved upward in $f$’s list.
    Hence, $(x,y)$ is again a blocking pair under $A$.
\end{itemize}

In both cases, $(x,y)$ blocks $M$ under $A$, contradicting the stability of $M$ in instance $A$.
\end{proof}

Recall that $w_1 \geq_f w_2 \geq_f \ldots \geq_f w_k$ are $k$ workers right above $w$ in $f$'s list such that the position of $w$ is shifted up to be above $w_1$ in $B$. From \Cref{lem:blockingPair}, we can then characterize the set $\mathcal{M}_A \setminus \mathcal{M}_B$.

\begin{lemma} \label{lem:characterize}
	$\mathcal{M}_A \setminus \mathcal{M}_B$ is the set of all stable matchings in $A$ that match $f$ to a partner between $w_1$ and $w_k$ in $f$'s list, and match $w$ to a partner below $f$ in $w$'s list.
\end{lemma}
\begin{proof}
	Assume $M$ is a stable matching in $A$ that contains $w_i f$ for $1 \leq i \leq k$ and $wf'$ such that $f >_w f'$. 
	In $B$,
	$f$ prefers $w$ to $w_i$, and hence $wf$ is a blocking pair. Therefore, $M$ is not stable under $B$ and $M \in \mathcal{M}_A \setminus \mathcal{M}_B$.
	
	To prove the other direction, let $M$ be a matching in $\mathcal{M}_A \setminus \mathcal{M}_B$. By \Cref{lem:blockingPair}, $wf$ is the only blocking pair of $M$ in $B$. 
	For that to happen, $p_M(w) <^B_w f$ and $p_M(f) <^B_f w$. We will show that $p_M(f) = w_i$ for $1 \leq i \leq k$. Assume not, then $p_M(f) <^B_f w_k$, and hence, $p_M(f) <^A_f w$. Therefore, $wf$ is a blocking pair in $A$, which is a contradiction. 
\end{proof}

Let $\mathcal{L}_{A}$ be the worker-optimal lattice ( i.e., the lattice where the top element is the worker-optimal stable matching) formed by $\mathcal{M}_{A}$. 

\begin{theorem} \label{cor:sublattice}
	The set $\mathcal{M}_A \setminus \mathcal{M}_B$ forms a sublattice of $\mathcal{L}_A$.
\end{theorem}
\begin{proof}
	Assume $\mathcal{M}_A \setminus \mathcal{M}_B$ is not empty. Let $M_1$ and $M_2$ be two matchings in $\mathcal{M}_A \setminus \mathcal{M}_B$. By \Cref{lem:characterize}, $M_1$ and $M_2$ both match $f$ to a partner between $w_1$ and $w_k$ in $f$'s list, and match $w$ to a partner below $f$ in $w$'s list. Since $M_1 \wedge M_2$ is the matching resulting from having each worker choose the more preferred partner and each firm choose the least preferred partner, $M_1 \wedge M_2$ also belongs to the set characterized by \Cref{lem:characterize}. A similar argument can be applied to the case of $M_1 \vee M_2$.
	Therefore $\mathcal{M}_A \setminus \mathcal{M}_B$ forms a sublattice of $\mathcal{L}_A$. 
\end{proof}

% We will denote the lattice formed by $\mathcal{M}_A \setminus \mathcal{M}_B$ as $\mathcal{L}_{AB}$. 

More detailed structural results for the upward-shift model can be found in~\cite{MV.robust}. Since an upward shift is a special case of an arbitrary permutation of a preference list, we now move to this more general setting.

\subsubsection{Single agent: arbitrary permutation change between \texorpdfstring{$A$}{A} and \texorpdfstring{$B$}{B}}
\label{subsection:semisublatticeNecessarySufficient}

Let $A$ be a stable matching instance, and let $B$ be an instance obtained by arbitrarily permuting the preference list of a single worker or a single firm.
\Cref{lem.eg} provides an example of such a permutation for which $\mathcal{M}_A \setminus \mathcal{M}_B$ is not a sublattice of $\mathcal{L}_A$, demonstrating that the structural guarantees obtained for upward shifts in \Cref{sec:sublattice} do not extend to this more general setting.

Nevertheless, some structure persists. For all such instances $B$, \Cref{lem:MABsemi} shows that $\mathcal{M}_A \setminus \mathcal{M}_B$ always forms a semi-sublattice of $\mathcal{L}_A$. Together with \Cref{thm:sublattice}, this implies that understanding semi-sublattices is both necessary and sufficient for characterizing robust stable matchings under arbitrary single-agent preference changes.

The next lemma pertains to the example given in \Cref{ex:notSublattice}, in which the set of workers is $\mathcal{W} = \{a,b,c,d\}$ and the set of firms is $\mathcal{F} = \{1,2,3,4\}$. Instance $B$ is obtained from instance $A$ by permuting the preference list of firm~$1$.

\begin{lemma} 
\label{lem.eg}
	$\mathcal{M}_A \setminus \mathcal{M}_B$ is not a sublattice of $\mathcal{L}_A$.
\end{lemma}
\begin{proof}
	$M_1 = \{1a,2b,3d,4c\}$ and $M_2 = \{1b,2a,3c,4d\}$ are stable matching with respect to instance $A$.
	Clearly, $M_1 \wedge_A M_2 = \{1a,2b,3c,4d\}$ is also a stable matching under $A$.
	
	In going from $A$ to $B$, the positions of workers $b$ and $c$ are swapped in firm 1's list. 
	Under $B$, $1c$ is a blocking pair for $M_1$ and $1a$ is a blocking pair for $M_2$.
	Hence, $M_1$ and $M_2$ are both in $\mathcal{M}_A \setminus \mathcal{M}_B$.
However, $M_1 \wedge_A M_2$ is a stable matching under $B$, and therefore is not in $\mathcal{M}_A \setminus \mathcal{M}_B$.
Hence, $\mathcal{M}_A \setminus \mathcal{M}_B$ is not closed under the $\wedge_A$ operation.
\end{proof}

\begin{figure}
	\begin{wbox}
	\begin{minipage}{.33\linewidth}
		\centering
		\begin{tabular}{l|llll}
			1  & b & a & c & d \\
			2  & a & b & c & d \\
			3  & d & c & a & b \\
			4  & c & d & a & b
		\end{tabular}
		
		\hspace{1cm}
		
		firms' preferences in $A$ 
		
		\hspace{1cm}
		
	\end{minipage}%
	\begin{minipage}{.34\linewidth}
	\centering
	\begin{tabular}{l|llll}
		\color{red}{1}  & \color{red}{c} & \color{red}{a} & \color{red}{b} & \color{red}{d} \\
		2  & a & b & c & d \\
		3  & d & c & a & b \\
		4  & c & d & a & b
	\end{tabular}
	
	\hspace{1cm}
	
	firms' preferences in $B$
	 
	\hspace{1cm}
\end{minipage}%
	\begin{minipage}{.33\linewidth}
		\centering
		\begin{tabular}{l|llll}
		a  & 1 & 2 & 3 & 4 \\
		b  & 2 & 1 & 3 & 4 \\
		c  & 3 & 1 & 4 & 2 \\
		d  & 4 & 3 & 1 & 2 
		\end{tabular}

		\hspace{1cm}

		workers' preferences in both instances
	\end{minipage} 
	\end{wbox}
	\caption{An example in which $\mathcal{M}_A \setminus \mathcal{M}_B$ is not a sublattice of $\mathcal{L}_A$.}
	\label{ex:notSublattice} 
\end{figure}

\begin{lemma}
	\label{lem:MABsemi}
	For any instance $B$ obtained by permuting the preference list of one worker or one firm,
	$\mathcal{M}_A \setminus \mathcal{M}_B$ forms a semi-sublattice of $\mathcal{L}_A$. 
\end{lemma}
\begin{proof}
	Assume that the preference list of a firm $f$ is permuted. 
	We will show that $\mathcal{M}_A \setminus \mathcal{M}_B$ is a join semi-sublattice of $\mathcal{L}_A$.
	By switching the role of workers and firms, 
	permuting the list of a worker will result in $\mathcal{M}_A \setminus \mathcal{M}_B$ 
	being a meet semi-sublattice of $\mathcal{L}_A$.
	
	Let $M_1$ and $M_2$ be two matchings in $\mathcal{M}_A \setminus \mathcal{M}_B$. 
	Hence, neither of them are in $\mathcal{M}_{B}$. 
	In other words, each has a blocking pair under instance $B$.
	
	Let $w$ be the partner of $f$ in $M_1 \vee_A M_2$.
	Then $w$ must also be matched to $f$ in either $M_1$ or $M_2$ (or both).
	We may assume that $w$ is matched to $f$ in $M_1$. 
	
	Let $xy$ be a blocking pair of $M_1$ under $B$. 
	We will show that $xy$ must also be a blocking pair of $M_1 \vee_A M_2$ under $B$.
	To begin, the firm $y$ must be $f$ since other preference lists remain unchanged. 
	Since $xf$ is a blocking pair of $M_1$ under $B$, $x >_f^B w$.
	Similarly, $f >_x f'$ where $f'$ is the $M_1$-partner of $x$. 
	Let $f''$ be the partner of $x$ in $M_1 \vee_A M_2$.
	Then $f' \geq_x f''$. It follows that $f >_x f''$.
	Since $x >_f^B w$ and $f >_x f''$,
	$xf$ must be a blocking pair of $M_1 \vee_A M_2$ under $B$.
\end{proof}

%Recall that $\mathcal{M}_A \cap \mathcal{M}_B$ forms a sublattice, say $\mathcal{L}'$, of $\mathcal{L}_A$.

\begin{proposition}
	\label{prop:MABcompute}
A set of edges defining the sublattice $\mathcal{L}'$, consisting of matchings in $\mathcal{M}_A \cap \mathcal{M}_B$,
	can be computed efficiently.
\end{proposition}
\begin{proof}
	 We have that $\mathcal{L}'$ and $\mathcal{M}_A \setminus \mathcal{M}_B$ partition $\mathcal{L}_A$,
	 with $\mathcal{M}_A \setminus \mathcal{M}_B$ being a semi-sublattice of $\mathcal{L}_A$, by \Cref{lem:MABsemi}.
	 Therefore,	\textsc{FindBouquet}$(\Pi)$
	 finds a set of edges defining $\mathcal{L}'$ by \Cref{thm:algFindFlower}.
	
	By \Cref{lem:computePoset}, the input $\Pi$ to \textsc{FindBouquet} can be computed in polynomial time. 
	Clearly, a membership oracle checking if a matching is in $\mathcal{L}'$ or 
	not can also be implemented efficiently. 
	Since $\Pi$ has $O(n^2)$ vertices (\Cref{lem:computePoset}), any step of \textsc{FindBouquet} takes polynomial time.	
\end{proof}

This gives us the following \Cref{thm:main} which will be useful in \Cref{sec:lattice} when dealing with multiple agents changing preferences Between $A$ and $B$.

\begin{theorem}
	\label{thm:main}
There is an algorithm for checking if there is a fully robust stable matching w.r.t. any set $S$ in time  $O(|S| p(n))$, where $p$ is a polynomial function, provided all the instances in $S$ are (0,1) w.r.t. some instance $A\in S$. Moreover, if the answer is yes, the set of all such matchings forms a sublattice of $\mathcal{L}_A$ and our algorithm finds a partial order that generates it.
\end{theorem}

\begin{proof}
    
Let $A, B_1, \ldots, B_k$ be polynomially many instances in $S$. Let $E_{i}$ be the set of edges defining $\mathcal{M}_A \cap \mathcal{M}_{B_i}$ for all $1 \leq i \leq k$.
By Corollary~\ref{cor.sublatticeIntersection}, $\mathcal{L}' = \mathcal{M}_A \cap \mathcal{M}_{B_1} \cap \ldots \cap \mathcal{M}_{B_k}$ is a sublattice of $\mathcal{L}_A$.

\begin{lemma}
	\label{lem:poset}
	$E = \bigcup_i E_{i}$ defines $\mathcal{L}'$.
\end{lemma}
\begin{proof}
	By~\Cref{lem:separating}, it suffices to show that for any closed subset $I$,
	$I$ does not separate an edge in $E$ iff $I$ generates a matching in $\mathcal{L}'$.
	
	$I$ does not separate an edge in $E$ iff
	$I$ does not separate any edge in $E_i$  for all $1 \leq i\leq k$ iff
	the matching generated by $I$ is in 
	$\mathcal{M}_A \cap \mathcal{M}_{B_i}$ for all $1 \leq i \leq k$ by~\Cref{lem:separating}. 
\end{proof}

By~\Cref{lem:poset}, a compression $\Pi'$ generating $\mathcal{L}'$ can be constructed from $E$ as described in Section~\ref{sec:alternative}.
By~\Cref{prop:MABcompute}, we can compute each $E_{i}$, and hence, $\Pi'$ efficiently. 
Clearly, $\Pi'$ can be used to check if a fully robust stable matching exists. To be precise, 
a fully robust stable matching exists iff there exists a proper closed subset of $\Pi'$.
This happens iff $s$ and $t$ belong to different meta-rotations in $\Pi'$, an easy to check
condition. Hence, we have~\Cref{thm:main}. 
\end{proof}

\subsubsection{Multiple agents: arbitrary permutation changes}
\label{sec:lattice}

We now consider the most general model of robustness, namely $(p,q)$-robust stable matchings. Formally, we allow $p$ workers and $q$ firms to arbitrarily permute their preference lists between two instances $A$ and $B$, defined on the same agent sets $\mathcal{W}=\{w_1,w_2,\ldots,w_n\}$ and $\mathcal{F}=\{f_1,f_2,\ldots,f_n\}$. We assume $0 \le p \le q \le n$, since all results are symmetric in $p$ and $q$.

The central question in this setting is whether the set of robust stable matchings,
$\mathcal{M}_A \cap \mathcal{M}_B$, forms a sublattice of $\mathcal{L}_A$ and $\mathcal{L}_B$.
As shown in \Cref{thm:sublattice}, when preference changes are restricted to one side of the market—namely, when $(A,B)$ is of type $(0,n)$ or $(n,0)$—the intersection $\mathcal{M}_A \cap \mathcal{M}_B$ is indeed a sublattice of both lattices.

When agents on both sides are allowed to change their preferences, this argument no longer applies.
In particular, $\mathcal{M}_A \setminus \mathcal{M}_B$ need not form a semi-sublattice, as it does in the $(0,1)$ case, and the join and meet operations induced by $\mathcal{L}_A$ and $\mathcal{L}_B$ may differ.
Consequently, for matchings $M,M' \in \mathcal{M}_A \cap \mathcal{M}_B$, it may happen that some workers prefer $M$ to $M'$ in instance $A$, while preferring $M'$ to $M$ in instance $B$—a phenomenon that cannot arise when changes are confined to a single side.

Despite these complications, it turns out that in many instances the intersection
$\mathcal{M}_A \cap \mathcal{M}_B$ still forms a sublattice of both $\mathcal{L}_A$ and $\mathcal{L}_B$.
In such cases, the same underlying set of matchings admits distinct lattice structures under the two instances: while the base set coincides, the lattice operations $(\vee_A,\wedge_A)$ and $(\vee_B,\wedge_B)$ are no longer identical.
An example illustrating this phenomenon appears in~\Cref{ex:both_sides_twisted_lattice}.

\begin{figure}[ht]
	\begin{wbox}
	\centering
	\begin{minipage}{.45\linewidth}
        \centering
		\begin{tabular}{p{0.25cm}|p{0.25cm}p{0.25cm}p{0.25cm}p{0.25cm}p{0.25cm}p{0.25cm}}
			1  & a & b & c & d & e & f\\
			2  & b & a & c & d & e & f\\
			3  & c & d & a & b & e & f\\
			4  & d & c & a & b & e & f\\
			5  & e & f & a & b & c & d\\
			6  & f & e & a & b & c & d
		\end{tabular}
		
		\hspace{0.2cm}
		
		Worker preferences in $A$
	\end{minipage}%
	\begin{minipage}{.45\linewidth}
        \centering
		\begin{tabular}{p{0.20cm}|p{0.20cm}p{0.20cm}p{0.20cm}p{0.20cm}p{0.20cm}p{0.20cm}}
            a  & 2 & 1 & 3 & 4 & 5 & 6\\
			b  & 1 & 2 & 3 & 4 & 5 & 6\\
			c  & 4 & 3 & 1 & 2 & 5 & 6\\
			d  & 3 & 4 & 1 & 2 & 5 & 6\\
			e  & 6 & 5 & 1 & 2 & 3 & 4\\
			f  & 5 & 6 & 1 & 2 & 3 & 4
		\end{tabular}
		
		\hspace{0.2cm}
		
		Firm preferences in $A$
	\end{minipage}

	\hspace{1.0cm}

	\begin{minipage}{.45\linewidth}
        \centering
		\begin{tabular}{p{0.25cm}|p{0.25cm}p{0.25cm}p{0.25cm}p{0.25cm}p{0.25cm}p{0.25cm}}
            \cellcolor{red!25}1  & \cellcolor{red!25}b & \cellcolor{red!25}a & c & d & e & f\\
			\cellcolor{red!25}2  & \cellcolor{red!25}a & \cellcolor{red!25}b & c & d & e & f\\
			3  & c & d & a & b & e & f\\
			4  & d & c & a & b & e & f\\
			5  & e & f & a & b & c & d\\
			6  & f & e & a & b & c & d
		\end{tabular}
		
		\hspace{0.2cm}
		
		Worker preferences in $B$
	\end{minipage}%
	\begin{minipage}{.45\linewidth}
        \centering
		\begin{tabular}{p{0.25cm}|p{0.25cm}p{0.25cm}p{0.25cm}p{0.25cm}p{0.25cm}p{0.25cm}}
            \cellcolor{red!25}a  & \cellcolor{red!25}1 & \cellcolor{red!25}2 & 3 & 4 & 5 & 6\\
			\cellcolor{red!25}b  & \cellcolor{red!25}2 & \cellcolor{red!25}1 & 3 & 4 & 5 & 6\\
			c  & 4 & 3 & 1 & 2 & 5 & 6\\
			d  & 3 & 4 & 1 & 2 & 5 & 6\\
			e  & 6 & 5 & 1 & 2 & 3 & 4\\
			f  & 5 & 6 & 1 & 2 & 3 & 4
		\end{tabular}
		
		\hspace{0.2cm}
		
		Firm preferences in $B$
	\end{minipage}
	\end{wbox}
	\caption{Peculiar example for multiple agents changing preferences. Worker and firm preferences are shown from most preferred to least.}
	\label{ex:both_sides_twisted_lattice}
\end{figure}

Instance $B$ is obtained from instance $A$ by permuting the preferences of workers $1$ and $2$, and firms $a$ and $b$. It can be seen that taking the join (or meet) of any two matchings in $\mathcal{M}_A \cap \mathcal{M}_B$ yields another matching in the intersection, regardless of which instance is considered. However, workers $1$ and $2$ prefer $M_2$ to $M_1$ in $A$, while they prefer $M_1$ to $M_2$ in $B$. \Cref{fig.example} illustrates this change in the partial ordering of matchings in $\mathcal{M}_A \cap \mathcal{M}_B$, depending on the underlying instance.

$M_1 = \{1b, 2a, 3d, 4c, 5e, 6f\}$ and $M_2 = \{1a, 2b, 3c, 4d, 5f, 6e\}$ are both stable under $A$ and $B$. \\
$X_1 = M_1 \wedge_A M_2 = \{1a, 2b, 3c, 4d, 5e, 6f\}$ and $X_2 = M_1 \vee_A M_2 = \{1b, 2a, 3d, 4c, 5f, 6e\}$. \\
$Y_1 = M_1 \wedge_B M_2 = \{1b, 2a, 3c, 4d, 5e, 6f\}$ and $Y_2 = M_1 \vee_B M_2 = \{1a, 2b, 3d, 4c, 5f, 6e\}$.

\begin{figure}[ht]
\centering
\begin{tikzpicture}
\tikzset{
mydot/.style={
  fill,
  circle,
  inner sep=1.5pt
  }
}
\path (0, 0) coordinate (A) (2, 2) coordinate (B) (4, 0) coordinate (C) (2, -2) coordinate (D);
\draw (A)-- (B) node [midway, above]{};
\draw (B)-- (C) node [midway, above]{};
\draw (A)-- (D) node [midway, above]{};
\draw (D)-- (C) node [midway, above]{};

\path (0, -2) coordinate (E) (2, 0) coordinate (F) (4, -2) coordinate (G) (2, -4) coordinate (H);
\draw (E)-- (F) node [midway, above]{};
\draw (F)-- (G) node [midway, above]{};
\draw (E)-- (H) node [midway, above]{};
\draw (H)-- (G) node [midway, above]{};

\draw (A)-- (E) node [midway, above]{};
\draw (B)-- (F) node [midway, above]{};
\draw (C)-- (G) node [midway, above]{};
\draw (D)-- (H) node [midway, above]{};

\node[mydot,label={[align=center]left:$\color{blue}Y_1$}] at (A) {};
\node[mydot,label={[align=center]above:$\cellcolor{red!25}X_1$}] at (B) {};
\node[mydot,label={[align=center]right:$M_2$}] at (C) {};
\node[mydot,label={[align=center]right:$ $}] at (D) {};
\node[mydot,label={[align=center]left:$M_1$}] at (E) {};
\node[mydot,label={[align=center]right:$ $}] at (F) {};
\node[mydot,label={[align=center]right:$\color{blue}Y_2$}] at (G) {};
\node[mydot,label={[align=center]below:$\cellcolor{red!25}X_2$}] at (H) {};

\end{tikzpicture}
\hspace{1cm}
\begin{tikzpicture}
\tikzset{
mydot/.style={
  fill,
  circle,
  inner sep=1.5pt
  }
}
\path (0, 0) coordinate (A) (2, 2) coordinate (B) (4, 0) coordinate (C) (2, -2) coordinate (D);
\draw (A)-- (B) node [midway, above]{};
\draw (B)-- (C) node [midway, above]{};
\draw (A)-- (D) node [midway, above]{};
\draw (D)-- (C) node [midway, above]{};

\path (0, -2) coordinate (E) (2,0) coordinate (F) (4, -2) coordinate (G) (2, -4) coordinate (H);
\draw (E)-- (F) node [midway, above]{};
\draw (F)-- (G) node [midway, above]{};
\draw (E)-- (H) node [midway, above]{};
\draw (H)-- (G) node [midway, above]{};

\draw (A)-- (E) node [midway, above]{};
\draw (B)-- (F) node [midway, above]{};
\draw (C)-- (G) node [midway, above]{};
\draw (D)-- (H) node [midway, above]{};

\node[mydot,label={[align=center]left:$M_1$}] at (A) {};
\node[mydot,label={[align=center]above:$\color{blue}Y_1$}] at (B) {};
\node[mydot,label={[align=center]right:\cellcolor{red!25}$X_1$}] at (C) {};
\node[mydot,label={[align=center]right:$ $}] at (D) {};
\node[mydot,label={[align=center]left:\cellcolor{red!25}$X_2$}] at (E) {};
\node[mydot,label={[align=center]right:$ $}] at (F) {};
\node[mydot,label={[align=center]right:$M_2$}] at (G) {};
\node[mydot,label={[align=center]below:$\color{blue}Y_2$}] at (H) {};

\end{tikzpicture}
\caption{Lattices $\mathcal{L}_A$ and $\mathcal{L}_B$ for example in \Cref{ex:both_sides_twisted_lattice}}
% . The set of robust stable matchings is the same but the 
\label{fig.example}
\end{figure}
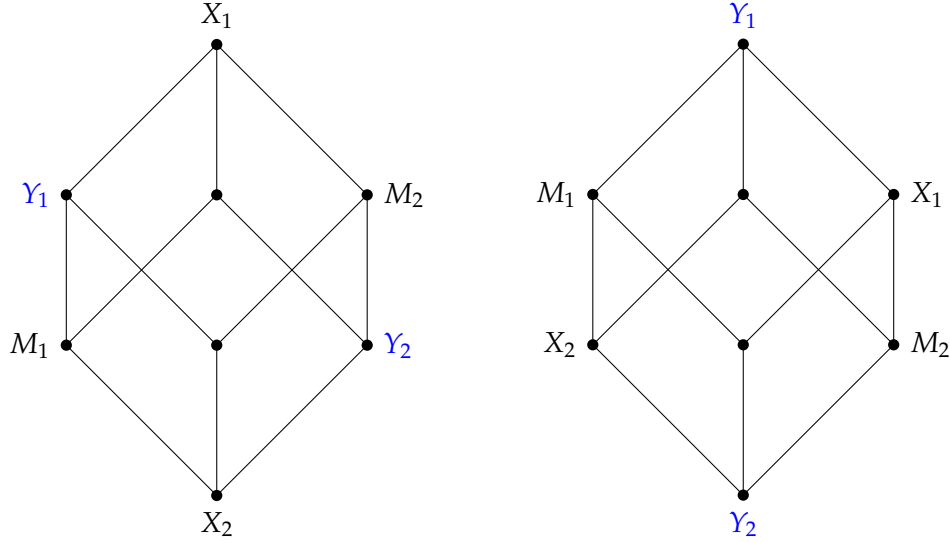

However, as \Cref{thm:both_sides_not_sublattice} shows, the intersection $\mathcal{M}_A \cap \mathcal{M}_B$ is not always a sublattice. An explicit example is provided in \Cref{ex:twoSidesNotSublattice}. Thus, while the sublattice property holds for $(0,n)$, it no longer holds for $(2,2)$.

\begin{restatable}{theorem}{bothsidesnotsublattice}
% ${}^\dagger$
\label{thm:both_sides_not_sublattice}
When $(A, B)$ are of type $(p, q)$ with $p, q \geq 2$, the set $\mathcal{M}_A \cap \mathcal{M}_B$ is not always a sublattice of $\mathcal{L}_A$ and $\mathcal{L}_B$.
\end{restatable}

\begin{proof} 
% [Proof of \Cref{thm:both_sides_not_sublattice}]
The example in \Cref{ex:twoSidesNotSublattice} shows that the intersection is not necessarily a sublattice in $(2, 2)$. Instance $B$ differs from $A$ in the preferences of workers $3$ and $4$ and firms $c$ and $d$.

Consider matchings $M_1 = \{1a, 2b, 3c, 4d\}$ and $M_2 = \{1b, 2a, 3d, 4c\}$ that are stable under $A$ and $B$. However, $M = M_1 \wedge_A M_2 = \{1a, 2b, 3d, 4c\}$ is not stable under $B$, as $a >_4 c$ and $4 >_a 1$, so $4a$ is a blocking pair of $M$ under $B$. 

Therefore, if there are at least two agents on each side having different preferences in $A$ and $B$, $\mathcal{M}_A \cap \mathcal{M}_B$ is not a sublattice of $\mathcal{L}_B$. 
\end{proof}

This raises the question: where does the sublattice property break down? We show that as long as no more than one agent changes preferences on one side, the sublattice property still holds. In particular, we establish that for $(1,n)$, the intersection $\mathcal{M}_A \cap \mathcal{M}_B$ remains a sublattice.

\begin{figure}[ht]
	\begin{wbox}
	\centering
	\begin{minipage}{.45\linewidth}
        \centering
		\begin{tabular}{p{0.5cm}|p{0.5cm}p{0.5cm}p{0.5cm}p{0.5cm}}
			1  & a & b & c & d \\
			2  & b & c & a & d \\
			3  & d & c & a & b \\
			4  & c & d & a & b
		\end{tabular}
		
		\hspace{1cm}
		
		Worker preferences in $A$ 
		
		\hspace{1cm}
		
	\end{minipage}%
	\begin{minipage}{.45\linewidth}
    \centering
	\begin{tabular}{p{0.5cm}|p{0.5cm}p{0.5cm}p{0.5cm}p{0.5cm}}
		a  & 2 & 4 & 1 & 3 \\
		b  & 1 & 2 & 3 & 4 \\
		c  & 3 & 4 & 2 & 1 \\
		d  & 4 & 3 & 2 & 1
	\end{tabular}
	
	\hspace{1cm}
	
	Firm preferences in $A$
	 
	\hspace{1cm}
\end{minipage}%

	\begin{minipage}{.45\linewidth}
        \centering
		\begin{tabular}{p{0.5cm}|p{0.5cm}p{0.5cm}p{0.5cm}p{0.5cm}}
		1  & a & b & c & d \\
		2  & b & c & a & d \\
		\cellcolor{red!25}3  & \cellcolor{red!25}c & \cellcolor{red!25}d & \cellcolor{red!25}a & \cellcolor{red!25}b \\
		\cellcolor{red!25}4  & \cellcolor{red!25}d & \cellcolor{red!25}a & \cellcolor{red!25}c & \cellcolor{red!25}b 
		\end{tabular}

		\hspace{1cm}

		Worker preferences in $B$
	\end{minipage} 
	\begin{minipage}{.45\linewidth}
        \centering
		\begin{tabular}{p{0.5cm}|p{0.5cm}p{0.5cm}p{0.5cm}p{0.5cm}}
		a  & 2 & 4 & 1 & 3 \\
		b  & 1 & 2 & 3 & 4 \\
		\cellcolor{red!25}c  & \cellcolor{red!25}4 & \cellcolor{red!25}2 & \cellcolor{red!25}3 & \cellcolor{red!25}1 \\
		\cellcolor{red!25}d  & \cellcolor{red!25}3 & \cellcolor{red!25}4 & \cellcolor{red!25}1 & \cellcolor{red!25}2 
		\end{tabular}

		\hspace{1cm}

		Firm preferences in $B$
	\end{minipage} 
	\end{wbox}
        \caption{$\mathcal{M}_A \cap \mathcal{M}_B$ is not a sublattice of $\mathcal{L}_A$. Preference lists are in decreasing order. The agents whose preferences change from $A$ to $B$ are marked in red. This notation is used for all examples.}
	\label{ex:twoSidesNotSublattice}
\end{figure}

Let $A$ and $B$ be a $(1,n)$ robust stable matching instance where only $w_1$ and all firms change their preferences from $A$ to $B$.

\begin{restatable}{theorem}{sublatticetwoside}
\label{thm:sublattice_two_side}
If $A$ and $B$ are of type $(1,n)$, then $\mathcal{M}_A \cap \mathcal{M}_B$ is a sublattice of both $\mathcal{L}_A$ and $\mathcal{L}_B$. 
\end{restatable}

\begin{proof}
Without loss of generality, assume $|\mathcal{M}_A \cap \mathcal{M}_B| > 1$, and let $M_1$ and $M_2$ be two distinct matchings in $\mathcal{M}_A \cap \mathcal{M}_B$. Let $\vee_A$ and $\vee_B$ denote the join operations under $A$ and $B$, respectively. Likewise, let $\wedge_A$ and $\wedge_B$ denote the meet operations under $A$ and $B$.

The join $M_1 \vee_A M_2$ is obtained by assigning each worker to their less preferred partner from $M_1$ and $M_2$, according to instance $A$. Since the preferences of workers $w_2, w_3, \ldots, w_n$ are identical in $A$ and $B$, their less preferred partners under $A$ and $B$ coincide. Thus, for all $w \neq w_1$, we have:
$$(M_1 \vee_A M_2)(w) = (M_1 \vee_B M_2)(w).$$

Since $M_1 \vee_A M_2$ and $M_1 \vee_B M_2$ are both stable matchings under their respective instances, they are also perfect matchings. Hence, the remaining worker $w_1$ must be matched to the same firm in both matchings:
$$(M_1 \vee_A M_2)(w_1) = (M_1 \vee_B M_2)(w_1).$$

Therefore, $M_1 \vee_A M_2 = M_1 \vee_B M_2$. A similar argument shows that $M_1 \wedge_A M_2 = M_1 \wedge_B M_2$. Hence, the join and meet operations under instances $A$ and $B$ are equivalent, and both $M_1 \vee_A M_2$ and $M_1 \wedge_A M_2$ belong to $\mathcal{M}_A \cap \mathcal{M}_B$. The theorem follows.
\end{proof}

Thus, when the instances are of type $(1,n)$, \Cref{thm:generalization} guarantees the existence of a poset that generates the corresponding sublattice. However, unlike the $(0,1)$ case, \Cref{thm:bouquet} does not apply to instances of type $(0,n)$.

\begin{restatable}{lemma}{lemnnotsemisublattice}
% ${}^\dagger$
\label{lem:0_n_not_semi_sublattice}
If $A$ and $B$ are of type $(0,n)$, then $\mathcal{M}_A \setminus \mathcal{M}_B$ is not always a semi-sublattice of $\mathcal{L}_A$.
\end{restatable}

\begin{proof}
% [Proof of \Cref{lem:0_n_not_semi_sublattice}]
Consider instances $A$ and $B$ as given in \Cref{ex:oneSideNotSemiSubLattice}. The set of workers is $\mathcal{W} = \{1, 2, 3, 4, 5\}$ and the set of firms is $\mathcal{F} = \{a, b, c, d, e\}$. Instance $B$ is obtained from instance $A$ by permuting the lists of firms $b$ and $c$. Consider matchings $M_1 = \{1b, 2a, 3c, 4d, 5e\}$ and $M_2 = \{1a, 2b, 3d, 4c, 5e\}$ that are both stable matchings with respect to instance $A$. With respect to instance $B$, $M_1$ has the blocking pair $(5, c)$ and $M_2$ has the blocking pair $(4, b)$. Hence, $M_1$ and $M_2$ are both in $\mathcal{M}_{A} \setminus \mathcal{M}_B$. However, the \emph{join} with respect to $A$, $M_1 \vee_A M_2 = \{1b, 2a, 3d, 4c, 5e\}$, is stable with respect to both $A$ and $B$.

Similarly, we see that $M_3 = \{1b, 2c, 3d, 4a, 5e\}$ and $M_4 = \{1d, 2a, 3b, 4c, 5e\}$ are stable with respect to $A$, but have blocking pairs $(1, a)$ and $(4, d)$ respectively, under $B$. However, their \emph{meet}, $M_3 \wedge_A M_4 = \{1b, 2a, 3d, 4c, 5e\}$, is stable with respect to both $A$ and $B$.

Hence, $\mathcal{M}_{A} \setminus \mathcal{M}_B$ is not a semi-sublattice of $\mathcal{L}_A$.  
\end{proof}

Note that, since $A$ and $B$ in \Cref{ex:oneSideNotSemiSubLattice} are actually $(0,2)$, this provides a tight bound for when $\mathcal{M}_A \setminus \mathcal{M}_B$ is a semi-sublattice of $\mathcal{L}_A$.

\begin{figure}[ht]
	\begin{wbox}
	\begin{minipage}{.4\linewidth}
        \centering
		\begin{tabular}{p{0.20cm}|p{0.20cm}p{0.25cm}p{0.25cm}p{0.25cm}p{0.25cm}}
			1  & a & b & d & c & e \\
			2  & b & a & c & d & e \\
			3  & c & d & b & a & e \\
			4  & b & d & c & a & e \\
            5  & c & e & a & b & d
		\end{tabular}
		
		\hspace{1cm}
		
		Worker preferences in $A$ and $B$
		
		\hspace{1cm}
		
	\end{minipage}%
	\begin{minipage}{.33\linewidth}
        \centering
	\begin{tabular}{l|lllll}
		a  & 4 & 2 & 1 & 3 & 5 \\
		b  & 3 & 1 & 2 & 4 & 5 \\
		c  & 2 & 4 & 3 & 1 & 5 \\
		d  & 1 & 3 & 4 & 2 & 5 \\
        e  & 5 & 1 & 2 & 3 & 4
	\end{tabular}
	
	\hspace{1cm}
	
	Firm preferences in $A$
	 
	\hspace{1cm}
 
        \end{minipage}%
	\begin{minipage}{.33\linewidth}
        \centering
		\begin{tabular}{l|lllll}
		a  & 4 & 2 & 1 & 3 & 5 \\
		\cellcolor{red!25}b  & \cellcolor{red!25}1 & \cellcolor{red!25}4 & \cellcolor{red!25}2 & \cellcolor{red!25}3 & \cellcolor{red!25}5 \\
            \cellcolor{red!25}c  & \cellcolor{red!25}4 & \cellcolor{red!25}5 & \cellcolor{red!25}3 & \cellcolor{red!25}1 & \cellcolor{red!25}2 \\
            d  & 1 & 3 & 4 & 2 & 5 \\
            e  & 5 & 1 & 2 & 3 & 4 
		\end{tabular}

		\hspace{1cm}

		Firm preferences in $B$

            \hspace{1cm}
	\end{minipage} 
	\end{wbox}
	\caption{$\mathcal{M}_{A} \setminus \mathcal{M}_B$ is not a semi-sublattice of $\mathcal{L}_A$.}
	\label{ex:oneSideNotSemiSubLattice} 
\end{figure} 

This shows that \Cref{thm:bouquet} cannot be used directly to find Birkhoff’s partial order in the $(0,n)$ setting. To circumvent this, we define $n$ hybrid instances $B_1, \dots, B_n$ such that for each $i \in \{1, 2, \ldots, n\}$, the pair $(A, B_i)$ is of type $(0,1)$. Specifically, the preference profile of agent $x$ in instance $B_i$ is defined as:
\begin{equation*}
>_x^{B_i} = 
\left \{
    \begin{array}{ll}
        >_x^B, & \text{if } x = f_i \\
        >_x^A, & \text{if } x \neq f_i
    \end{array}
\right\}
\end{equation*}
That is, instance $B_i$ is identical to instance $A$, except that the preferences of firm $f_i$ are replaced by those in instance $B$.

Defined this way, \Cref{thm:breaking_instances_one_side} shows that the matchings stable under both $A$ and $B$ are precisely those stable under $A$ and all the hybrid instances $B_1, B_2, \ldots, B_n$.

\begin{restatable}{theorem}{breakinginstancesoneside}
% ${}^\dagger$
\label{thm:breaking_instances_one_side}
$\mathcal{M}_A \cap \mathcal{M}_B = \mathcal{M}_A \cap \mathcal{M}_{B_1} \cap \mathcal{M}_{B_2} \cap \ldots \cap \mathcal{M}_{B_n}$.
\end{restatable}

\begin{proof}
% [Proof of \Cref{thm:breaking_instances_one_side}]
We will show that any blocking pair for a perfect matching under the preference lists of instance $A$ or instance $B$ is a blocking pair under the preference lists of at least one of the instances $A$ or $B_1, \dots, B_n$, and vice versa.

Let $M$ be a perfect matching with a blocking pair under the preference lists in $A$ or in $B$. If it has a blocking pair under $A$, then $M$ is clearly not in $\mathcal{M}_{A} \cap \mathcal{M}_{B_1} \cap \ldots \cap \mathcal{M}_{B_n}$. Suppose $(w, f)$ is a blocking pair for $B$ and $f >_w^B M(w)$, and so $w >_f^B M(f)$. Workers have identical lists in all instances, so $w$ prefers $f$ to $M(w)$ in every instance $B_i$. Firm $f$ has the same preference list in $B$ and in $B_k$ for some $k \in \{1, \dots, n\}$. Then $w >_f^{B_k} M(f)$, and $(w, f)$ is a blocking pair for instance $B_k$, so $M$ is not a stable matching in $\mathcal{M}_{A} \cap \mathcal{M}_{B_1} \cap \ldots \cap \mathcal{M}_{B_n}$.

Let $M$ be a perfect matching with blocking pair $(w, f)$ under the preference lists in one of the instances $A, B_1, \dots, B_n$. As in the previous case, if $(w, f)$ is a blocking pair under the preference lists of $A$, then $M$ is not a stable matching in $\mathcal{M}_A \cap \mathcal{M}_B$. Otherwise, $(w, f)$ is a blocking pair for some instance $B_k$. We will show that $(w, f)$ must be a blocking pair for $B$. 

Assume that $f$ is not the agent whose preference list is identical in $B$ and $B_k$. Then both $w$ and $f$ have identical preference lists in $A$ and in $B_k$, and so this is a blocking pair for $A$. Otherwise, it must be the case that $f$ changed its preference list in $B_k$ and has the same changed list in $B$. Then $(w, f)$ is a blocking pair for $B$ as well, as $f >_w^{B} M(w)$ and $w >_f^{B} M(f)$. We conclude that the two sets are equal. 
\end{proof}

\begin{restatable}{theorem}{birkhoff}
% ${}^\dagger$
\label{thm:birkhoff_one_side}
If $A$ and $B$ are of type $(0,n)$, then Birkhoff’s partial order generating $\mathcal{M}_A \cap \mathcal{M}_B$ can be computed efficiently, and the matchings in this set can also be enumerated with polynomial delay.
\end{restatable}

\begin{proof}
Since each $B_i$ differs from instance $A$ at exactly one agent, by \Cref{prop:MABcompute}, edge sets $E_i$ defining $\mathcal{M}_A \cap \mathcal{M}_{B_i}$ can be computed efficiently. By \Cref{lem:poset} and \Cref{thm:main}, the union $E = \bigcup_i E_i$ defines the sublattice $\mathcal{M}_A \cap \mathcal{M}_B$, and can be used to enumerate all matchings in this intersection efficiently.
\end{proof}

A similar technique can be attempted to the case where $(A, B)$ is of type $(1,n)$. Let $w_1 \in \mathcal{W}$ and $f_1, \dots, f_n \in \mathcal{F}$ be the worker and firms whose preferences differ between $A$ and $B$. For any subset $X \subseteq \mathcal{W} \cup \mathcal{F}$, let $B_X$ denote the instance in which each agent in $X$ has the same preferences as in $B$, and all other agents have the same preferences as in $A$. Then:

\begin{restatable}{theorem}{breakinginstancestwoside}
% ${}^\dagger$
\label{thm:breaking_instances_two_side}
$\mathcal{M}_A \cap \mathcal{M}_B = \mathcal{M}_A \cap \mathcal{M}_{B_{\{w_1, f_1\}}} \cap \mathcal{M}_{B_{\{w_1, f_2\}}} \cap \dots \cap \mathcal{M}_{B_{\{w_1, f_n\}}}$.
\end{restatable}

\begin{proof}
% [Proof of \Cref{thm:breaking_instances_two_side}]
    $(\supseteq)$
    Let $\mu \notin \mathcal{M}_A \cap \mathcal{M}_B$. Then $\mu$ has a blocking pair with respect to $A$ or $B$. If it is $A$, then it is trivially not in $\mathcal{M}_A \cap \mathcal{M}_{B_{\{w_1, f_1\}}} \cap \mathcal{M}_{B_{\{w_1, f_2\}}} \cap \dots \cap \mathcal{M}_{B_{\{w_1, f_n\}}}$. So suppose it has a blocking pair $(w, f)$ with respect to $B$. 
    
    Suppose $w = w_1$ and $f \in \{f_1, \dots, f_n\}$. So $f >_{w}^B \mu(w)$ and $w >_{f}^B \mu(f)$. But these are the preferences of $w$ and $f$ in $B_{\{w, f\}}$, so $\mu \notin \mathcal{M}_A \cap \mathcal{M}_{B_{\{w_1, f_1\}}} \cap \mathcal{M}_{B_{\{w_1, f_2\}}} \cap \dots \cap \mathcal{M}_{B_{\{w_1, f_n\}}}$. 

    Else, if $w \neq w_1$ and $f \in \{f_1, \dots, f_n\}$, then $w$ has the same preferences in $A$ and $B$, as well as in every $B_{\{w_1, f\}}$, while $f$ has the same preferences in $B$ and $B_{\{w_1, f\}}$. Hence, $(w, f)$ is a blocking pair for instance $B_{w_1, f}$. 
    
    $(\subseteq)$
    Similarly, consider $\mu \notin \mathcal{M}_A \cap \mathcal{M}_{B_{w_1, f_1}} \cap \dots \cap \mathcal{M}_{B_{w_1, f_n}}$. Then $\mu$ has a blocking pair $(w, f)$, and as in the previous case, assume the blocking pair is not with respect to $A$. Then it is with respect to some instance $B_{\{w_1, f_i\}}$. 

    There are three cases:
    \begin{enumerate} 
        \item $w \neq w_1$ and $f = f_i$: Then $f_i$ has the same list in $B$ and $B_{w_1, f_i}$, $w$ has the same list in $A$, $B$, and $B_{w_1, f_i}$; hence it is a blocking pair with respect to $B$.
        
        \item $w = w_1$ and $f \neq f_i$: Then $w_1$ has the same list in $A$, $B$, and $B_{\{w_1, f_i\}}$, and $f$ has the same list in $A$, $B$, and $B_{\{w_1, f_i\}}$; hence it is a blocking pair with respect to $B$.
        
        \item $w = w_1$ and $f = f_i$: Then $w_1$ has the same list in $B$ and $B_{\{w_1, f_i\}}$, and $f$ has the same list in $B$ and $B_{\{w_1, f_i\}}$; hence it is a blocking pair with respect to $B$.
    \end{enumerate}
    The other case of $w \neq w_1$ and $f \neq f_i$ does not arise because then $(w,f)$ would be a blocking pair of $A$, which cannot be, as $\mu \in \mathcal{M}_A$. 
    
    Hence, $(w, f)$ is a blocking pair with respect to $B$, and $\mu \notin \mathcal{M}_B$; hence it is not in $\mathcal{M}_A \cap \mathcal{M}_B$.

    We conclude that the sets are the same. 
\end{proof}

\begin{restatable}{corollary}{birkhoffonen}
% ${}^\dagger$
\label{cor:birkhoff_one_n}
Computing Birkhoff’s partial order for the $(1,n)$ case is in $\mathcal{P}$ if and only if computing it for the $(1,1)$ case is in $\mathcal{P}$.
\end{restatable}

Unlike the $(0,1)$ case, the set $\mathcal{M}_A \setminus \mathcal{M}_B$ is, by \Cref{lem:1_1_not_semi_sublattice}, not always a semi-sublattice when $(A, B)$ is of type $(1,1)$. As a result, \Cref{thm:bouquet} cannot be applied, and characterizing the sublattice structure remains an open problem.

\begin{restatable}{lemma}{lnotsemisublattice}
% ${}^\dagger$
\label{lem:1_1_not_semi_sublattice}
If $A$ and $B$ are of type $(1,1)$, then $\mathcal{M}_A \setminus \mathcal{M}_B$ is not always a semi-sublattice of $\mathcal{L}_A$.
\end{restatable}

\begin{proof}
% [Proof of \Cref{lem:1_1_not_semi_sublattice}]
Consider instances $A$ and $B$ as given in \Cref{ex:1-1notSemisublattice}. The set of workers is $\mathcal{W} = \{1, 2, 3, 4, 5\}$ and the set of firms is $\mathcal{F} = \{a, b, c, d, e\}$. Instance $B$ is obtained from instance $A$ by permuting the lists of firm $c$ and worker $3$. Consider matchings $M_1 = \{1b, 2c, 3a, 4d, 5e\}$ and $M_2 = \{1a, 2b, 3c, 4e, 5d\}$ that are both stable with respect to $A$. With respect to instance $B$, $M_1$ has the blocking pair $(3, d)$ and $M_2$ has the blocking pair $(4, c)$. Hence, $M_1$ and $M_2$ are both in $\mathcal{M}_{A} \setminus \mathcal{M}_B$. However, the \emph{join} with respect to $A$, $M_1 \vee_A M_2 = \{1a, 2b, 3c, 4d, 5e\}$, is stable with respect to $A$ and $B$. 

Similarly, we see that $M_3 = \{1c, 2a, 3b, 4d, 5e\}$ and $M_4 = \{1b, 2c, 3a, 4e, 5d\}$ are stable with respect to $A$ but have blocking pairs $(3, d)$ and $(4, c)$ respectively, under $B$. However, their \emph{meet}, $M_3 \wedge_A M_4 = \{1c, 2a, 3b, 4e, 5d\}$, is stable with respect to $A$ and $B$. 

Hence, $\mathcal{M}_{A} \setminus \mathcal{M}_B$ is not a semi-sublattice of $\mathcal{L}_A$. 
\end{proof}

\begin{figure}[ht]
	\begin{wbox}
	\centering
	\begin{minipage}{.45\linewidth}
        \centering
		\begin{tabular}{p{0.25cm}|p{0.25cm}p{0.25cm}p{0.25cm}p{0.25cm}p{0.25cm}}
			1  & a & b & c & d & e \\
			2  & b & c & a & d & e \\
			3  & c & a & b & d & e \\
			4  & d & c & e & a & b \\
            5  & e & d & a & b & c
		\end{tabular}
		
		\hspace{1cm}
		
		Worker preferences in $A$ 
		
		\hspace{1cm}
		
	\end{minipage}%
	\begin{minipage}{.45\linewidth}
    \centering
	\begin{tabular}{p{0.25cm}|p{0.25cm}p{0.25cm}p{0.25cm}p{0.25cm}p{0.25cm}}
		a  & 2 & 3 & 1 & 4 & 5 \\
		b  & 3 & 1 & 2 & 4 & 5 \\
		c  & 1 & 2 & 3 & 4 & 5 \\
		d  & 5 & 3 & 4 & 1 & 2 \\
        e  & 4 & 5 & 1 & 2 & 3
	\end{tabular}
	
	\hspace{1cm}
	
	Firm preferences in $A$
	 
	\hspace{1cm}
\end{minipage}%

	\begin{minipage}{.45\linewidth}
        \centering
		\begin{tabular}{p{0.25cm}|p{0.25cm}p{0.25cm}p{0.25cm}p{0.25cm}p{0.25cm}}
		1  & a & b & c & d & e \\
		2  & b & c & a & d & e \\
		\cellcolor{red!25}3  & \cellcolor{red!25}c & \cellcolor{red!25}d & \cellcolor{red!25}a & \cellcolor{red!25}b & \cellcolor{red!25}e \\
		4  & d & c & e & a & b \\
        5  & e & d & a & b & c 
		\end{tabular}

		\hspace{1cm}

		Worker preferences in $B$
	\end{minipage} 
	\begin{minipage}{.45\linewidth}
        \centering
		\begin{tabular}{p{0.25cm}|p{0.25cm}p{0.25cm}p{0.25cm}p{0.25cm}p{0.25cm}}
		a  & 2 & 3 & 1 & 4 & 5 \\
		b  & 3 & 1 & 2 & 4 & 5 \\
		\cellcolor{red!25}c  & \cellcolor{red!25}1 & \cellcolor{red!25}4 & \cellcolor{red!25}3 & \cellcolor{red!25}2 & \cellcolor{red!25}5 \\
		d  & 5 & 3 & 4 & 1 & 2 \\
        e  & 4 & 5 & 1 & 2 & 3
		\end{tabular}

		\hspace{1cm}

		Firm preferences in $B$
	\end{minipage} 
	\end{wbox}
	\caption{ $\mathcal{M}_{A} \setminus \mathcal{M}_B$ is not a semi-sublattice of $\mathcal{L}_A$, even when $(A, B)$ is of type $(1,1)$.}
	\label{ex:1-1notSemisublattice} 
\end{figure}

\subsection{Algorithms to Compute Worker- and Firm-Optimal Robust Stable Matchings}
\label{sec:optimal_matchings}

\Cref{thm:sublattice_two_side} shows that when the instance $(A,B)$ is of type $(1,n)$, the set of robust stable matchings forms a sublattice of both $\mathcal{L}_A$ and $\mathcal{L}_B$. This structural property implies the existence of well-defined worker-optimal and firm-optimal robust stable matchings. In this section, we present algorithms to compute these extremal matchings.

The rotation poset provides a powerful description of how to traverse a stable matching lattice by characterizing the minimal transformations between matchings. However, it does not by itself yield a robust stable matching. Consequently, it is necessary to first compute at least one robust stable matching using an explicit algorithm; once such a matching is obtained, the remaining robust stable matchings can be generated by applying rotations.

For the $(0,n)$ case—when preference changes are confined to a single side of the market—existing algorithms from the literature can be adapted to compute worker-optimal and firm-optimal robust stable matchings. We discuss these adaptations in Appendix~\ref{app:one_side}. These approaches, however, do not extend to the $(1,n)$ setting, where both sides of the market are involved. For this case, we design new algorithms based on Gale and Shapley’s Deferred Acceptance (DA) algorithm~\cite{GaleS}.

The Deferred Acceptance Algorithm (\Cref{alg:gs_algorithm}) proceeds in iterations. In each iteration, the following steps occur:  
$(a)$ The proposing side (e.g., workers) proposes to their most preferred firm that has not yet rejected them.  
$(b)$ Each firm tentatively accepts its most preferred proposal received in that round and rejects all others.  
$(c)$ Each worker eliminates the firms that rejected them from their preference list.  

\begin{algorithm}[ht]
	\begin{wbox}
		\textsc{DeferredAcceptanceAlgorithm}($A$): \\
		\textbf{Input:} Stable matching instance $A$ \\
		\textbf{Output:} Perfect matching $M$ \\
        Each worker has a list of all firms in decreasing order of priority.
		\begin{enumerate}
			\item Until all firms receive a proposal, do:
			\begin{enumerate}
				\item $\forall w \in W$, $w$ proposes to their best uncrossed firm.
				\item $\forall f \in F$, $f$ tentatively accepts their \textit{\textbf{best}} proposal and rejects the rest.
				\item $\forall w \in W$, if $w$ is rejected by a firm $f$, they cross $f$ off their list.
			\end{enumerate}  
			\item Output the perfect matching $M$.
		\end{enumerate}
	\end{wbox}
	\caption{Workers proposing Deferred Acceptance Algorithm, by \cite{GaleS} (from \cite{MM.Book-Online}) }
	\label{alg:gs_algorithm} 
\end{algorithm}

The process continues until a perfect matching is formed, at which point the algorithm outputs it as a stable matching. The key idea is that whenever a rejection occurs, the corresponding worker-firm pair can never be part of any stable matching. We modify this algorithm to compute worker- and firm-optimal stable matchings in the intersection $\mathcal{M}_A \cap \mathcal{M}_B$ for the $(1,n)$ setting.

For the $(1,n)$ case, the worker-optimal \Cref{alg:daalgorithm_worker_both_sides} runs the DA algorithm simultaneously in multiple rooms, where the preference lists for each room correspond to individual instances. Agents act according to preferences in each room: in every iteration, workers propose and firms tentatively accept their best proposals and reject the rest. However, the rejections apply \emph{across all rooms}. For instance, if a worker is rejected by a firm in room $\mathcal{R}_A$, they remove that firm from their list in \emph{all} rooms. This synchronization ensures that if a perfect matching is returned, it is stable with respect to all instances. The primary technical challenge is the possibility that the algorithm produces distinct perfect matchings in different rooms. \Cref{lem:nomultmatchings} shows this cannot occur, a fact that follows from the intersection $\mathcal{M}_A \cap \mathcal{M}_B$ having the same partial order in both lattices.

The firm-optimal \Cref{alg:daalgorithm_firm_both_sides} proceeds analogously, with firms proposing and workers tentatively accepting. A key distinction is that workers issue \emph{preemptive} rejections: in each room, a worker rejects \emph{all} firms ranked below their current tentative match—including firms that have not yet proposed—based on their local preference list. This ensures the resulting perfect matching is identical in all rooms.

\begin{restatable}{theorem}{nalgoworks}
% ${}^\dagger$
\label{thm:n_1_algo_works}
Let $A$ and $B$ be two stable matching instances that are of type $(1,n)$. Then \Cref{alg:daalgorithm_worker_both_sides} and \Cref{alg:daalgorithm_firm_both_sides} find the worker- and firm-optimal robust stable matchings, respectively, or correctly report that no such matching exists.
\end{restatable}

\begin{algorithm}[ht]
	\begin{wbox}
		\textsc{WorkerOptimal}($A,B$): \\
		\textbf{Input:} Stable matching instances $A$ and $B$ on agents $W \cup F$.\\
		\textbf{Output:}  Perfect matching $M$ or $\boxtimes$ (when no robust stable matching exists).\\
        Assume there are two rooms, $\mathcal{R}_A$ and $\mathcal{R}_B$ corresponding to instances $A$ and $B$ respectively. Each worker has a list, initialized to all firms, that they look at while proposing.
		\begin{enumerate}
    		% \item $\forall w \in W, w$ maintains a set of all the firms that haven't rejected it yet; Initially $\mathcal{F}$;
			% \item $\forall f \in F, f$ maintains a set of all the workers that proposed to it so far; Initially $\emptyset$;
			\item Until there is no rejection in any room or some worker is rejected by all firms, do:
			\begin{enumerate}
				    \item $\forall I \in \{A,B\}, \forall w \in W: w$ proposes to their best-in-$I$ uncrossed firm in $\mathcal{R}_I$.
				    \item $\forall I \in \{A,B\}, \forall f \in F: f$ tentatively accepts their best-in-$I$ proposal in room $\mathcal{R}_I$ and rejects the rest.
				    \item $\forall w \in W:$ if $w$ is rejected by a firm $f$ in \textbf{\textit{any}} room, they cross $f$ off their list. 
			\end{enumerate}  
			\item 
                \begin{enumerate}
                    \item If some worker is rejected by all firms, output $\boxtimes$.
                    \item Else, the acceptances define perfect matchings in each room. If they are the same in all rooms, output the perfect matching ($M$). 
                \end{enumerate}
		\end{enumerate}
	\end{wbox}
	\caption{Algorithm to find the worker-optimal stable matching. Note that proposers (workers) maintain a single list (or set) across rooms. They may propose to different firms in each room but update their list synchronously based on rejections from any room.}
	\label{alg:daalgorithm_worker_both_sides} 
\end{algorithm} 

\begin{algorithm} [ht]
	\begin{wbox}
		\textsc{FirmOptimal}($A,B$): \\
		\textbf{Input:} Stable matching instances $A$ and $B$ on agents $W \cup F$.\\
		\textbf{Output:}  Perfect matching $M$ or $\boxtimes$ (when no robust stable matching exists). \\
        Assume there are two rooms, $\mathcal{R}_A$ and $\mathcal{R}_B$ corresponding to instances $A$ and $B$ respectively. Each firm has a list, initialized to all workers, that they look at while proposing.
		\begin{enumerate}
			\item Until there is no rejection in any room or some firm is rejected by all firms, do:
			\begin{enumerate}
				    \item $\forall I \in \{A,B\}, \forall f \in F: f$ proposes to their best-in-$I$ uncrossed worker in $\mathcal{R}_I$.
				    \item $\forall I \in \{A,B\}, \forall w \in W: w$ tentatively accepts their best-in-$I$ proposal (call it $f^I_w$) in room $\mathcal{R}_I$ and rejects the rest.
                        \item $\forall I \in \{A,B\}, \forall w \in W: w$
                        sends \textbf{\textit{preemptive rejections}} to all firms $f'$ not yet rejected, that are worse than their current option in $\mathcal{R}_I$, i.e., $f^I_w \geq^I_w f'$.
				    \item $\forall f \in F:$ if $f$ is rejected by a worker $w$ in some room, they cross $w$ off their list. 
			\end{enumerate}  
			\item 
                \begin{enumerate}
                    \item If some firm is rejected by all workers, output $\boxtimes$.
                    \item Else, the acceptances define perfect matchings in each room. If they are the same in all rooms, output the perfect matching($M$). 
                \end{enumerate}
		\end{enumerate}

	\end{wbox}
        \caption{Algorithm to find the firm-optimal stable matching. Note that in step 1(c), workers may send rejections to firms that have not yet proposed to them.}
	\label{alg:daalgorithm_firm_both_sides} 
\end{algorithm} 

The proof of \Cref{thm:n_1_algo_works} is split between
\Cref{sec:daalgorithm_worker_both_sides} and
\Cref{sec:daalgorithm_firm_both_sides}.

\subsubsection{Proof of Correctness of \texorpdfstring{\Cref{alg:daalgorithm_worker_both_sides}}{Algorithm~2}}
\label{sec:daalgorithm_worker_both_sides}

We analyze the worker-proposing version of the algorithm for the $(1,n)$ setting, in which exactly one worker changes preferences between instances $A$ and $B$. The algorithm runs the Deferred Acceptance algorithm on a collection of appropriately defined instances and returns the worker-optimal robust stable matching in $\mathcal{M}_A \cap \mathcal{M}_B$, if one exists. Let $w_1$ denote the unique worker whose preference list differs between $A$ and $B$.

\begin{lemma}
\label{lem:nomultmatchings}
\Cref{alg:daalgorithm_worker_both_sides} always terminates and returns either a perfect matching or the failure symbol $\boxtimes$.
\end{lemma}

\begin{proof}
% [Proof of \Cref{lem:nomultmatchings}]
    The algorithm always terminates since there is at least one rejection per iteration before the final iteration. 

    Suppose the algorithm terminates with two distinct perfect matchings, $\mu_A$ in room $A$ and $\mu_B$ in room $B$. 
    Then there must exist a worker $w \neq w_1$ with the same preferences in $A$ and $B$ but different matches in $\mu_A$ and $\mu_B$. Such a worker must exist since only $w_1$ changed their preferences, and the matchings are perfect. Suppose, without loss of generality, that $\mu_A(w) >_w^{A, B} \mu_B(w)$. As $w$ has the same lists in $A$ and $B$, $w$ must have proposed to $\mu_A(w)$ in room $B$, been rejected, and crossed them off their list. However, per the algorithm, $w$ should have also crossed off $\mu_A(w)$ from their list in room $A$, hence $\mu_A$ cannot be a perfect matching achieved in room $A$. 
\end{proof}

\begin{lemma}
\label{lem:n_1_worker_bp}
    If firm $f$ rejects $w$ in either room, $w$ and $f$ cannot be matched in any matching in $\mathcal{M}_A \cap \mathcal{M}_B$.
\end{lemma}

\begin{proof}
% [Proof of \Cref{lem:n_1_worker_bp}]
    We use the property of the DA algorithm for a single instance $A$ which states that if a firm $f$ rejects a worker $w$, then $(w, f)$ are not matched in any stable matching with respect to $A$~\cite{GaleS}. 

    Suppose, without loss of generality, that $f$ rejects $w$ in room $A$. Then $(w, f)$ are not matched in any stable matching in $\mathcal{M}_A$, and hence they cannot be matched in any stable matching in $\mathcal{M}_A \cap \mathcal{M}_B$.
\end{proof}

\begin{lemma}
\label{lem:daalgorithm_worker_both_sides}
    If $A$ and $B$ admit a robust stable matching, then \Cref{alg:daalgorithm_worker_both_sides} finds one.
\end{lemma}

\begin{proof}
% [Proof of \Cref{lem:daalgorithm_worker_both_sides}]
Suppose there exists such a stable matching, but the algorithm terminates by outputting $\boxtimes$. This only happens when a worker is rejected by all firms. However, by \Cref{lem:n_1_worker_bp}, this means that the worker has no feasible partner in any stable matching, which is a contradiction.
\end{proof}

\begin{lemma}
\label{lem:daalgorithm_worker_both_sides2}
    If \Cref{alg:daalgorithm_worker_both_sides} returns a matching, it must be stable under both $A$ and $B$.
\end{lemma}

\begin{proof}
% [Proof of \Cref{lem:daalgorithm_worker_both_sides2}]
Suppose $\mu$ is the perfect matching returned, and without loss of generality, suppose it has a blocking pair $(w, f)$ with respect to $A$. That is, $f >_w^A \mu(w)$ and $w >_f^A \mu(f)$. Then $w$ must have proposed to $f$ during some iteration of the algorithm and been rejected in room $A$ for some worker $w'$ whom $f$ preferred to $w$, eventually ending up with $\mu(f)$. As firms only get better matches as the algorithm progresses, it must be the case that $\mu(f) >_f w$, and hence $(w, f)$ is not a blocking pair.
\end{proof}

\begin{lemma}
\label{lem:daalgorithm_worker_both_sides3}
    If \Cref{alg:daalgorithm_worker_both_sides} returns a matching, it must be the worker-optimal matching.
\end{lemma}

\begin{proof}
% [Proof of \Cref{lem:daalgorithm_worker_both_sides3}]
If not, some worker $w \neq w_1$ has a better partner in the worker-optimal matching. But then, that worker would have proposed to that partner in every room and been rejected in at least one of them, resulting in a contradiction.
\end{proof}

\subsubsection{Proof of Correctness of \texorpdfstring{\Cref{alg:daalgorithm_firm_both_sides}}{Algorithm~3}}
\label{sec:daalgorithm_firm_both_sides}

We now analyze the firm-proposing version of the algorithm for the $(1,n)$ setting, in which exactly one worker changes preferences between instances $A$ and $B$. The algorithm runs the Deferred Acceptance algorithm on a collection of appropriately defined instances and returns the firm-optimal robust stable matching in $\mathcal{M}_A \cap \mathcal{M}_B$, if one exists. Let $w_1$ denote the unique worker whose preference list differs between $A$ and $B$.

\begin{lemma}
\label{lem:daalgorithm_firm_both_sides}
    The \Cref{alg:daalgorithm_firm_both_sides} terminates and always returns a perfect matching or $\boxtimes$.
\end{lemma}

\begin{proof}
% [Proof of \Cref{lem:daalgorithm_firm_both_sides}]
    The algorithm terminates since there is at least one rejection per iteration. 

    Suppose the algorithm terminates with distinct perfect matchings $\mu_A$ in room $A$ and $\mu_B$ in room $B$. Then there is some worker $w \neq w_1$ who has the same lists in $A$ and $B$ but distinct matches $\mu_A(w)$ and $\mu_B(w)$. Suppose, without loss of generality, that $\mu_A(w) >_w^{A, B} \mu_B(w)$. Then $w$ must have received a proposal from $\mu_A(w)$ in some iteration, since to be tentatively matched, a worker must have received a proposal from that firm. Workers reject all firms that are less preferred in all rooms, so $w$ must have rejected $\mu_B(w)$ in room $R_A$ as well as in room $R_B$. Therefore, $\mu_B$ cannot be the final perfect matching in room $R_B$—a contradiction.
\end{proof}

\begin{lemma}
\label{lem:n_1_firm_bp}
    If firm $f$ is rejected by worker $w$ in either room, then $w$ and $f$ cannot be matched in any strongly stable matching in $\mathcal{M}_A \cap \mathcal{M}_B$.
\end{lemma}

\begin{proof}
% [Proof of \Cref{lem:n_1_firm_bp}]
We use the property of the DA algorithm for a single instance $A$, which states that if a firm $f$ is rejected by worker $w$, then $(w, f)$ are not matched in any stable matching with respect to $A$~\cite{GaleS}. 

We also note the additional fact that once a worker rejects a firm, they are only matched to partners they strictly prefer over that firm from that point onward. This is equivalent to preemptively rejecting all firms ranked below $f$-an idea we carry over to this algorithm.

Suppose, without loss of generality, that $w$ rejects $f$ in room $A$. Then $(w, f)$ are not matched in any stable matching in $\mathcal{M}_A$, and hence they cannot be matched in any stable matching in $\mathcal{M}_A \cap \mathcal{M}_B$.
\end{proof}
\begin{lemma}
\label{lem:daalgorithm_firm_both_sides2}
If $A$ and $B$ admit a robust stable matching, then \Cref{alg:daalgorithm_firm_both_sides} finds one.
\end{lemma}

\begin{proof}
% [Proof of \Cref{lem:daalgorithm_firm_both_sides2}]
Suppose there exists such a stable matching, but the algorithm terminates by outputting $\boxtimes$. This only happens when a firm is rejected by all workers. However, by \Cref{lem:n_1_firm_bp}, this means that the firm has no feasible partner in any of the stable matchings, which is a contradiction.
\end{proof}

\begin{lemma}
\label{lem:daalgorithm_firm_both_sides3}
If \Cref{alg:daalgorithm_firm_both_sides} returns a matching, it must be stable under both $A$ and $B$.
\end{lemma}

\begin{proof}
% [Proof of \Cref{lem:daalgorithm_firm_both_sides3}]
Suppose the algorithm returns a perfect matching $\mu$, and assume, without loss of generality, that it has a blocking pair $(w, f)$ with respect to $A$. That is, $w >_f^A \mu(f)$ and $f >_w^A \mu(w)$. Then $f$ must have proposed to $w$ during some iteration of the algorithm in room $R_A$ and been rejected for some firm $f'$ that $w$ preferred (under $A$). Eventually, $w$ ends up with $\mu(w)$, which it prefers the most. Hence, $\mu(w) >_w^A f$, contradicting the assumption that $(w, f)$ is a blocking pair.
\end{proof}

\begin{lemma}
\label{lem:daalgorithm_firm_both_sides4}
If \Cref{alg:daalgorithm_firm_both_sides} returns a matching, it must be the firm-optimal matching.
\end{lemma}

\begin{proof}
% [Proof of \Cref{lem:daalgorithm_firm_both_sides4}]
If not, some firm $f$ has a better partner in the firm-optimal matching. But then, that firm would have proposed to that partner in all rooms and been rejected in at least one room, leading to a contradiction.
\end{proof}

These together prove that \Cref{alg:daalgorithm_worker_both_sides,alg:daalgorithm_firm_both_sides} find the worker- and firm-optimal robust stable matchings  respectively, completing the proof of \Cref{thm:n_1_algo_works}.

\subsection{Integrality of the Polytope}
\label{sec:polytope}

In \Cref{sec:lp_intro}, we described a linear programming formulation whose feasible region is the stable matching polytope and whose vertices are integral. We now extend this framework to robust stable matchings by defining an appropriate linear program whose feasible region captures matchings that are stable under multiple instances. Our goal is to characterize precisely when the resulting robust stable matching polytope is integral, as integrality enables efficient LP-based algorithms for computing robust stable matchings.

We show that the robust stable matching polytope is integral when $(A, B)$ is of type either $(n,1)$ or $(1,n)$, which in turn yields efficient algorithms for computing robust matchings in these settings.

The linear program for the robust stable matching instance $(A, B)$ is given below. The first two and the last constraints ensure that any feasible solution $x$ is a fractional perfect matching, while the middle two constraints prevent blocking pairs under instances $A$ and $B$, respectively.

\begin{equation}
\label{eq:lp_multi}
    \begin{aligned}[b]
        \text{maximize} \quad & 0 \\
        \text{subject to} \quad
        & \sum_w x_{wf} = 1 && \forall f \in F, \\
        & \sum_f x_{wf} = 1 && \forall w \in W, \\
        & \sum_{f >^A_w f'} x_{wf'} - \sum_{w' >^A_f w} x_{w'f} \leq 0 && \forall w \in W,\ \forall f \in F, \\
        & \sum_{f >^B_w f'} x_{wf'} - \sum_{w' >^B_f w} x_{w'f} \leq 0 && \forall w \in W,\ \forall f \in F, \\
        & x_{wf} \geq 0 && \forall w \in W,\ \forall f \in F.
    \end{aligned}
\end{equation}

A solution to LP~\eqref{eq:lp_multi} corresponds to a fractional matching that is stable under both $A$ and $B$. Analogous to the single-instance case (see \Cref{sec:lp_intro}), let $\theta \in [0,1]$ be chosen uniformly at random. Construct two integral perfect matchings $\mu^A_{\theta}$ and $\mu^B_{\theta}$ by performing the interval-based rounding procedure using the preference orders from instances $A$ and $B$, respectively.

\begin{restatable}{lemma}{lp}
% ${}^\dagger$
\label{lem:lp}
The matchings $\mu^A_{\theta}$ and $\mu^B_{\theta}$ are identical for every $\theta \in [0,1]$, i.e., $\mu^A_{\theta} = \mu^B_{\theta}$.
\end{restatable}

\begin{proof} 
% [Proof of \Cref{lem:lp}]
    $\mu^A_{\theta}$ and $\mu^B_{\theta}$ are perfect matchings, and we know that the preferences of $n-1$ workers $w_1, \dots, w_{n-1}$ are the same in $A$ and $B$. This means that these $n-1$ workers are paired to $n-1$ distinct firms in both matchings. Since they are perfect matchings, the $n^{\text{th}}$ worker must be paired with the remaining firm. Hence, the matchings are identical. 
\end{proof}

\begin{restatable}{theorem}{nintegralpolytope}
% ${}^\dagger$
\label{thm:n_1_integral_polytope} 
If $(A,B)$ is of type $(1,n)$ then $\mu^A_{\theta}$ is stable under both $A$ and $B$ and robust fractional stable matching polytope has integer optimal vertices.
\end{restatable}

\begin{proof} 
% [Proof of \Cref{thm:n_1_integral_polytope}]
    From \Cref{lem:lp}, $\mu^A_{\theta}$ and $\mu^B_{\theta}$ are the same. Therefore, since it is integral and a solution to LP~\eqref{eq:lp_multi}, it is stable under both $A$ and $B$. This means that the set of all fractional robust stable matchings can be written as a convex combination of integral solutions, similar to \Cref{thm:lp_polytope}. This proves that the polytope is integral.
\end{proof}

Hence, in the $(n,1)$ and $(1,n)$ settings, LP~\eqref{eq:lp_multi} provides an efficient method to compute a robust stable matching. However, this technique fails for $(2,2)$ instances—the proof of \Cref{lem:lp} does not extend, and, in fact, the polytope may not even be integral, see example in \Cref{ex:2_2_not_integral}.

\begin{restatable}{theorem}{notintegralpolytope}
% ${}^\dagger$
\label{thm:2_2_not_integral_polytope}
If $(A,B)$ is of type $(2,2)$ then robust stable matching polytope is not always integral.
\end{restatable}

\begin{proof}
Consider the example provided in \Cref{ex:2_2_not_integral}, where instances $A$ and $B$ are defined on 4 workers $\{1,2,3,4\}$ and 4 firms $\{a,b,c,d\}$. Only workers 1 and 2 and firms $a$ and $b$ change their preferences from $A$ to $B$.

Instance $A$ has exactly two stable matchings:
\[
M_A^1 = \{(1,a), (2,b), (3,c), (4,d)\}, \quad
M_A^2 = \{(1,b), (2,a), (3,d), (4,c)\}.
\]
Instance $B$ also has exactly two stable matchings:
\[
M_B^1 = \{(1,b), (2,a), (3,c), (4,d)\}, \quad
M_B^2 = \{(1,a), (2,b), (3,d), (4,c)\}.
\]
And so, there are no robust stable matchings for this pair of $A$ and $B$. 

Now, consider the fractional matching $M = \frac{1}{2}(M_A^1 + M_A^2) = \frac{1}{2}(M_B^1 + M_B^2)$. Since $M$ is a convex combination of stable matchings under both $A$ and $B$, it is a feasible solution to LP~\eqref{eq:lp_multi}. Furthermore, $M$ is the unique solution to the LP in this instance.

Therefore, the robust stable matching polytope is a singleton fractional point, showing that the polytope is not integral in general for $(2,2)$ instances.
\end{proof}

\begin{figure}[ht]
	\begin{wbox}
	\centering
	\begin{minipage}{.45\linewidth}
        \centering
		\begin{tabular}{p{0.5cm}|p{0.5cm}p{0.5cm}p{0.5cm}p{0.5cm}p{0.5cm}}
			1  & a & c & b & d \\
			2  & b & a & c & d \\
			3  & c & d & a & b \\
			4  & d & c & a & b \\
		\end{tabular}
		
		\hspace{1cm}
		
		Worker preferences in $A$ 
		
		\hspace{1cm}
		
	\end{minipage}%
	\begin{minipage}{.45\linewidth}
        \centering
	\begin{tabular}{p{0.5cm}|p{0.5cm}p{0.5cm}p{0.5cm}p{0.5cm}p{0.5cm}}
		a  & 2 & 1 & 3 & 4 \\
		b  & 1 & 2 & 3 & 4 \\
		c  & 4 & 1 & 3 & 2 \\
		d  & 3 & 1 & 4 & 2 \\
	\end{tabular}
	
	\hspace{1cm}
	
	Firm preferences in $A$
	 
	\hspace{1cm}
        \end{minipage}%

        \begin{minipage}{.45\linewidth}
        \centering
		\begin{tabular}{p{0.5cm}|p{0.5cm}p{0.5cm}p{0.5cm}p{0.5cm}p{0.5cm}}
			\cellcolor{red!25}1  & \cellcolor{red!25}b & \cellcolor{red!25}d & \cellcolor{red!25}a & \cellcolor{red!25}c \\
			\cellcolor{red!25}2  & \cellcolor{red!25}a & \cellcolor{red!25}b & \cellcolor{red!25}c & \cellcolor{red!25}d \\
			3  & c & d & a & b \\
			4  & d & c & a & b \\
		\end{tabular}
		
		\hspace{1cm}
		
		Worker preferences in $B$ 
		
		\hspace{1cm}
		
	\end{minipage}%
	\begin{minipage}{.45\linewidth}
        \centering
	\begin{tabular}{p{0.5cm}|p{0.5cm}p{0.5cm}p{0.5cm}p{0.5cm}p{0.5cm}}
		\cellcolor{red!25}a  & \cellcolor{red!25}1 & \cellcolor{red!25}2 & \cellcolor{red!25}3 & \cellcolor{red!25}4 \\
		\cellcolor{red!25}b  & \cellcolor{red!25}2 & \cellcolor{red!25}1 & \cellcolor{red!25}3 & \cellcolor{red!25}4 \\
		c  & 4 & 1 & 3 & 2 \\
		d  & 3 & 1 & 4 & 2 \\
	\end{tabular}
	
	\hspace{1cm}
	
	Firm preferences in $B$
	 
	\hspace{1cm}
        \end{minipage}%

	\end{wbox}
	\caption{The robust stable matching polytope for this $(2,2)$ instance is not integral.}
	\label{ex:2_2_not_integral} 
\end{figure}

% \subsubsection{Extension to Multiple Instances}

All the results from this section extend to more than two instances. Let $(A_1, A_2, \ldots, A_k)$ be a robust stable matching instance that is $(1,n)$, i.e., $n-1$ of the $n$ workers have identical preference profiles across all instances. Then:

\begin{restatable}{theorem}{nmultiple}
% ${}^\dagger$
\label{thm:n_1_multiple}
For a $(1,n)$ robust stable matching instance $(A_1, A_2, \ldots, A_k)$ with $k \geq 2$, with the same worker changing preferences across instances, the set $\mathcal{L}' = \mathcal{M}_{A_1} \cap \mathcal{M}_{A_2} \cap \ldots \cap \mathcal{M}_{A_k}$ forms a sublattice of each of the lattices $\mathcal{L}_{A_i}$. The LP formulation can be used to find the matchings in $\mathcal{L}'$, and the worker-optimal and firm-optimal robust stable matchings can be computed in polynomial time. Furthermore, if the instances are also $(0,n)$, then Birkhoff's partial order generating $\mathcal{L}'$ can be computed efficiently, and its matchings can be enumerated with polynomial delay.
\end{restatable}

\begin{proof}
% [Proof of \Cref{thm:n_1_multiple}]
    $\mathcal{L}'$ is a sublattice of each $\mathcal{L}_{A_i}$ as a corollary of \Cref{thm:sublattice_two_side}. The LP approach can be used to find stable matchings in the intersection: The new formulation will have $k$ blocking pair conditions—one for each instance $A_i$—to ensure that no worker–firm pair $(w,f)$ is a blocking pair for $x$ under that instance. \Cref{alg:daalgorithm_worker_both_sides,alg:daalgorithm_firm_both_sides} can be used to find the worker- and firm-optimal matchings. Define a room for each instance and run the Deferred Acceptance algorithm in each room according to the preferences of the corresponding instance. If the same perfect matching $M$ is obtained in every room, it is output by the algorithm and is stable with respect to all instances.

    When the instances are $(0,n)$, Birkhoff's partial order for $\mathcal{L}'$ can be found by identifying the edge set representing $\mathcal{M}_{A_1} \cap \mathcal{M}_{A_i}$ for each instance $A_i$, and by \Cref{lem:poset}, the union of these edge sets defines $\mathcal{L}'$. This provides a way to efficiently enumerate the matchings.
\end{proof}

\subsection{A Robust Stable Matching Algorithm for the General Case}
\label{sec:robust_general_algo}

\Cref{sec:structure,sec:optimal_matchings,sec:polytope} establish that finding robust stable matchings for $(1,n)$ instances can be done in polynomial time, whereas \cite{NP-Two-stable} shows that the problem becomes NP-hard when $(A,B)$ is of type $(n,n)$. We now consider the fully general setting in which instances $A$ and $B$ differ arbitrarily on the preferences of a subset of agents $S \subseteq \mathcal{W} \cup \mathcal{F}$, where $|S \cap \mathcal{W}| = p$ and $|S \cap \mathcal{F}| = q$.

\Cref{alg:robust_xp} finds a robust stable matching whenever one exists. By \Cref{thm:robust_xp}, the algorithm runs in time $O(n^{p+q+2})$ and can be modified to enumerate all robust stable matchings with the same asymptotic delay. This yields an XP-time algorithm parameterized by the total number of agents whose preferences differ between the two instances.

Intuitively all possible partner assignments for agents in $S$ are enumerated, there are at most $O(n^{p+q})$ possibilities. Each one defines a partial matching $M$ for the agents in $S$, and the algorithm attempts to determine whether $M$ can be extended to a full matching that is stable under $A$ and $B$. First it verifies that $M$ does not contain any blocking pairs within $T = S \cup M(S)$. If no blocking pair exists, define truncated instance $X$ on the remaining agents $U = (W \cup F) \setminus T$. In $X$, the preference lists of agents in $U$ are shortened to eliminate any potential blocking pairs involving agents in $T$. A stable matching $M'$ is computed on $X$. If $M'$ is a perfect matching on $U$, it can be shown that the combined matching $M^\star = M \cup M'$ is a stable matching on the full set of agents in both $A$ and $B$.

\begin{algorithm}[ht]
  \begin{wbox}
    \textsc{RobustStableMatching}$(A,B)$:\\
    \textbf{Input:} Stable matching instances $A$ and $B$ on agents $W \cup F$.\\
    \textbf{Output:} Perfect matching $M^\star$ or $\boxtimes$ (when no robust stable matching exists).\\
    $S \subseteq W \cup F$ are agents whose preferences differ in $A$ and $B$.
    \begin{enumerate}
      \item \textbf{For} each assignment of partners to agents in $S$ that defines a valid partial matching:
      \begin{enumerate}
        \item Let $M$ be the resulting partial matching and $T \coloneqq S \cup M(S)$.
        \item \textbf{If} there exists $w \in T \cap W$ and $f \in T \cap F$ such that $(w,f)$ is a blocking pair with respect to $M$ in either $A$ or $B$, \textbf{continue} to the next iteration.
        \item Remove all agents in $T$ to obtain the instance $X$ on $U \coloneqq (W \cup F)\setminus T$;
        \item For each $a \in T$ and $b \in U$, if $b \succ_a M(a)$ in either $A$ or $B$, then truncate $b$'s preference list at $a$ in $X$ (i.e., remove $a$ and all agents ranked below $a$).
        \item Find stable matching $M'$ in $X$.
        \item {If} $M'$ is a perfect matching on $U$, \textbf{return} $M^\star \coloneqq M \cup M'$.
      \end{enumerate}
      \item \textbf{Return} $\boxtimes$.
    \end{enumerate}
  \end{wbox}
  \caption{An XP-time algorithm to find a robust stable matching for small $(p,q)$.}
  \label{alg:robust_xp}
\end{algorithm}

\begin{restatable}{theorem}{robustalgoworks}
% ${}^\dagger$
\label{thm:robust_xp}
    For instances $A$ and $B$ of type $(1,n)$. \Cref{alg:robust_xp} finds a robust stable matching in $O\!\left(n^{p+q+2}\right)$ time, or correctly reports that no such matching exists. And the set of robust stable matchings can be enumerated with $O\!\left(n^{p+q+2}\right)$ delay.
 \end{restatable}

\begin{proof}
    
% \subsection{Proof of Correctness of \Cref{alg:robust_xp}}
% \label{sec:robust_xp_proofs}

Let $A$ and $B$ be two instances on workers $W$ and firms $F$. Let $S$ be the set of agents whose preferences differ between $A$ and $B$. In a fixed iteration of \Cref{alg:robust_xp}, let $M$ denote the partial matching on $S$ chosen in Step~1(a), let $T \coloneqq S \cup M(S)$, and let $U \coloneqq (W \cup F)\setminus T$. The truncated instance $X$ on $U$ is obtained in Steps~1(c)--(d). Since all agents with differing preferences are removed, the preferences of agents in $U$ are identical in both $A$ and $B$.

\begin{lemma}
\label{lem:robust_completeness}
If there exists a robust stable matching $M^\star \in M_A \cap M_B$, then \Cref{alg:robust_xp} returns a perfect matching.
\end{lemma}

\begin{proof}
% [Proof of \Cref{lem:robust_completeness}]
Consider the iteration that selects the partial matching $M \coloneqq M^\star|_S$ in Step~1(a). Let $M' \coloneqq M^\star|_U$. 

First, note that there are no blocking pairs $(w, f)$ with $w, f \in T$ for $M$ under either instance $A$ or $B$; otherwise, $M^\star$ would not be stable in the corresponding instance. Therefore, the algorithm does not abort at Step~1(b) during this iteration.

Next, observe that $M'$ is also stable in $X$. Otherwise, any blocking pair for $M'$ in $X$ would correspond to a blocking pair in both $A$ and $B$, contradicting the stability of $M^\star$. Furthermore, since $M^\star$ is a perfect matching on $W \cup F$, $M'$ must be a perfect matching on $U$.

Since $M'$ is a perfect stable matching in $X$, the Rural Hospital Theorem implies that every stable matching in $X$ is perfect. Hence, the algorithm must output a perfect matching in this iteration.

\end{proof}

\begin{lemma}
\label{lem:robust_soundness}
If \Cref{alg:robust_xp} returns a perfect matching $M^\star$, then $M^\star \in \mathcal{M}_A \cap \mathcal{M}_B$.
\end{lemma}

\begin{proof}
% [Proof of \Cref{lem:robust_soundness}]
Let $M \coloneqq M^\star|_S$ be the partial matching on $S$, and let $M' \coloneqq M^\star|_U$ be the perfect stable matching in $X$ obtained in the iteration in which the algorithm returned $M^\star$. Fix any instance $I \in \{A, B\}$. Any potential blocking pair $(a, b)$ under $I$ must fall into one of the following cases:
\begin{enumerate}
    \item \emph{Both agents in $T$.} This is ruled out by Step~1(b), which discards the iteration if any such pair blocks $M$ in either instance.
    
    \item \emph{One agent in $T$ and the other in $U$.} Without loss of generality, let $a \in T$ and $b \in U$. If $(a, b)$ is a blocking pair under $I$, then $b \succ_a^I M(a)$ and $a \succ_b^I M'(b)$. However, since $b \succ_a^I M(a)$, agent $a$ and all agents ranked worse than $a$ are removed from $b$'s preference list in instance $X$ during Step~1(d). Thus, $b$ cannot be matched to an agent ranked below $a$, contradicting $a \succ_b^I M'(b)$.
    
    \item \emph{Both agents in $U$.} If $(a, b)$ is a blocking pair, then $b \succ_a^I M'(a)$ and $a \succ_b^I M'(b)$. Note that we only truncate the lower ends of each agent’s preference list. In particular, if $x \in U$ is matched to $M'(x)$, then every agent that $x$ prefers to $M'(x)$ in instance $I$ must still appear in $x$'s preference list in instance $X$. Therefore, the pair $(a, b)$ must also be blocking under instance $X$, contradicting the fact that Step~1(e) returned $M'$ as stable.
\end{enumerate}
Thus $M^\star$ admits no blocking pair in $I$, and since $I \in \{A, B\}$ was arbitrary, $M^\star \in \mathcal{M}_A \cap \mathcal{M}_B$.
\end{proof}

\begin{lemma}
\label{lem:robust_xp}
\Cref{alg:robust_xp} decides the existence of a robust stable matching in time $O\!\left(n^{p+q+2}\right)$.
\end{lemma}

\begin{proof}
% [Proof of \Cref{lem:robust_xp}]
We bound the work per iteration and the total number of iterations:
\begin{enumerate}
    \item \emph{Number of iterations.} Step~1 iterates over all injective partner assignments to the $|S| = p + q$ agents, which yields $O(n^{p+q})$ possibilities.
    
    \item \emph{Run time per-iteration.} Step~1(b) checks for blocking pairs within $T$; Steps~1(c)--(d) construct the truncated instance $X$; Step~1(f) finds a stable matching in $X$; and Step~1(g) verifies whether $M'$ is a perfect matching on $U$. Each of these steps can be implemented in $O(n^2)$ time using standard data structures from the $O(n^2)$ implementation of the Deferred Acceptance algorithm.
\end{enumerate}

Hence, the overall runtime is $O(n^{p+q+2})$. Correctness follows from \Cref{lem:robust_completeness,lem:robust_soundness}.
\end{proof}

\begin{lemma}
\label{lem:robust_xp_enumeration}
The set of robust stable matchings of instances $A$ and $B$ can be enumerated with $O(n^{p+q+2})$ delay.
\end{lemma}

\begin{proof}
% [Proof of \Cref{lem:robust_xp_enumeration}]
\Cref{alg:robust_xp} can be modified to enumerate all robust stable matchings as follows: instead of returning a single stable matching in Step~1(f), enumerate all stable matchings of the truncated instance $X$ using a known enumeration algorithms of its stable matching lattice structure with $O(n^2)$ delay per matching.

Since the outer loop of Step~1 iterates over $O(n^{p+q})$ choices of partial matchings on $S$, and each such iteration performs $O(n^2)$ work (or terminates early at Step~1(b)), the delay between any two outputs is at most $O(n^{p+q+2})$.

Hence, all robust stable matchings can be enumerated with $O(n^{p+q+2})$ delay.
\end{proof}

Together, the preceding lemmas establish \Cref{thm:robust_xp}.

\end{proof}
% \robustalgoworks*

\begin{restatable}{corollary}{robust_xp}
\label{cor:xp_time}
If a constant number of agents change their preferences, then deciding whether a robust stable matching exists and finding one can be achieved in polynomial time. All such matchings can be enumerated with polynomial delay.
\end{restatable}

\Cref{thm:robust_xp} implies \Cref{cor:xp_time} as the running time is polynomial if a constant number of agents change preferences. The decision and enumeration problems for robust stable matchings can be solved efficiently. This shows that the set of robust stable matchings can be expressed as a union of $O(n^{p+q})$ stable matching lattices, one for each consistent assignment of partners to agents in $S$. 

However the partial orders defining these lattices can differ significantly. It remains unclear whether this decomposition yields an efficient method for computing Birkhoff’s partial order even in the $(1,1)$ case—and by \Cref{cor:birkhoff_one_n}, also in the $(1,n)$ case. These observations naturally lead to the following open problems.

\begin{remark}
\textbf{Open Problem.} The problem of computing Birkhoff’s partial order for the $(1,1)$ case and by extension, for $(1,n)$ is open.
\end{remark}

\begin{remark}
\textbf{Open Problem.} Given a $(p,q)$ robust stable matching instance $(A,B)$ with $2 \leq p,q \leq n$, $p+q < 2n$, and $p+q = \omega(1)$, it remains open whether determining the existence of a robust stable matching—and computing one if it exists—is in $\mathcal{P}$.
\end{remark}

\section{Acknowledgements}
\label{sec:ack}

This work was supported in part by the NSF under grant CCF-2230414.
We would like to thank Simon Murray for his insights, which led to~\Cref{alg:robust_xp} and all the anonymous reviewers for their valuable comments.

\bibliographystyle{alpha}
\bibliography{refs}

\newcommand{\etalchar}[1]{$^{#1}$}
\begin{thebibliography}{GMRV26}

\bibitem[ABF{\etalchar{+}}17]{aziz2}
H.~Aziz, P.~Biro, T.~Fleiner, S.~Gaspers, R.~de~Haan, N.~Mattei, and B.~Rastegari.
\newblock Stable matching with uncertain pairwise preferences.
\newblock In {\em International Joint Conference on Autonomous Agents and Multiagent Systems}, pages 344--352, 2017.

\bibitem[ABG{\etalchar{+}}16]{aziz1}
H.~Aziz, P.~Biro, S.~Gaspers, R.~de~Haan, N.~Mattei, and B.~Rastegari.
\newblock Stable matching with uncertain linear preferences.
\newblock In {\em International Symposium on Algorithmic Game Theory}, pages 195--206, 2016.

\bibitem[APR09]{NYC-school}
Atila Abdulkadiro{\u{g}}lu, Parag~A Pathak, and Alvin~E Roth.
\newblock Strategy-proofness versus efficiency in matching with indifferences: Redesigning the {NYC} high school match.
\newblock {\em American Economic Review}, 99(5):1954--78, 2009.

\bibitem[APRS05]{Boston}
Atila Abdulkadiro{\u{g}}lu, Parag~A Pathak, Alvin~E Roth, and Tayfun S{\"o}nmez.
\newblock The boston public school match.
\newblock {\em American Economic Review}, 95(2):368--371, 2005.

\bibitem[B{\etalchar{+}}37]{Birkhoff}
Garrett Birkhoff et~al.
\newblock Rings of sets.
\newblock {\em Duke Mathematical Journal}, 3(3):443--454, 1937.

\bibitem[BBHN21]{BBHN21a}
Niclas Boehmer, Robert Bredereck, Klaus Heeger, and Rolf Niedermeier.
\newblock Bribery and control in stable marriage.
\newblock {\em Journal of Artificial Intelligence Research}, 71:993--1048, 2021.

\bibitem[BCK{\etalchar{+}}20]{BredereckCKLN20}
Robert Bredereck, Jiehua Chen, Dusan Knop, Junjie Luo, and Rolf Niedermeier.
\newblock Adapting stable matchings to evolving preferences.
\newblock In {\em The Thirty-Fourth {AAAI} Conference on Artificial Intelligence, {AAAI} 2020, The Thirty-Second Innovative Applications of Artificial Intelligence Conference, {IAAI} 2020, The Tenth {AAAI} Symposium on Educational Advances in Artificial Intelligence, {EAAI} 2020, New York, NY, USA, February 7-12, 2020}, pages 1830--1837. {AAAI} Press, 2020.

\bibitem[BCK24]{BERCZI}
Kristóf Bérczi, Gergely Csáji, and Tamás Király.
\newblock Manipulating the outcome of stable marriage and roommates problems.
\newblock {\em Games and Economic Behavior}, 147:407--428, 2024.

\bibitem[BFK{\etalchar{+}}21]{BFK+21a}
Robert Bredereck, Piotr Faliszewski, Andrzej Kaczmarczyk, Rolf Niedermeier, Piotr Skowron, and Nimrod Talmon.
\newblock Robustness among multiwinner voting rules.
\newblock {\em Artificial Intelligence}, 290:103403, 2021.

\bibitem[BGS24]{robust_pop_matchings_BGS}
Martin Bullinger, Rohith~Reddy Gangam, and Parnian Shahkar.
\newblock Robust popular matchings.
\newblock In {\em Proceedings of the 23rd International Conference on Autonomous Agents and Multiagent Systems}, AAMAS '24, page 225–233, Richland, SC, 2024. International Foundation for Autonomous Agents and Multiagent Systems.

\bibitem[BHN22a]{boehmer_incremental_mfcs}
Niclas Boehmer, Klaus Heeger, and Rolf Niedermeier.
\newblock {Deepening the (Parameterized) Complexity Analysis of Incremental Stable Matching Problems}.
\newblock In Stefan Szeider, Robert Ganian, and Alexandra Silva, editors, {\em 47th International Symposium on Mathematical Foundations of Computer Science (MFCS 2022)}, volume 241 of {\em Leibniz International Proceedings in Informatics (LIPIcs)}, pages 21:1--21:15, Dagstuhl, Germany, 2022. Schloss Dagstuhl -- Leibniz-Zentrum f{\"u}r Informatik.

\bibitem[BHN22b]{boehmer_incremental_aaai}
Niclas Boehmer, Klaus Heeger, and Rolf Niedermeier.
\newblock Theory of and experiments on minimally invasive stability preservation in changing two-sided matching markets.
\newblock {\em Proceedings of the AAAI Conference on Artificial Intelligence}, 36(5):4851--4858, Jun. 2022.

\bibitem[CNS18]{NP-Two-stable-Chen}
Jiehua Chen, Rolf Niedermeier, and Piotr Skowron.
\newblock Stable marriage with multi-modal preferences.
\newblock In {\em Proceedings of the 2018 ACM Conference on Economics and Computation}, EC '18, page 269–286, New York, NY, USA, 2018. Association for Computing Machinery.

\bibitem[Cs{\'a}24]{Csaj24a}
Gergely Cs{\'a}ji.
\newblock Popular and dominant matchings with uncertain and multimodal preferences.
\newblock In Kate Larson, editor, {\em Proceedings of the Thirty-Third International Joint Conference on Artificial Intelligence, {IJCAI-24}}, pages 2740--2747. International Joint Conferences on Artificial Intelligence Organization, 8 2024.
\newblock Main Track.

\bibitem[CSS21]{Chen-Matchings-Under-preferences}
Jiehua Chen, Piotr Skowron, and Manuel Sorge.
\newblock Matchings under preferences: Strength of stability and tradeoffs.
\newblock {\em ACM Trans. Econ. Comput.}, 9(4), oct 2021.

\bibitem[DF81]{DubinsF}
Lester~E Dubins and David~A Freedman.
\newblock Machiavelli and the {G}ale-{S}hapley algorithm.
\newblock {\em The American Mathematical Monthly}, 88(7):485--494, 1981.

\bibitem[EFS09]{EFS09a}
Edith Elkind, Piotr Faliszewski, and Arkadii Slinko.
\newblock On distance rationalizability of some voting rules.
\newblock In {\em Proceedings of the 12th Conference on Theoretical Aspects of Rationality and Knowledge}, TARK '09, page 108–117, New York, NY, USA, 2009. Association for Computing Machinery.

\bibitem[EIV23]{MM.Book-Online}
Federico Echenique, Nicole Immorlica, and Vijay~V. Vazirani, editors.
\newblock {\em Online and Matching-Based Market Design}.
\newblock Cambridge University Press, 2023.

\bibitem[FIM07]{Fleiner}
Tamás Fleiner, Robert~W. Irving, and David~F. Manlove.
\newblock Efficient algorithms for generalized stable marriage and roommates problems.
\newblock {\em Theoretical Computer Science}, 381(1):162--176, 2007.

\bibitem[FKJ16]{genPreferences}
Linda Farczadi, Georgiou Konstantinos, and Könemann Jochen.
\newblock Stable marriage with general preferences.
\newblock {\em Theory of Computing Systems}, 59(4):683--699, 2016.

\bibitem[FR16]{FaRo15a}
Piotr Faliszewski and J{\"o}rg Rothe.
\newblock Control and bribery in voting.
\newblock In Felix Brandt, Vincent Conitzer, Ulle Endriss, J.~Lang, and Ariel~D. Procaccia, editors, {\em Handbook of Computational Social Choice}, chapter~7. Cambridge University Press, 2016.

\bibitem[G{\"a}r75]{Gard75a}
Peter G{\"a}rdenfors.
\newblock Match making: {A}ssignments based on bilateral preferences.
\newblock {\em Behavioral Science}, 20(3):166--173, 1975.

\bibitem[GI89]{GusfieldI}
Dan Gusfield and Robert~W Irving.
\newblock {\em The stable marriage problem: structure and algorithms}.
\newblock MIT press, 1989.

\bibitem[GMRV22]{GMRV-nearby-instances}
Rohith~Reddy Gangam, Tung Mai, Nitya Raju, and Vijay~V. Vazirani.
\newblock {A Structural and Algorithmic Study of Stable Matching Lattices of "Nearby" Instances, with Applications}.
\newblock In {\em 42nd IARCS Annual Conference on Foundations of Software Technology and Theoretical Computer Science (FSTTCS 2022)}, 2022.

\bibitem[GMRV26]{GMRV-general-instance}
Rohith~Reddy Gangam, Tung Mai, Nitya Raju, and Vijay~V. Vazirani.
\newblock {Stable Matching: Dealing with Changes in Preferences}.
\newblock In {\em Proceedings of the 25th International Conference on Autonomous Agents and Multi-Agent Systems (AAMAS 2026)}, 2026.
\newblock to appear.

\bibitem[GS62]{GaleS}
David Gale and Lloyd~S Shapley.
\newblock College admissions and the stability of marriage.
\newblock {\em The American Mathematical Monthly}, 69(1):9--15, 1962.

\bibitem[GSOS17]{genc2}
Begum Genc, Mohamed Siala, Barry O'Sullivan, and Gilles Simonin.
\newblock Finding robust solutions to stable marriage.
\newblock {\em arXiv preprint arXiv:1705.09218}, 2017.

\bibitem[GSSO17]{genc1}
Begum Genc, Mohamed Siala, Gilles Simonin, and Barry O’Sullivan.
\newblock On the complexity of robust stable marriage.
\newblock In {\em International Conference on Combinatorial Optimization and Applications}, pages 441--448. Springer, 2017.

\bibitem[IL86]{irving2}
Robert~W Irving and Paul Leather.
\newblock The complexity of counting stable marriages.
\newblock {\em SIAM Journal on Computing}, 15(3):655--667, 1986.

\bibitem[Irv85]{irving}
Robert~W Irving.
\newblock An efficient algorithm for the “stable roommates” problem.
\newblock {\em Journal of Algorithms}, 6(4):577--595, 1985.

\bibitem[Irv94]{IRVING-Indifferences}
Robert~W. Irving.
\newblock Stable marriage and indifference.
\newblock {\em Discrete Applied Mathematics}, 48(3):261--272, 1994.

\bibitem[Knu76]{knuth1976marriages}
Donald~Ervin Knuth.
\newblock Marriages stables.
\newblock {\em Technical report}, 1976.

\bibitem[Knu97a]{Knuth-book}
Donald~Ervin Knuth.
\newblock {\em Stable marriage and its relation to other combinatorial problems: An introduction to the mathematical analysis of algorithms}.
\newblock American Mathematical Soc., 1997.

\bibitem[Knu97b]{Knuth}
Donald~Ervin Knuth.
\newblock {\em Stable marriage and its relation to other combinatorial problems: An introduction to the mathematical analysis of algorithms}, volume~10.
\newblock American Mathematical Soc., 1997.

\bibitem[KPG16]{Kunysz-ssm}
Adam Kunysz, Katarzyna~E. Paluch, and Pratik Ghosal.
\newblock Characterisation of strongly stable matchings.
\newblock In Robert Krauthgamer, editor, {\em Proceedings of the Twenty-Seventh Annual {ACM-SIAM} Symposium on Discrete Algorithms, {SODA} 2016, Arlington, VA, USA, January 10-12, 2016}, pages 107--119. {SIAM}, 2016.

\bibitem[Man02]{MANLOVE-indifference-lattice}
David~F. Manlove.
\newblock The structure of stable marriage with indifference.
\newblock {\em Discrete Applied Mathematics}, 122(1):167--181, 2002.

\bibitem[Man13]{Manlove-book}
David Manlove.
\newblock {\em Algorithmics of Matching Under Preferences}.
\newblock World Scientific, 2013.

\bibitem[ML18]{Menon_Larson}
Vijay Menon and Kate Larson.
\newblock Robust and approximately stable marriages under partial information.
\newblock In George Christodoulou and Tobias Harks, editors, {\em Web and Internet Economics}, pages 341--355, Cham, 2018. Springer International Publishing.

\bibitem[MO19]{NP-Two-stable}
Shuichi Miyazaki and Kazuya Okamoto.
\newblock Jointly stable matchings.
\newblock {\em J. Comb. Optim.}, 38(2):646–665, 2019.

\bibitem[MV18]{MV.robust}
Tung Mai and Vijay~V. Vazirani.
\newblock Finding stable matchings that are robust to errors in the input.
\newblock In {\em European Symposium on Algorithms}, 2018.

\bibitem[Rot82]{Roth-IC-1982economics}
Alvin~E Roth.
\newblock The economics of matching: Stability and incentives.
\newblock {\em Mathematics of operations research}, 7(4):617--628, 1982.

\bibitem[Rot85]{Roth85}
A.~E. Roth.
\newblock The college admissions problem is not equivalent to the marriage problem.
\newblock {\em Journal of Economic Theory}, 36(2):277--288, 1985.

\bibitem[Rot16]{Roth}
Alvin~E. Roth.
\newblock Al {R}oth's game theory, experimental economics, and market design page, 2016.
\newblock \url{http://stanford.edu/~alroth/alroth.html#MarketDesign}.

\bibitem[RS12]{RS}
Alvin~E. Roth and Lloyd~S. Shapley.
\newblock The 2012 {N}obel {P}rize in {E}conomics, 2012.
\newblock \url{https://www.nobelprize.org/nobel_prizes/economic-sciences/laureates/2012/}.

\bibitem[Spi95]{SPIEKER-indifference-lattice}
Boris Spieker.
\newblock The set of super-stable marriages forms a distributive lattice.
\newblock {\em Discrete Applied Mathematics}, 58(1):79--84, 1995.

\bibitem[Sta96]{Stanley}
Richard Stanley.
\newblock Enumerative combinatorics, vol. 1, wadsworth and brooks/cole, pacific grove, ca, 1986; second printing, 1996.

\bibitem[SYE13]{SYE13a}
Dmitry Shiryaev, Lan Yu, and Edith Elkind.
\newblock On elections with robust winners.
\newblock In {\em Proceedings of the 2013 International Conference on Autonomous Agents and Multi-Agent Systems}, AAMAS '13, page 415–422, Richland, SC, 2013. International Foundation for Autonomous Agents and Multiagent Systems.

\bibitem[TS98]{Teo-Sethuraman-Lp}
Chung-Piaw Teo and Jay Sethuraman.
\newblock The geometry of fractional stable matchings and its applications.
\newblock {\em Mathematics of Operations Research}, 23(4):874--891, 1998.

\bibitem[{Van}89]{VANDEVATE}
John~H. {Vande Vate}.
\newblock Linear programming brings marital bliss.
\newblock {\em Operations Research Letters}, 8(3):147--153, 1989.

\end{thebibliography}

\appendix

\section{Proof of Birkhoff's Theorem using Stable Matching Lattices}
\label{app:gen_proof}

We will prove \Cref{thm:generalization} in the context of stable matching lattices; this is
w.l.o.g. since stable matching lattices are as general as finite distributive lattices.
In this context, the proper elements of partial order $\Pi$ will be rotations, 
%Let $\mathcal{L} = L(P)$ be the corresponding stable matching lattice. Let $P_f$ denote a compression of $P$
%and $\mathcal{L}'$ denote a sublattice of $\mathcal{L}$. 
and meta-elements are called \emph{meta-rotations}.
Let $\mathcal{L} = L(\Pi)$ be the corresponding stable matching lattice. 

Clearly it suffices to show that:
\begin{itemize}
	\item Given a compression $\Pi'$, $L(\Pi')$ is isomorphic to a sublattice of $\mathcal{L}$.
	\item Any sublattice $\mathcal{L}'$ is $L(\Pi')$ for some compression $\Pi'$.
\end{itemize}
These two proofs are given
in Sections \ref{sec.forward} and \ref{sec.backward}, respectively.

\subsection{\texorpdfstring{$L(\Pi')$}{L(Pi')} is isomorphic to a sublattice of \texorpdfstring{$L(\Pi)$}{L(Pi)}}
\label{sec.forward}

Let $I$ be a closed subset of $\Pi'$; clearly $I$ is a set of meta-rotations. Define $\rot(I)$ to
be the union of all meta-rotations in $I$, i.e.,
\[ \rot(I) = \{ \rho \in A: A \text{ is a meta-rotation in } I \} .\]

We will define the process of {\em elimination of a meta-rotation} $A$ of $\Pi'$ to be the
elimination of the rotations in $A$ in an order consistent with partial order $\Pi$. Furthermore,
{\em elimination of meta-rotations in $I$} will mean starting from stable matching $M_0$ in
lattice $\mathcal{L}$ and eliminating all meta-rotations in $I$
in an order consistent with $\Pi'$. Observe that this is equivalent to starting from stable matching 
$M_0$ in $\mathcal{L}$ and eliminating all rotations
in $\rot(I)$ in an order consistent with partial order $\Pi$. This follows from
Definition~\ref{def:compression}, since
if there exist rotations $x,y$ in $\Pi$ such that $x$ is in meta-rotation $X$, $y$ is in meta-rotation $Y$ and $x$ precedes $y$, then $X$ must also precede $Y$. 
Hence, if the elimination of all rotations in $\rot(I)$ gives matching $M_I$, then elimination
of all meta-rotations in $I$ will also give the same matching.

Finally, to prove the statement in the title of this section, it suffices to observe that if
$I$ and $J$ are two proper closed sets of the partial order $\Pi'$ then
\[ \rot(I \cup J) = \rot(I) \cup \rot(J) \ \ \ \ \mbox{and} \ \ \ \ 
\rot(I \cap J) = \rot(I) \cap \rot(J) . \]
It follows that the set of matchings obtained by elimination of meta-rotations in a proper
closed set of $\Pi'$ are closed under the operations of meet and join and hence form a 
sublattice of $\mathcal{L}$.

\subsection{Sublattice \texorpdfstring{$\mathcal{L}'$}{L'} is generated by a compression \texorpdfstring{$\Pi'$}{Pi'} of \texorpdfstring{$\Pi$}{Pi}}
\label{sec.backward}

We will obtain compression $\Pi'$ of $\Pi$ in stages. First, we show how to partition the set of 
rotations of $\Pi$ to obtain the meta-rotations of $\Pi'$. We then find precedence relations among 
these meta-rotations to obtain $\Pi'$. Finally, we show $L(\Pi') = \mathcal{L}'$.

Notice that $\mathcal{L}$ can be represented by its Hasse diagram $H(\mathcal{L})$. 
Each edge of $H(\mathcal{L})$
{\em contains} exactly one (not necessarily unique) rotation of $\Pi$.
Then, by ~\Cref{lem:seqElimination}, for any two stable matchings $M_1, M_2 \in \mathcal{L}$ such that
$M_1 \prec M_2$, all paths from $M_1$ to $M_2$ in $H(\mathcal{L})$ contain the same set of rotations.

\begin{definition} 
	For $M_1, M_2 \in \mathcal{L}'$, $M_2$ is said to be an {\em $\mathcal{L}'$-direct successor} of $M_1$ iff 
	$M_1 \prec M_2$ and there is no $M \in \mathcal{L}'$ such that $M_1 \prec M \prec M_2$.
	Let $M_1 \prec \ldots \prec M_k$ be a sequence of matchings in $\mathcal{L}'$ such that $M_{i+1}$ 
	is an $\mathcal{L}'$-direct successor of $M_{i}$ for all $1 \leq i \leq k-1$. Then any
	path in $H(\mathcal{L})$ from $M_1$ to $M_k$ containing $M_{i}$, for all $1 \leq i \leq k-1$, is called 
	an {\em $\mathcal{L}'$-path}. 
\end{definition}

Let $M_{0'}$ and $M_{z'}$ denote the worker-optimal and firm-optimal matchings, respectively, 
in $\mathcal{L}'$. For $M_1, M_2 \in \mathcal{L}'$ with $M_1 \prec M_2$, let $S_{M_1, M_2}$ denote the set of 
rotations contained on any $\mathcal{L}'$-path from $M_1$ to $M_2$. Further,
let $S_{M_0,M_{0'}}$ and $S_{M_{z'},M_z}$ denote the set of rotations contained on any path from 
$M_0$ to $M_{0'}$ and $M_{z'}$ to $M_z$, respectively in $H(\mathcal{L})$. Define the following set
whose elements are sets of rotations.
$$ \mathcal{S} =  \{S_{M_i,M_j} \ | \  M_j \text{ is an $\mathcal{L}'$-direct successor of } M_i,
$$
$$\ 
\text{for every pair of matchings} \ M_i, M_j \ \text{in} \ \mathcal{L}' \} \bigcup $$
$$ \ \  \{ S_{M_0,M_{0'}}, \ S_{M_{z'},M_z} \}. $$

\begin{lemma} 
	\label{lem:partition}
	$\mathcal{S}$ is a partition of $\Pi$.
\end{lemma}
\begin{proof}
	First, we show that any rotation must be in an element of $\mathcal{S}$. Consider a path $p$ from $M_0$ to $M_z$ in the $H(\mathcal{L})$ such that $p$ goes from $M_{0'}$ to $M_{z'}$ via an $\mathcal{L}'$-path. Since $p$ is a path from $M_0$ to $M_z$, all rotations of $\Pi$ are contained on
	$p$ by Lemma~\ref{lem:seqElimination}. Hence, they all appear in the sets in $\mathcal{S}$.
	
	Next assume that there are two pairs $(M_1,M_2) \not = (M_3,M_4)$ of $\mathcal{L}'$-direct successors 
	such that $S_{M_1,M_2} \not = S_{M_3,M_4}$ and $X = S_{M_1,M_2} \cap S_{M_3,M_4} \not = \emptyset$. The set of rotations eliminated from $M_0$ to $M_2$ is  
	\[S_{M_0,M_2} = S_{M_0,M_1} \cup S_{M_1,M_2}.\]
	Similarly, 
	\[S_{M_0,M_4} = S_{M_0,M_3} \cup S_{M_3,M_4}.\]
	Therefore, 
	\[S_{M_0,M_2 \vee M_3} =S_{M_0,M_3} \cup S_{M_1,M_2} \cup S_{M_0,M_1}. \]
	\[S_{M_0,M_1 \vee M_4} = S_{M_0,M_3} \cup S_{M_3,M_4} \cup S_{M_0,M_1}. \]
	Let $M = (M_2 \vee M_3) \wedge (M_1 \vee M_4)$, we have  
	\[S_{M_0, M} = S_{M_0,M_3} \cup S_{M_0,M_1} \cup X.\]
	Hence, 
	\[S_{M_0, M \wedge M_2} = S_{M_0,M_1} \cup X.\]
	Since $X \subset S_{M_1,M_2}$ and $S_{M_1,M_2} \cap S_{M_0,M_1} = \emptyset$, $X \cap S_{M_0,M_1} = \emptyset$. Therefore, 
	\[S_{M_0,M_1} \subset S_{M_0, M \wedge M_2} \subset S_{M_0,M_2},\] 
	and hence $M_2$ is not a $\mathcal{L}'$-direct successor of $M_1$, leading to a contradiction.
\end{proof}

We will denote $S_{M_0,M_{0'}}$ and $S_{M_{z'},M_{z}}$ by $A_s$ and $A_t$, respectively.
The elements of $\mathcal{S}$ will be the meta-rotations of $\Pi'$.  Next, we need to define  
precedence relations among these meta-rotations to complete the construction of $\Pi'$. 
For a meta-rotation $A \in \mathcal{S}$, $A \neq A_t$, define the following subset of $\mathcal{L}'$:
\[ \mathcal{M}^A = \{M \in \mathcal{L}' \ \mbox{such that} \ A \subseteq S_{M_0,M} \} .\]

\begin{lemma}
	\label{lem:latticeContainingMetaRotaion}
	For each meta-rotation $A \in \mathcal{S}$, $A \neq A_t$, $\mathcal{M}^A$ forms a sublattice $\mathcal{L}^A$ of $\mathcal{L}'$.
\end{lemma}

\begin{proof}
	Take two matchings $M_1,M_2$ such that $S_{M_0,M_1}$ and $S_{M_0,M_2}$ are supersets of $A$. Then $S_{M_0,M_1 \wedge M_2} = S_{M_0,M_1} \cap S_{M_0,M_2}$ and $S_{M_0,M_1 \vee M_2} = S_{M_0,M_1} \cup S_{M_0,M_2}$ are also supersets of $A$.
\end{proof}

Let $M^A$ be the worker-optimal matching in the lattice $\mathcal{L}^A$. Let $p$ be any $\mathcal{L}'$-path from 
$M_{0'}$ to $M^A$ and let $\pre(A)$ be the set of meta-rotations appearing before $A$ on $p$. 

\begin{lemma}
	The set $\pre(A)$ does not depend on $p$. Furthermore, on any $\mathcal{L}'$-path from $M_{0'}$ containing 
	$A$, each meta-rotation in $\pre(A)$ appears before $A$.
	\label{lem:preMetaPoset} 
\end{lemma}

\begin{proof}
	Since all paths from $M_{0'}$ to $M^A$ give the same set of rotations, 
	all $\mathcal{L}'$-paths from $M_{0'}$ to $M^A$ give the same set of meta-rotations.
	Moreover, $A$ must appear last in the any $\mathcal{L}'$-path from $M_{0'}$ to $M^A$;
	otherwise, there exists a matching in $\mathcal{L}^A$ preceding $M^A$, giving a contradiction.
	It follows that $\pre(A)$ does not depend on $p$.
	
	Let $q$ be an $\mathcal{L}'$-path from $M_{0'}$ that contains matchings $M', M \in \mathcal{L}'$, where $M$ is 
	an $\mathcal{L}'$-direct successor of $M'$. Let $A$ denote the meta-rotation that is contained on edge 
	$(M', M)$. Suppose there is a meta-rotation $A' \in \pre(A)$ such that $A'$ does not appear before
	$A$ on $q$. Then $S_{M_0,M^A \wedge M} = S_{M_0,M^A} \cap S_{M_0,M}$ contains $A$ but 
	not $A'$. Therefore $M^A \wedge M$ is a matching in $\mathcal{L}^A$ preceding $M^A$, giving is a 
	contradiction.	Hence all matchings in $\pre(A)$ must appear before $A$ on all such paths $q$.
\end{proof}

Finally, add precedence relations from all meta-rotations in $\pre(A)$ to $A$, for each
meta-rotation in $\mathcal{S} - \{A_t\}$. Also, add precedence relations from all meta-rotations 
in $\mathcal{S} - \{A_t\}$ to $A_t$. This completes the construction of $\Pi'$. Below we show
that $\Pi'$ is indeed a compression of $\Pi$, but first we need to establish that this construction 
does yield a valid poset.

\begin{lemma}
	$\Pi'$ satisfies transitivity and anti-symmetry.
\end{lemma}
\begin{proof}
	First we prove that $\Pi'$ satifies transitivity.
	Let $A_1, A_2, A_3$ be meta-rotations such that $A_1 \prec A_2$ and $A_2 \prec A_3$.
	We may assume that $A_3 \not = A_t$.
	Then $A_1 \in \pre(A_2)$ and $A_2 \in \pre(A_3)$. 
	Since $A_1 \in \pre(A_2)$, $S_{M_0,M^{A_2}}$ is a superset of $A_1$.
	By ~\Cref{lem:latticeContainingMetaRotaion}, $M^{A_1} \prec M^{A_2}$. 
	Similarly, $M^{A_2} \prec M^{A_3}$. Therefore $M^{A_1} \prec M^{A_3}$, and hence $A_1 \in \pre(A_3)$.
	
	Next we prove that $\Pi'$ satisfies anti-symmetry. Assume that there exist meta-rotations $A_1, A_2$ such that
	$A_1 \prec A_2$ and $A_2 \prec A_1$. Clearly $A_1, A_2 \not = A_t$. Since $A_1 \prec A_2$, $A_1 \in \pre(A_2)$. Therefore, $S_{M_0,M^{A_2}}$ is a superset of $A_1$. It follows that $M^{A_1} \prec M^{A_2}$. Applying a similar argument we get $M^{A_2} \prec M^{A_1}$. Now, we get a contradiction,
	since $A_1$ and $A_2$ are different meta-rotations.
\end{proof}

\begin{lemma}
	$\Pi'$ is a compression of $\Pi$.
\end{lemma}
\begin{proof}
	Let $x,y$ be rotations in $\Pi$ such that $x \prec y$. 
	Let $X$ be the meta-rotation containing $x$ and $Y$ be the meta-rotation containing $y$.
	It suffices to show that $X \in \pre(Y)$.
	Let $p$ be an $\mathcal{L}'$-path from $M_0$ to $M^Y$.
	Since $x \prec y$, $x$ must appear before $y$ in $p$.
	Hence, $X$ also appears before $Y$ in $p$.
	By ~\Cref{lem:preMetaPoset}, $X \in \pre(Y)$ as desired.
\end{proof}

Finally, the next two lemmas prove that $L(\Pi') = \mathcal{L}'$.

\begin{lemma}
	\label{lem:LPsubsetLbar}
	Any matching in $L(\Pi')$ must be in $\mathcal{L}'$.
\end{lemma}
\begin{proof}
	For any proper closed subset $I$ in $\Pi'$, let $M_I$ be the matching generated by eliminating meta-rotations in $I$.	 
	Let $J$ be another proper closed subset in $\Pi'$ such that $J = I \setminus \{A\}$, 
	where $A$ is a maximal meta-rotation in $I$. 
	Then $M_J$ is a matching in $\mathcal{L}'$ by induction. 
	Since $I$ contains $A$, $S_{M_0,M_I} \supset A$. Therefore, $M^A \prec M_I$. 
	It follows that $M_I = M_J \vee M^A \in \mathcal{L}'$.
\end{proof}

\begin{lemma} Any matching in $\mathcal{L}'$ must be in $L(\Pi')$.
\end{lemma}
\begin{proof}
	Suppose there exists a matching $M$ in $\mathcal{L}'$ such that $M \not \in L(\Pi')$. 
	Then it must be the case that $S_{M_0,M}$ cannot be partitioned into meta-rotations which
	form a closed subset of $\Pi$. Now there are two cases.
	
	First, suppose that $S_{M_0,M}$ can be partitioned into meta-rotations, but they do not form a closed subset of $\Pi'$.
	Let $A$ be a meta-rotation such that $S_{M_0,M} \supset A$, and there exists $B \prec A$ such that $S_{M_0,M} \not \supset B$.
	By Lemma~\ref{lem:latticeContainingMetaRotaion}, $M \succ M^A$ and hence $S_{M_0,M}$ is a superset
	of all meta-rotations in $\pre(A)$, giving is a contradiction. 
	
	Next, suppose that $S_{M_0,M}$ cannot be partitioned into meta-rotations in $\Pi'$. Since the set of meta-rotations partitions $\Pi$, there exists a meta-rotation $X$ such that $Y = X \cap S_{M_0,M}$ is a non-empty subset of $X$.
	Let $J$ be the set of meta-rotations preceding $X$ in $\Pi$. 
	
	$(M_J \vee M) \wedge M^X$ is the matching generated by 
	meta-rotations in $J \cup Y$.
	Obviously, $J$ is a closed subset in $\Pi'$. Therefore, $M_J \in L(\Pi')$.
	By Lemma~\ref{lem:LPsubsetLbar}, $M_J \in \mathcal{L}'$.
	Since $M,M^X \in \mathcal{L}'$, $(M_J \vee M) \wedge M^X \in \mathcal{L}'$ as well. 
	The set of rotations contained on a path from $M_J$ to $(M_J \vee M) \wedge M^X$ in $H(\mathcal{L})$ is exactly $Y$.
	Therefore, $Y$ can not be a subset of any meta-rotation, contradicting the fact that $Y = X \cap S_{M_0,M}$ is a non-empty subset of $X$.
\end{proof}

%\subsection{Proof of Theorem \ref{thm:generalization}}

\section{Algorithms when One Side Changes Preferences}
\label{app:one_side}

In this setting, we consider instances where the workers' preferences remain unchanged while the firms' preferences differ, i.e., instances that are $(0,n)$. The case of $(n,0)$ is equivalent by symmetry. 
% We begin by recalling the following proposition from \cite{GMRV-nearby-instances}.

% \begin{proposition}
% \label{prop.sublattice_one_side} (\cite{GMRV-nearby-instances}, Proposition 6)
% If $A$ and $B$ are $(0,n)$, then the matchings in $\mathcal{M}_A \cap \mathcal{M}_B$ form a sublattice in each of the two lattices.
% \end{proposition}

% \begin{proof}[Proof of \Cref{prop.sublattice_one_side}]
% If $| \mathcal{M}_A \cap \mathcal{M}_B| \leq 1$, then $\mathcal{M}_A \cap \mathcal{M}_B$ is trivially a sublattice of both $\mathcal{L}_A$ and $\mathcal{L}_B$. So assume $| \mathcal{M}_A \cap \mathcal{M}_B| > 1$ and let $M_1$ and $M_2$ be matchings in $\mathcal{M}_A \cap \mathcal{M}_B$. Let $\vee_A$ and $\vee_B$ denote the join operations under $A$ and $B$, respectively, and let $\wedge_A$ and $\wedge_B$ denote the meet operations.

% By the definition of the join operation in \Cref{sec:latticeOfSM}, $M_1 \vee_A M_2$ is the matching where each worker is assigned to their less-preferred partner (or equivalently, each firm to its more-preferred partner) among $M_1$ and $M_2$ under instance $A$. Since every worker has the same preference list in both $A$ and $B$, the less-preferred partner of each worker between $M_1$ and $M_2$ is the same under both instances. Therefore, $M_1 \vee_A M_2 = M_1 \vee_B M_2$. A similar argument holds for the meet, showing $M_1 \wedge_A M_2 = M_1 \wedge_B M_2$. Thus, both $M_1 \vee_A M_2$ and $M_1 \wedge_A M_2$ are in $\mathcal{M}_A \cap \mathcal{M}_B$, proving that $\mathcal{M}_A \cap \mathcal{M}_B$ forms a sublattice in both lattices.
% \end{proof}

The \Cref{thm:sublattice} shows that the join and meet of stable matchings in $\mathcal{M}_A \cap \mathcal{M}_B$ are the same under the ordering of $\mathcal{L}_A$ and $\mathcal{L}_B$. This implies that workers agree on which matchings in the intersection they prefer. The key driver of this result is that workers have the \emph{same preferences} in both instances, allowing many arguments used in single-instance stable matching lattices to carry over. Additionally, the set $\mathcal{M}_A \cap \mathcal{M}_B$ forms the same sublattice in both $\mathcal{L}_A$ and $\mathcal{L}_B$ in the sense that the partial orders coincide; that is, $(\mathcal{M}_A \cap \mathcal{M}_B, \preceq_A) = (\mathcal{M}_A \cap \mathcal{M}_B, \preceq_B)$. This implies the existence of worker-optimal and firm-optimal stable matchings in $\mathcal{M}_A \cap \mathcal{M}_B$, and suggests that a variant of the Deferred Acceptance algorithm could succeed in this setting.

This motivates the notion of a *compound instance* that encodes agent preferences across both input instances. We present a variant of the Deferred Acceptance algorithm (\Cref{alg:daalgorithm_firm_one_side}) on this new instance to find such matchings. See \Cref{fig:compoundinstance} for the construction.

\begin{definition}
\label{def:compound_instance}
Let $A$ and $B$ be $(0,n)$. Define the compound instance $X$ as follows:
\begin{enumerate}
    \item For each worker $w \in \mathcal{W}$, the preference list in $X$ is the same as in $A$ and $B$.
    \item For each firm $f \in \mathcal{F}$, $f$ prefers worker $w_i$ to $w_j$ in $X$ if and only if $f$ prefers $w_i$ to $w_j$ in both $A$ and $B$:
    \[
    \forall f \in \mathcal{F}, \forall w_i, w_j \in \mathcal{W}, \quad w_i >^X_f w_j \iff (w_i >^A_f w_j \text{ and } w_i >^B_f w_j).
    \]
\end{enumerate}
\end{definition}

In the compound instance $X$, workers have totally ordered preference lists, while firms may have partial orders. If a firm $f$ does not strictly prefer $w_i$ to $w_j$ in $X$, we say $f$ is indifferent between them, denoted $w_i \sim^X_f w_j$.

Given a matching $M$, a worker-firm pair $(w, f) \not\in M$ is said to be a strongly blocking pair with respect to $M$ if $w$ prefers $f$ to their assigned partner in $M$ and $f$ either strictly prefers $w$ or is indifferent between $w$ and its partner in $M$. A matching with no strongly blocking pairs is said to be \emph{strongly stable} under $X$. Let $\mathcal{M}_X$ denote the set of matchings that are strongly stable under $X$. See \cite{IRVING-Indifferences,Fleiner} for more information on Strongly stable matchings.

The following lemma shows that compound instances are meaningful in our context:

\begin{lemma} \cite{Fleiner}
\label{lem:strongly_stable}
If $X$ is the compound instance of $A$ and $B$, then a matching $M$ is strongly stable under $X$ if and only if it is stable under both $A$ and $B$. That is, $\mathcal{M}_X = \mathcal{M}_A \cap \mathcal{M}_B$.
\end{lemma}

\begin{proof}[Proof of \Cref{lem:strongly_stable}]
Let $(w, f)$ be a strongly blocking pair for matching $M$ under $X$. Since $w$ has identical preferences in $A$ and $B$, they must strictly prefer $f$ to $M(w)$ in both instances. Firm $f$ must either strictly prefer $w$ to its partner in $M$ or be indifferent between them. In either case, $f$ strictly prefers $w$ in at least one of $A$ or $B$, making $(w, f)$ a blocking pair under that instance. Thus, $M$ is not stable under both.

Conversely, if $(w, f)$ is a blocking pair under $A$ or $B$, then it satisfies the conditions for being strongly blocking under $X$. Therefore, $\mathcal{M}_X = \mathcal{M}_A \cap \mathcal{M}_B$.
\end{proof}
\begin{algorithm}[ht]
	\begin{wbox}
		\textsc{CompoundInstance}($A, B$): \\
		\textbf{Input:} Stable matching instances $A$ and $B$ with the same preferences on the workers' side. \\
		\textbf{Output:} Instance $X$.
		\begin{enumerate}
			\item $\forall w \in W$, $w$'s preferences in $X$ are the same as in $A$ and $B$.
			\item $\forall f \in F$, $\forall w_i, w_j \in W$, $w_i >_f^X w_j$ if and only if $w_i >_f^A w_j$ and $w_i >_f^B w_j$.
			\item Return $X$.
		\end{enumerate}
	\end{wbox}
	\caption{Subroutine for constructing a compound instance.}
	\label{fig:compoundinstance} 
\end{algorithm} 

\subsection{Worker and firm optimal matchings}
\label{subsection.alg_one_sided_errors}

\begin{algorithm}[ht]
	\begin{wbox}
		\textsc{WorkerOptimalForCompoundInstance}($X$): \\
		\textbf{Input:} Compound stable matching instance $X$ \\
		\textbf{Output:} Perfect matching $M$ or $\boxtimes$  \\
        Each worker has a list of all firms, initially uncrossed, in decreasing order of priority.
		\begin{enumerate}
			\item Until all workers receive an acceptance or some worker is rejected by all firms, do:
			\begin{enumerate}
				\item $\forall w \in W$: $w$ proposes to their best uncrossed firm.
				\item $\forall f \in F$: $f$ tentatively accepts their \textbf{\textit{best-ever}} proposal, i.e., the proposal that they strictly prefer to every other proposal ever received, and rejects the rest.
				\item $\forall w \in W$: if $w$ is rejected by a firm $f$, they cross $f$ off their list. 
			\end{enumerate}  
			\item If all workers get an acceptance, output the perfect matching ($M$); else output $\boxtimes$.
		\end{enumerate}
	\end{wbox}
	\caption{\cite{IRVING-Indifferences}Finding the worker-optimal stable matching using compound instance.}
	\label{alg:daalgorithm_worker_one_side} 
\end{algorithm} 

\begin{algorithm}[ht]
	\begin{wbox}
		\textsc{FirmOptimalForCompoundInstance}($X$): \\
		\textbf{Input:} Compound stable matching instance $X$ \\
		\textbf{Output:} Perfect matching $M$ or $\boxtimes$ \\
        Each firm has a DAG over all workers, initially uncrossed, representing its partial order preferences.
		\begin{enumerate}
			\item Until all firms receive an acceptance or some firm is rejected by all workers, do:
			\begin{enumerate}
				\item $\forall f \in F$: $f$ proposes to the \textit{\textbf{best uncrossed workers}}, i.e., all uncrossed workers $w$ such that there is no uncrossed $w'$ with $w' >^X_f w$.
				\item $\forall w \in W$: $w$ tentatively accepts the best proposal and rejects the rest.
				\item $\forall f \in F$: if $f$ is rejected by a worker $w$, they cross $w$ off their DAG.
			\end{enumerate}
			\item If all firms receive an acceptance, output the perfect matching ($M$); else output $\boxtimes$.
		\end{enumerate}
	\end{wbox}
	\caption{\cite{IRVING-Indifferences}Finding the firm-optimal stable matching using compound instance.}
	\label{alg:daalgorithm_firm_one_side} 
\end{algorithm} 

The original Deferred Acceptance Algorithm(\Cref{alg:gs_algorithm}) works in iterations. In each iteration, three steps occur: (a) the proposing side (say, workers) proposes to their best firm that hasn’t yet rejected them; (b) each firm tentatively accepts their best proposal of the round and rejects all others; (c) all workers cross off the firms that rejected them. The loop continues until a perfect matching is found and output. The key idea is that once a rejection occurs, that worker–firm pair cannot be part of any stable matching. 

We modify this algorithm to find the worker- and firm-optimal stable matchings in the intersection for the $(0, n)$ setting, ensuring that strong blocking pairs are prevented.

\begin{theorem}
\label{thm:n_0_algo_works}
\Cref{alg:daalgorithm_worker_one_side} and \Cref{alg:daalgorithm_firm_one_side} find the worker- and firm-optimal matchings respectively, if they exist.
\end{theorem}

The worker-optimal Algorithm in \Cref{alg:daalgorithm_worker_one_side} (simplified algorithm of \cite{IRVING-Indifferences} and equivalent to algorithms in \cite{NP-Two-stable, genPreferences, GMRV-nearby-instances}) differs from the classic DA algorithm as follows: firms accept a proposal only if it is strictly better than every proposal they have ever received (including from previous rounds). This avoids strong blocking pairs involving firm indifference. Note that a firm may reject all proposals in a round if it is indifferent among the proposing workers, which is necessary to maintain strong stability.

In the firm-optimal Algorithm in \Cref{alg:daalgorithm_firm_one_side} (equivalent to Algorithm SUPER in \cite{IRVING-Indifferences}), firms propose using a DAG of partial preferences. Each firm proposes to all uncrossed workers $w$ such that no uncrossed worker $w'$ exists with $w' >^X_f w$. Workers (with total preferences) tentatively accept their best proposal and reject the rest. If a worker rejects a firm, then by strict preferences, they cannot form a blocking pair later. If multiple workers tentatively accept a firm’s proposals, some may later reject the firm, and the algorithm either continues or terminates with no perfect matching.

Note: when $A = B$, then $X$ has total preferences on both sides, and these algorithms reduce to the standard DA algorithm.

\begin{remark} (\cite{Fleiner})
Strongly stable matchings in an instance $Y$ with partial order preferences on one side form a lattice. Since each partial order has finitely many linear extensions, $\mathcal{M}_Y$ can be written as the intersection of finitely many lattices. However, as the number of linear extensions may be exponential in $n$, this does not yield an efficient method for computing the Birkhoff partial order.
\end{remark}

\subsection{Proof of Correctness of Algorithm~\ref{alg:daalgorithm_worker_one_side}}
\label{sec:daalgorithm_worker_one_side}

We have instances $A$ and $B$ which are $(0,n)$, meaning that workers have complete order preferences and firms have partial order preferences in $X$. In \Cref{alg:daalgorithm_worker_one_side}, workers propose, and unlike in \Cref{alg:gs_algorithm}, firms tentatively accept only the worker whom they strictly prefer to every other proposer they have received—in the current round and all previous rounds.

\begin{lemma}
\label{lem:n_0_worker_bp}
If firm $f$ rejects worker $w$ in any round of \Cref{alg:daalgorithm_worker_one_side}, then $w$ and $f$ can never be matched in any matching that is strongly stable under $X$.
\end{lemma}

\begin{proof}[Proof of \Cref{lem:n_0_worker_bp}]
Assume there is a rejection (otherwise, the result is trivial). Let firm $f$ reject worker $w$. We will show that any matching $M$ such that $(w,f) \in M$ has a strong blocking pair. We proceed by induction on the round number in which the rejection occurs.

Assume the rejection happens in the first round. Then $f$ only rejects $w$ if it receives a proposal from another worker $w'$ whom it prefers or is indifferent to. Since $w'$ proposed to $f$ in the first round, $w'$ must prefer $f$ to $M(w')$. Hence, $(w', f)$ is a strong blocking pair to $M$.

Now assume the rejection occurs in round $k > 1$. If $f$ rejects $w$, then there exists a worker $w'$ who has proposed to $f$ in this or a previous round and such that $w' >^X_f w$ or $w' \sim^X_f w$. If $w'$ was rejected by $f$ in an earlier round, then by the induction hypothesis, $w'$ cannot be matched to any firm it proposed to before $f$, including $f$, in any stable matching. Hence, $f >_{w'} M(w')$, and so $(w', f)$ is a strong blocking pair.
\end{proof}

\begin{lemma}
\label{lem:n_0_worker_bp2}
If $X$ admits a strongly stable matching, then \Cref{alg:daalgorithm_worker_one_side} finds one.
\end{lemma}

\begin{proof}[Proof of \Cref{lem:n_0_worker_bp2}]
Suppose there exists a strongly stable matching, but the algorithm terminates by outputting $\boxtimes$. This only happens when a worker is rejected by all firms. However, by \Cref{lem:n_0_worker_bp}, this means that the worker has no feasible partner in any strongly stable matching, which is a contradiction.
\end{proof}

\begin{lemma}
\label{lem:n_0_worker_bp3}
If \Cref{alg:daalgorithm_worker_one_side} returns a matching, it must be strongly stable under $X$.
\end{lemma}

\begin{proof}[Proof of \Cref{lem:n_0_worker_bp3}]
Assume not. Let $(w, f)$ be a blocking pair for the matching returned by the algorithm. Then $w$ must have proposed to $f$ during the execution and been rejected. Hence, $f$ must be matched to someone it strictly prefers at the end of the algorithm, implying that $(w, f)$ is not a strong blocking pair—a contradiction.
\end{proof}

\begin{lemma}
\label{lem:n_0_worker_bp4}
If \Cref{alg:daalgorithm_worker_one_side} returns a matching, it must be the worker-optimal strongly stable matching.
\end{lemma}

\begin{proof}[Proof of \Cref{lem:n_0_worker_bp4}]
If not, then some worker $w$ has a better partner in the worker-optimal matching. But in that case, $w$ must have proposed to that partner and been rejected during the algorithm, leading to a contradiction.
\end{proof}

\subsection{Proof of Correctness of Algorithm~\ref{alg:daalgorithm_firm_one_side}}
\label{sec:daalgorithm_firm_one_side}

We again have instances $A$ and $B$ which are $(0,n)$, meaning that workers have complete order preferences and firms have partial order preferences in $X$. In \Cref{alg:daalgorithm_firm_one_side}, firms propose, and unlike in \Cref{alg:gs_algorithm}, they may propose to multiple workers in the same iteration.

\begin{lemma}
\label{lem:n_0_firm_bp}
If worker $w$ rejects firm $f$ in any round of \Cref{alg:daalgorithm_firm_one_side}, then $w$ and $f$ can never be matched in any matching stable under both instances $A$ and $B$.
\end{lemma}

\begin{proof}[Proof of \Cref{lem:n_0_firm_bp}]
Suppose, for contradiction, that $w$ rejects firm $f$ in some round, yet they are matched in a stable matching $\mu$.

We proceed by induction on the round number. Suppose $w$ rejects $f$ in the first round. Then there is another firm $f_1$ that $w$ strictly prefers to $f$. Since firms propose to their most preferred uncrossed workers, and this is the first round, there is no worker $w'$ such that $f_1$ prefers to $w$; i.e., $w >_{f_1} w'$ or $w \sim_{f_1} w'$. In matching $\mu$, where $w$ is matched to $f$, this implies that $(w, f_1)$ is a strong blocking pair.

Suppose $w$ rejects $f$ in the $k$th round. Then there is a firm $f_k$ whom $w$ strictly prefers to $f$. One of the following must hold:
\begin{enumerate}
    \item $w >_{f_k} \mu(f_k)$ or $w \sim_{f_k} \mu(f_k)$: Then $(w, f_k)$ is a strong blocking pair.
    \item $\mu(f_k) >_{f_k} w$: Then $f_k$ must have proposed to $\mu(f_k)$ in a previous round and been rejected. By the induction hypothesis, this leads to a different strong blocking pair $(w', f')$ for $\mu$, implying that $\mu$ is not stable.
\end{enumerate}
\end{proof}

\begin{lemma}
\label{lem:n_0_firm_bp2}
\Cref{alg:daalgorithm_firm_one_side} terminates.
\end{lemma}

\begin{proof}[Proof of \Cref{lem:n_0_firm_bp2}]
For any round of the algorithm, one of the following holds:
\begin{enumerate}
    \item Fewer than $n$ proposals are made in round $k$: Then some firm $f$ has no uncrossed workers left to propose to in its list (DAG). The algorithm terminates before this round and outputs $\boxtimes$.
    \item Exactly $n$ proposals are made to $n$ distinct workers: Then all workers receive a single proposal, and either a rejection occurs or the algorithm terminates at the end of the round with a perfect matching.
    \item More than $n$ proposals are made: Then at least one worker $w$ receives multiple proposals, rejects at least one, and at least one firm crosses $w$ off its list.
\end{enumerate}
Thus, in every round, there is either a rejection or the algorithm terminates. Since there are at most $O(n^2)$ rejections, the algorithm terminates in $\poly(n)$ rounds.
\end{proof}

\begin{lemma}
\label{lem:n_0_firm_bp3}
If $X$ admits a strongly stable matching, then \Cref{alg:daalgorithm_firm_one_side} finds one.
\end{lemma}

\begin{proof}[Proof of \Cref{lem:n_0_firm_bp3}]
Suppose there exists a strongly stable matching, but the algorithm terminates by outputting $\boxtimes$. This only happens when a firm is rejected by all workers. However, by \Cref{lem:n_0_firm_bp}, this implies that the firm has no feasible partner in any strongly stable matching, which is a contradiction.
\end{proof}

\begin{lemma}
\label{lem:n_0_firm_bp4}
If \Cref{alg:daalgorithm_firm_one_side} returns a matching, it must be strongly stable under $X$.
\end{lemma}

\begin{proof}[Proof of \Cref{lem:n_0_firm_bp4}]
Assume not. Let $(w, f)$ be a blocking pair to the returned matching. Then $f$ must have proposed to $w$ during the algorithm and been rejected. Thus, $w$ must be matched to a more preferred firm at the end of the algorithm, implying $(w, f)$ is not a strong blocking pair—a contradiction.
\end{proof}

\begin{lemma}
\label{lem:n_0_firm_bp5}
If \Cref{alg:daalgorithm_firm_one_side} returns a matching, it must be the firm-optimal strongly stable matching.
\end{lemma}

\begin{proof}[Proof of \Cref{lem:n_0_firm_bp5}]
If not, then some firm $f$ has a better partner in the firm-optimal matching. But then, that firm must have proposed to that partner and been rejected during the algorithm, which leads to a contradiction.
\end{proof}

Note that all the above proofs hold for any compound instance $X$ that has partial order preferences on one side and complete order preferences on the other. Thus, the setting naturally generalizes beyond just two instances, as observed in \Cref{thm:n_1_multiple}. Also note that both algorithms behave exactly like the original Deferred Acceptance algorithm when $X$ has complete order preferences on both sides; that is, when instances $A$ and $B$ are the same.

\end{document}